\documentclass[11pt, onecolumn]{IEEEtran}

\usepackage{cite}
\usepackage{amssymb}
\usepackage{amsmath,bm}
\usepackage{amssymb}
\usepackage{epsfig}
\usepackage{epsf}
\usepackage{subfigure}
\usepackage{graphicx}
\usepackage{algorithm}
\usepackage{algorithmic}
\usepackage{url}
\usepackage{authblk}
\usepackage{amsthm}

\usepackage[bottom]{footmisc}

\date{}
\def\BibTeX{{\rm B\kern-.05em{\sc i\kern-.025em b}\kern-.08em
    T\kern-.1667em\lower.7ex\hbox{E}\kern-.125emX}}

\newtheorem{claim}{Claim}

\newtheorem{theorem}{Theorem}
\newtheorem{definition}{Definition}

\newtheorem{lemma}{Lemma}
\newtheorem{proposition}{Proposition}
\newtheorem{remark}{Remark}

\allowdisplaybreaks

\newcommand{\Text}[1]{\text{\textnormal{#1}}}

\newenvironment{lemmarep}[1]{\noindent {\bf #1.}\begin{it}}{\end{it}}

\newenvironment{claimrep}[1]{\noindent {\bf #1.}\begin{it}}{\end{it}}

\ifCLASSINFOpdf
\else
\fi

\begin{document}

\title{%Interference Decomposition Bounds and the 
MISO Broadcast Channel with Hybrid CSIT: Beyond Two Users \footnote{This work has been presented in part in \cite{OursICC2015, OursISIT2015}.  The research of S. Lashgari and A. S. Avestimehr  is supported  by NSF Grants CAREER 0953117, CCF-1161720, NETS-1161904, and ONR award N000141310094; and the research of R. Tandon is supported by the NSF Grant CCF 14-22090.
}}

\author[1]{Sina Lashgari}
\author[2]{Ravi Tandon}
\author[3]{Salman Avestimehr}
\affil[1]{School of Electrical and Computer Engineering,
Cornell University, Ithaca, NY}
\affil[2]{Discovery Analytics Center and Department of Computer Science, Virginia Tech, 
Blacksburg, VA}
\affil[3]{Department of Electrical Engineering,
University of Southern California,
 Los Angeles, CA}

%\author[1]{Sina Lashgari}
%\author[2]{Ravi Tandon}
%\author[3]{Salman Avestimehr}
%\affil[1]{Department of Electrical and Computer Engineering, Cornell University, Ithaca NY\\ SL2232@cornell.edu}
%\affil[2]{Discovery Analytics Center \& Department of Computer Science, Virginia Tech, Blacksburg VA \\ tandonr@vt.edu}
%\affil[3]{Department of Electrical Engineering, University of Southern California, L.A. CA\\ avestimehr@ee.usc.edu}
\vspace{-10pt}

%% Create the title:
\maketitle

\begin{abstract}

We study the impact of heterogeneity of channel-state-information available at the transmitters (CSIT) on the capacity of broadcast channels with a multiple-antenna transmitter and $k$ single-antenna receivers (MISO BC).
In particular, we consider the $k$-user MISO BC, where the CSIT with respect to each receiver can be either instantaneous/perfect, delayed, or not available; and we study the impact of  this heterogeneity of CSIT  on the degrees-of-freedom (DoF) of such network.
We first focus on the $3$-user MISO BC; and we  completely characterize the DoF region for all possible heterogeneous CSIT configurations, assuming  linear encoding strategies at the transmitters.
The result shows that the state-of-the-art achievable schemes in the literature are indeed sum-DoF optimal, when restricted to linear encoding schemes.
To prove the result, we develop  a novel bound,  called \emph{Interference Decomposition Bound}, which provides a lower bound on the interference dimension at a receiver which supplies delayed CSIT based on the average dimension of constituents of that interference, thereby decomposing the interference into its individual components.
Furthermore, we extend our outer bound on the DoF region to the general $k$-user MISO BC, and demonstrate that it leads to an approximate characterization of  linear sum-DoF to within an additive gap of $0.5$ for a broad range of CSIT configurations.
Moreover, for the special case where only one receiver supplies delayed CSIT, we completely characterize the linear sum-DoF.
\end{abstract}

\section{Introduction}

Channel state information at the transmitters (CSIT) plays a crucial role in the design and operation of multi-user wireless networks.
Timely and accurate knowledge of the channels can potentially help the transmitters mitigate the interference that they cause at the unintended receivers, therefore enabling them to increase the communication rate to their intended receivers.
The common procedure for obtaining channel state information (CSI) is to send training symbols (or pilots) at the transmitters, and then estimate the channels at the receivers and feed the estimates back to the transmitters. As a result of this feedback mechanism, CSI may not always be perfect and instantaneous. For instance, CSIT may be outdated due to the fast fading nature of the channels or slow feedback mechanism, it can be noisy (imperfect), or not available at all. 
Therefore, one can expect that in a large network there would be various types of CSI available at the transmitters with respect to different receivers.
This results in communication scenarios with \emph{heterogeneous}  or \emph{hybrid}  CSIT.

As a result, there has been a growing interest in studying the impact of CSIT on the capacity of wireless networks, especially the broadcast channel.
In particular, it was shown in \cite{MAT} that even when the transmitter(s) only have access to delayed CSIT, there is significant potential for degrees-of-freedom (DoF) gain.
They studied the problem of $k$-user multiple-input single-output  broadcast channel (MISO BC) with delayed CSIT, and showed that for such network $\text{DoF}= \frac{k}{1+ \frac{1}{2} +  \ldots + \frac{1}{k} }$.
This work was followed by several other works which studied other network configurations under the assumption of delayed CSIT, including interference channel \cite{abdoli,varanasi,VMABinaryIC,OursAllerton2014}, X-channel \cite{xchannel, Ours},  multi-hop networks \cite{AAMultiHopJ}, and other variations of delayed  CSIT \cite{Yang}.

Most of these prior works assume that the entire network state information is obtained with delay.
However, in a large network, one can expect various types of CSIT available at the transmitters with respect to different receivers.
As a result, there have also been several works on studying the impact of heterogeneous (or hybrid) CSIT on the capacity of  wireless networks,
where the CSIT with respect to each receiver can now be either instantaneous/perfect ($P$), delayed ($D$), or not available ($N$)
 \cite{retrospective, RaviMaddah,MIMOICHybrid, synergistics, MTU, Mohanty,  3UserHybrid}.
However, studying networks under the assumption of heterogeneous CSIT becomes quite challenging, to the extent that only the  DoF for  $2$-user MISO BC is characterized \cite{GholamiJafar, synergistics}; and beyond the 2-user network configuration even the DoF is unknown and the problem remains widely open.

To make progress on the MISO BC beyond 2 users, in this paper we focus on characterizing the degrees of freedom when restricted to linear schemes (also called LDoF). Our motivation to focus on the linear degrees of freedom is based on recent progress made in \cite{Ours,OursAllerton}, where the concept of LDoF was introduced;
and it was shown that for 2-user X-channel with delayed CSIT, LDoF can be characterized, while the information theoretic DoF remains open.
Linear schemes are also of significant practical interests due to their low complexity; and in fact, the majority of DoF-optimal schemes developed so far for networks with delayed CSIT are linear.

We consider the problem of MISO BC with hybrid CSIT, with a multiple-antenna transmitter and $k$ single-antenna receivers ($k > 2$), and study its linear degrees of freedom.
 The channels are time-varying, and the CSIT provided by each receiver is either instantaneous ($P$), delayed ($D$), or none ($N$). We first study the case of $k = 3$, and fully characterize the LDoF for all $3^3$ possible hybrid CSIT configurations. The result is obtained by developing a general outer bound on the LDoF region, and a matching achievable scheme for each of the CSIT configurations.

The outer bound, which is the main contribution of this paper,
is based on three main ingredients.
The first ingredient  is a novel lemma,  called \emph{Interference Decomposition Bound}. 
It  essentially lower bounds  the interference dimension at a receiver with delayed CSIT by the 
average dimension of its constituents, thereby decomposing the interference into its individual components.
As a result of Interference Decomposition Bound, we can then focus on analyzing the dimension of constituents of interference at receivers which supply delayed CSIT, in order to derive an upper bound on LDoF.
Proof of Interference Decomposition Bound is based on temporal analysis of dimensions of transmit signals at different receivers, leading to necessary conditions on the increments of such dimensions using the delayed CSIT constraint.

The second main ingredient of the converse proof  is  \emph{MIMO  Rank Ratio Inequality for Broadcast Channel}, which provides a lower bound on the dimension of  interference components at receivers supplying delayed CSIT.
In particular, the bound states that if the transmitter employs linear precoding schemes, the dimension of each interference component  at a single-antenna receiver supplying delayed CSIT is at least half of the dimension of the corresponding signal at any other single-antenna receiver. 
This inequality can be viewed as a variation of the Rank Ratio Inequality proved in \cite{Ours}, which shows that if two \emph{distributed}  single-antenna transmitters employ linear strategies, the  dimensions of received linear subspace at  a single-antenna receiver supplying delayed CSIT is at least  $\frac{2}{3}$ of the dimension of the same signal at any other single-antenna receiver.
Note that the key difference between the two lemmas lies in the assumption of \emph{distributed} antennas  in Rank Ratio Inequality, which changes the proof techniques required to establish the inequality.

Finally, the third ingredient of the converse, called  \emph{Least Alignment Lemma},  provides a lower bound on the dimension of  interference components at receivers supplying no CSIT.
In particular, the bound states that once the transmitter(s) in a network has no CSIT of a certain receiver, the least amount of alignment will occur at that receiver, meaning that transmit signals will occupy the maximal signal dimensions at that receiver.
As a result, Least Alignment Lemma implies that the dimension of interference caused  at a receiver which supplies no CSIT by the message intended for another receiver is at least equal to the dimension of the message itself.
Using the three main ingredients we develop a converse proof which characterizes the linear DoF region for all  $3^3$ possible hybrid CSIT configurations of the 3-user MISO BC.

We next extend the key proof ingredients of the converse for 3-user MISO BC to the general $k$-user setting.
 In particular, we extend the Interference Decomposition Bound to 
the $k$-user setting to provide lower bound on the dimension of any interfering signal in  an arbitrary receiver supplying delayed CSIT.
In addition, we present a generalized version of MIMO Rank Ratio Inequality for BC, which provides a lower bound on the dimension of \emph{joint} received signals at any arbitrary subset of receivers supplying delayed CSIT.
Additionally, we extend the Least Alignment Lemma  and show that under linear schemes, for arbitrary transmit signals the dimension of received signal at a receiver supplying no CSIT cannot be less than any other receiver.

By extending the converse tools to the general $k$-user  setting, we provide a  new outer bound on the linear DoF region of the general $k$-user MISO BC with arbitrary hybrid CSIT configuration.
We demonstrate that our new outer bound leads to  an approximate  linear sum-DoF characterization to within an additive gap of $0.5$ for networks with more number of receivers supplying instantaneous CSIT than   delayed CSIT; and the approximation gap decays exponentially with the increase in number of receivers supplying instantaneous CSIT. 
Furthermore, by using the outer bound and providing a new multi-phase achievable scheme, we present the exact characterization of linear sum-DoF  for  networks in which only one receiver supplies delayed CSIT.

{\bf Notation.} We use small letters (e.g. $ x$) for scalars, arrowed letters (e.g. $\vec x$) for vectors,  capital letters (e.g. $X$) for matrices, and  calligraphic font (e.g. $\mathcal X$) for sets.
We also use bold letters (e.g. $\bold x$) for random entities, and non-bold letters for deterministic values (e.g., realizations of random variables).

\section{System Model} \label{model}

We consider the Gaussian $k$-user multiple-input single-output broadcast channel (MISO BC) as depicted in Figure ~\ref{xch}. It consists of
 a transmitter with $m$ antennas, and $k$ single-antenna receivers, $\Text{Rx}_1,\Text{Rx}_2,\ldots, \Text{Rx}_k$, where  $m\geq k$.
The transmitter has a separate message for each of the receivers.

\begin{figure}[h!]
\centering
\includegraphics[scale=.26, trim= 10mm 10mm 10mm 10mm]{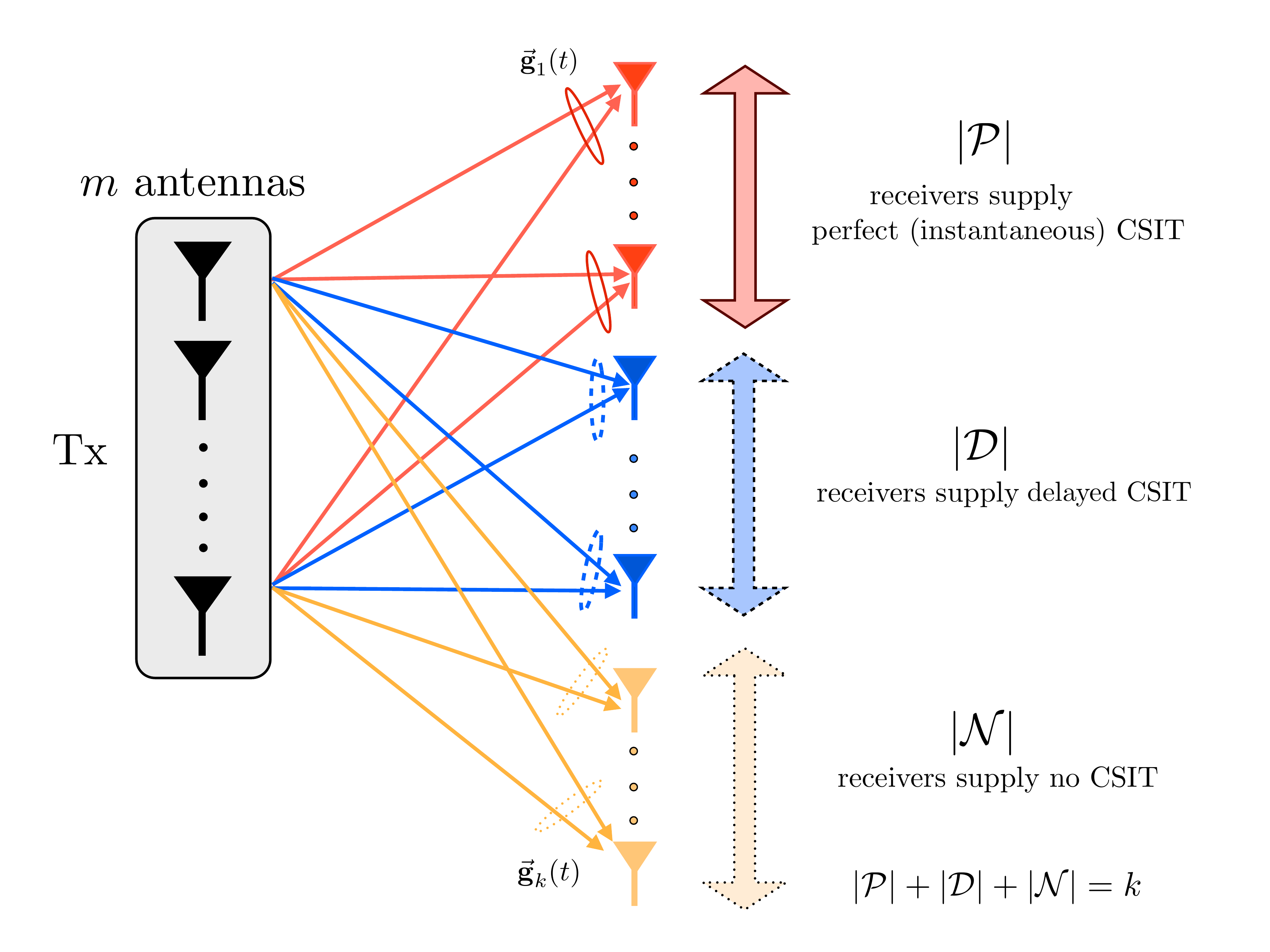}\\
\caption{Network configuration for $k$-user MISO BC.
}\label{xch}
\end{figure}

Consider communication over  $n$ time slots.
The received signal at $\text{Rx}_{j}$ ($j\in \{1,2,\ldots, k\}$) at time $t$ is given by
\begin{equation}
\bold{y}_{j}(t)=\vec{\bold{g}_{j}}(t) \vec{\bold{x}}(t)+\bold{z}_{j}(t),
\end{equation}
where
$\vec{\bold{x}}(t)  \in \mathbb C^{m}$ is the transmit signal vector at time $t$;
$\vec{\bold{g}_{j}}(t) \in \mathbb C^{1\times m}$ denotes the channel coefficients of the channel from 
$\text{Tx}$ to $\text{Rx}_{j}$;
and $\bold{z}_{j}(t)$ denotes the additive white noise which is distributed as $  \mathcal C\mathcal N(0,1)$.
The elements of the channel coefficients vector  $\vec{\bold{g}_{j}}(t)$ are i.i.d, drawn from a continuous distribution and also i.i.d across time and users. $\bm{ \mathcal {G}}(t)$ denotes the set of all $k$ channel vectors at time $t$.
In addition, we denote by $\bm{\mathcal {G}}^n$ the set of all channel coefficients from time 1 to $n$, i.e., 
\begin{equation} \label{chancoeff}
\bm{\mathcal {G}}^n=\{\vec{\bold{g}_{j}}(t): j =1,2,\ldots, k ,\quad t=1,\ldots,n\}.
\end{equation}
We denote the vector of transmit signals  in a block of length $n$ by $\vec {\bf x}^n$, where $\vec {\bf x}^n$ is the result of concatenation of transmit signal vectors $\vec{\bold{x}}(1), \ldots, \vec{\bold{x}}(n)$.
We assume $\text{Tx}$ obeys an average power constraint, $\frac{1}{n}E\{|| \vec {\bf x}^n||^2 \}\leq P_0$.

We focus on scenarios in which channel state information available at the transmitter (CSIT) with respect to different receivers can be instantaneous ($P$), delayed ($D$), or none ($N$). We refer to these scenarios as \emph{fixed hybrid scenarios}, or \emph{hybrid} in short.
In particular, CSIT with respect to
$\Text{Rx}_j$, $j=1,2,\ldots, k$, is denoted by $I_j\in \{ P,N,D\}$, as defined in \cite{synergistics}. 
In this notation,  $I_j=P$ indicates that $\Text{Tx}$ has access to instantaneous CSIT with respect to $\Text{Rx}_j$; i.e., at time $t$, $\Text{Tx}$ has access to $\{\vec{\bold{g}_{j}}(1), \ldots, \vec{\bold{g}_{j}}(t)\}$.
Similarly, $I_j=D$ indicates delayed CSIT with respect to $\Text{Rx}_j$; i.e., at time $t$, $\Text{Tx}$ has access to $\{\vec{\bold{g}_{j}}(1), \ldots, \vec{\bold{g}_{j}}(t-1)\}$.
Finally, $I_j=N$ indicates no CSIT, which means the channel to $\Text{Rx}_j$ is not known to the $\Text{Tx}$ at all.
We assume that the type of CSIT for each receiver is fixed and does not alternate over time (nevertheless, channels are time-varying). Therefore, there are $3^k$ different  fixed hybrid scenarios.
As an example, we use $PDD$ to denote the 3-user MISO BC where the first receiver provides instantaneous CSIT, while the other two provide delayed CSIT.

\begin{definition}\label{CSITsets}
We  denote the set of indices of users in states $P,D,N$ by $\mathcal P, \mathcal D, \mathcal N$, respectively.
In addition, 
for an ordered set  $\mathcal S$ we denote by $\pi_{\mathcal S}$ the ordered set obtained by a permutation of the elements of $\mathcal S$, where we denote the elements of the new ordered  set by $\pi_{\mathcal S}(1), \pi_{\mathcal S}(2), \ldots , \pi_{\mathcal S}(|\mathcal S|)$.
\end{definition}

Note that according to Definition \ref{CSITsets}, $\mathcal P \cup \mathcal D\cup \mathcal N = \{1,2,\ldots,k\}$ and $\mathcal P \cap \mathcal D = \mathcal D \cap \mathcal N =\mathcal P \cap \mathcal N = \emptyset $.
Based on the above description of channel state information, the channel outcomes available to $\Text{Tx}$ at time $t$ are denoted by the following set:
\begin{equation}
\bm{\mathcal {\tilde G}}^t = \{ \bold{G}_i^t; i \in \mathcal P\} \cup  \{ \bold{G}_j^{t-1};  j \in \mathcal D \} .
\end{equation}

We restrict ourselves to linear coding strategies as defined in ~\cite{Ours}, in which degrees-of-freedom (DoF) represents the dimension of the linear subspace of transmitted signals.
More specifically, consider a communication scheme with block length $n$, in which the $\text{Tx}$ wishes to deliver a vector $\vec {\bf x}_{j} \in \mathbb C^{m_{j}(n)}$ of $m_{j}(n) \in \mathbb{N}$ information symbols to $\text{Rx}_j$ ($j \in \{1,2.\ldots,k\}$).
Each information symbol is a random variable with variance $P_0$.
These information symbols are then modulated with precoding matrices ${\bf V}_{j}(t)\in \mathbb C^{m\times  m_{j}(n)}$ at times $t=1,2,\ldots,n$. 
Note that the precoding matrix ${\bf V}_{j}(t)$ depends only upon the outcome of $\bm{\mathcal {\tilde G}}^t$ due to the hybrid CSIT constraint:

\begin{equation}  \label{coding}
{\bf V}_{j}(t)=f_{j,t}^{(n)} \left( \bm{\mathcal {\tilde G}}^t \right).
\end{equation}

Based on this linear precoding, $\text{Tx}$ will then send $ \vec{{\bf x}}(t)= \sum_{j=1}^{k}{\bf V}_{j}(t) \vec {\bf x}_{j}$ at time $t$. We can rewrite $\vec{{\bf x}}(t)$ as following.
\begin{equation}
 \vec{{\bf x}}(t)= [{\bf V}_{1}(t) \ldots{\bf V}_{k}(t) ]   [\vec {\bf x}_{1};\ldots ; \vec {\bf x}_{k}],
\end{equation}
where $[A;B]$ denotes the vertical concatenation of matrices $A$ and $B$ (i.e., $\left[\begin{array}{c}{A} \\ { B}\end{array}\right]).$

We denote by $\bold{V}_{j}^n\in \mathbb C^{nm\times m_{j}(n)}$ the overall precoding matrix of $\text{Tx}$ for $\text{Rx}_j$, such that   the rows $1+(t-1)m,\ldots ,tm$ of $\bold{V}_{j}^n$ constitute ${\bf V}_{j}(t)$.
In addition, we denote the precoding function used by $\text{Tx}$ by $f^{(n)}=  \{ f_{j,t}^{(n)}   \}_{\substack{t=1,\ldots, n\\ j=1,\ldots, k}}$.

Based on the above setting, the received signal at $\text{Rx}_{j}$ ($j\in \{1,2,\ldots, k\}$) after the $n$ time steps of the communication will be
\begin{equation}
\vec {\bf y}_{j}^n=\bold{G}_{j}^n [{\bf V}_{1}^n \ldots{\bf V}_{k}^n ]   [\vec {\bf x}_{1};\ldots ; \vec {\bf x}_{k}] + \vec {\bf z}_{j}^n,
\end{equation}
 where $\bold{G}_j^n\in \mathbb C^{n\times nm}$ is the block diagonal channel coefficients matrix where the channel coefficients  of timeslot $t$ (i.e. $\vec{{\bf g}_j}(t)$) are in the row $t$, and in the columns $1+(t-1)m, \ldots, tm$ of $\bold{G}_j^n$, and the rest of the elements of $\bold{G}_j^n$ are zero.\footnote{For $j\in \{1,\ldots, k\}$, we define $\bold{G}_{j}^0 [{\bf V}_{1}^0 \ldots{\bf V}_{k}^0 ]  \triangleq \vec 0 $; therefore, for instance,  $\Text{rank}\left[ \bold{G}_{j}^0 [{\bf V}_{1}^0 \ldots{\bf V}_{k}^0 ] \right]=0$.}

Now, consider the decoding of $\vec {\bf x}_{j}$ at $\text{Rx}_j$ (i.e., decoding the $m_{j}(n)$ information symbols  for $\text{Rx}_j$). The corresponding interference subspace at $\text{Rx}_j$ will be
 \begin{align*}
   \bm{\mathcal I}_{j}=\text{colspan} \left( \bold{G}_{j}^n [\cup_{i\ne j}{\bf V}_{i}^n  ] \right),
\end{align*} 
where $[\cup_{i\neq j}{\bf V}_i^n  ]$ 
is the matrix formed by row concatenation of matrices ${\bf V}_{i}^n $ for $i\neq j$, and $\text{colspan}(.)$ of a matrix corresponds to the sub-space that is spanned by its columns. 
 Let $\bm{\mathcal I}_{j}^\perp \subseteq \mathbb{C}^n$ denote the  orthogonal subspace of $\bm{\mathcal I}_{j}$. Then, in the regime of asymptotically high transmit powers (i.e., ignoring the noise), the decodability of information symbols at $\text{Rx}_j$ corresponds to the constraint that  the image of $\text{colspan}(\bold{G}_{j}^n \bold{V}_{j}^n)$ on $\bm{\mathcal I}_{j}^\perp$ has dimension $m_{j}(n)$:
\begin{equation}
\Text{dim} \left( \text{Proj}_{\bm{\mathcal I}_{j}^\perp} \text{colspan} \left( \bold{G}_{j}^n \bold{V}_{j}^n \right) \right) =\Text{dim} \left( \text{colspan} \left( \bold{G}_{j}^n\bold{V}_{j}^n \right) \right) = m_{j}(n),
\end{equation}
which can be shown by simple linear algebra to be equivalent to the following:
\begin{align}
 \text{rank} [\bold{G}_{j}^n [\cup_{i=1}^k{\bf V}_{i}^n  ]] -  \text{rank} [\bold{G}_{j}^n [\cup_{i\ne j}{\bf V}_{i}^n  ]]  &=\text{rank} [\bold{G}_{j}^n {\bf V}_{j}^n  ] = m_j(n). \label{decode0}
\end{align}

Based on this setting, we now define the linear degrees-of-freedom of the $k$-user MISO broadcast channel with hybrid CSIT.
\begin{definition} \label{DoFdef}
$k$-tuple $(d_{1},d_{2},\ldots ,d_{k})$ degrees-of-freedom are linearly achievable if there exists a sequence
$\{  f^{(n)} \}_{n=1}^{\infty}$ such that for each $n$ and the corresponding choice of $(m_{1}(n),m_{2}(n),\ldots, m_{k}(n))$, $(\bold{V}_{1}^n,\bold{V}_{2}^n,\ldots, \bold{V}_{k}^n)$ satisfy the decodability condition of (\ref{decode0}) with probability 1; i.e., for all $ j \in \{1,\ldots, k\} $,
\begin{align}
 \Text{rank} [\bold{G}_{j}^n [\cup_{i=1}^k{\bf V}_{i}^n  ]] -  \Text{rank} [\bold{G}_{j}^n [\cup_{i\ne j}{\bf V}_{i}^n  ]]  &\stackrel{a.s.}{=} \Text{rank} [\bold{G}_{j}^n {\bf V}_{j}^n  ] \stackrel{a.s.}{=} m_j(n), \label{decode}
\end{align}
and
\begin{equation}
d_{j}=\lim_{n\to\infty} \frac{m_{j}(n)}{n}. \label{DoFcond}
\end{equation}

We also define the linear degrees-of-freedom region $\Text{LDoF}_{\Text{region}}$ as the closure of the set of all linearly achievable $k$-tuples $(d_{1},d_{2},\ldots ,d_{k})$.
Furthermore, the  linear sum-degrees-of-freedom ($\Text{LDoF}_{\Text{sum}}$) is defined as follows:
\begin{equation}
\Text{LDoF}_{\Text{sum}} \triangleq \max\quad \sum_{j=1}^{k}d_{j},\qquad \textrm{s.t. }\quad (d_{1},d_{2},\ldots ,d_{k})\in\Text{LDoF}_{\Text{region}}.
\end{equation}
\end{definition}

In what follows we first focus on the case of $k=3$,  and completely characterize the $\Text{LDoF}_{\Text{region}}$ for 3-user MISO BC with hybrid CSIT. We then extend our bounds and present new outer bounds on the $\Text{LDoF}_{\Text{region}}$ of the general $k$-user MISO BC with hybrid CSIT.

\section{3-user MISO Broadcast Channel with Hybrid CSIT} 
In this section we focus on 3-user MISO broadcast channel with hybrid CSIT.
In particular, we first state the complete characterization of $\Text{LDoF}_{\Text{region}}$ for all  hybrid CSIT configurations; and then, we present the proof  based on 3 key lemmas.

\begin{theorem} \label{mainthm3}
Given a hybrid CSIT configuration, i.e., a partition of $3$ users into disjoint sets $\mathcal P, \mathcal D,$ and $\mathcal N$ as defined in Definition \ref{CSITsets},
the 
$\Text{LDoF}_{\Text{region}}$  is characterized as follows:
\begin{align}
\Text{LDoF}_{\Text{region}}= \bigg \{ \quad (d_1,d_2,d_3) \quad  | \quad & 0\leq d_1,d_2,d_3\leq 1, \nonumber\\
& \forall i\in   \mathcal D ,  \forall \pi_{\mathcal P \cup \mathcal D \setminus i}, \quad    \sum_{j =1}^{|\mathcal P|+ |\mathcal D|-1} \frac{d_{\pi_{\mathcal P \cup \mathcal D \setminus i} (j)}}{2^j} + d_{i} +\sum_{j\in \mathcal N} d_j \leq 1,  \nonumber\\
&\forall \pi_{\mathcal D}, \quad \sum_{j\in \mathcal P}\frac{d_j}{3} + \sum_{j=1}^{|\mathcal D|}  \frac{d_{\pi_{\mathcal D}(j)}}{j} +\sum_{j\in \mathcal N} d_j \leq 1,   \nonumber\\
& \forall i\in  \mathcal P \cup \mathcal D , \quad     d_{i} +\sum_{j\in \mathcal N} d_j \leq 1 \quad
\bigg \}. \label{Dregion}
\end{align}
The $\Text{LDoF}_{\Text{region}}$ and the corresponding
 $\Text{LDoF}_{\Text{sum}}$ for different CSIT configurations are summarized in Table \ref{tab:template}.
\end{theorem}

Note that although there are $3^3$ different CSIT configurations for 3-user MISO BC, many of them are permutations of one another, e.g. $PPD,PDP,DPP$. As a result, there are only $10$ distinct CSIT configurations which are presented in Table  \ref{tab:template}.

\renewcommand{\arraystretch}{2}

\begin{table*}[t]
\centering
\begin{tabular}{|c|c|c|}
\hline
$\text{CSIT States}$ &$\text{Linear Degrees of Freedom Region ($\Text{LDoF}_{\Text{region}}$)}$ & $\text{LDoF}_{\text{sum}}$ \\
\hline

\hline
PPP & {$\!\begin{aligned}
\Text{LDoF}_{\Text{region}}= \bigg\{ (d_1,d_2,d_3) \quad  | \quad      & 0\leq d_1,d_2,d_3\leq 1 \bigg\} \nonumber \end{aligned}$} & 3 \\
\hline 
PPD  & {$\begin{aligned}
\Text{LDoF}_{\Text{region}}= \bigg\{ (d_1,d_2,d_3) \quad  | \quad       & 0\leq d_1,d_2,d_3\leq 1, \quad   \frac{d_1}{2} + \frac{d_2}{4} + d_3 \leq 1, \quad  \frac{d_1}{4} + \frac{d_2}{2} + d_3 \leq 1 \bigg\} \nonumber \end{aligned}$} & $\frac{9}{4}$ \\[1mm]
\hline
PPN  & {$\!\begin{aligned}
\Text{LDoF}_{\Text{region}}=\bigg \{ (d_1,d_2,d_3) \quad  | \quad       & 0\leq d_1,d_2,d_3\leq 1 ,\quad  & d_1 +  d_3 \leq 1, \quad d_2+ d_3 \leq 1\bigg \} \nonumber \end{aligned}$} & $2$ \\
\hline 
PDD & {$\!\begin{aligned}
\Text{LDoF}_{\Text{region}}= \bigg\{  (d_1,d_2,d_3) \quad  | \quad       & 0\leq d_1,d_2,d_3\leq 1, \nonumber\\  & \frac{d_1}{2} + \frac{d_2}{4} + d_3 \leq 1,\qquad  \frac{d_1}{2} + d_2 + \frac{d_3}{4} \leq 1,\nonumber\\
& \frac{d_1}{3} + \frac{d_2}{2} + d_3 \leq 1, \qquad \frac{d_1}{3} + d_2 + \frac{d_3}{2} \leq 1 \bigg\} \nonumber \end{aligned}$} & $\frac{9}{5}$ \\
\hline
PDN & {$\begin{aligned}
\Text{LDoF}_{\Text{region}}= \bigg\{ (d_1,d_2,d_3) \quad  | \quad       & 0\leq d_1,d_2,d_3\leq 1,\quad & \frac{d_1}{2} + d_2 + d_3 \leq 1 , \quad  d_1 + d_3\leq 1\nonumber\bigg\} \end{aligned}$} & $\frac{3}{2}$ \\
\hline
DDD & {$\!\begin{aligned}
\Text{LDoF}_{\Text{region}}= \bigg\{ (d_1,d_2,d_3) \quad  | \quad       & 0\leq d_1,d_2,d_3\leq 1 ,\nonumber\\
 &\frac{d_1}{3} + \frac{d_2}{2} + d_3 \leq 1,\quad  \frac{d_1}{3} + d_2 + \frac{d_3}{2} \leq 1, \quad  \frac{d_1}{2} + \frac{d_2}{3} + d_3 \leq 1, \nonumber\\
&\frac{d_1}{2} + d_2 + \frac{d_3}{3} \leq 1,\quad  {d_1} + \frac{d_2}{2} + \frac{d_3}{3} \leq 1, \quad  d_1 + \frac{d_2}{3} + \frac{d_3}{2} \leq 1 \bigg\} \nonumber \end{aligned}$} & $\frac{18}{11}$ \\
\hline
DDN & {$\!\begin{aligned}
\Text{LDoF}_{\Text{region}}= \bigg\{ (d_1,d_2,d_3) \quad  | \quad       & 0\leq d_1,d_2,d_3\leq 1,\quad
& \frac{d_1}{2} + d_2 + d_3 \leq 1, \quad d_1 + \frac{d_2}{2} + d_3 \leq 1 \bigg \} \nonumber \end{aligned}$} & $\frac{4}{3}$ \\
\hline
PNN, DNN, NNN & {$\!\begin{aligned}
\Text{LDoF}_{\Text{region}}=\bigg \{ (d_1,d_2,d_3) \quad  | \quad       & 0\leq d_1,d_2,d_3\leq 1, \quad  d_1 + d_2 + d_3 \leq 1 \nonumber \bigg\} \end{aligned}$} & $1$ \\

\hline

\end{tabular}
\vspace{3mm}
\caption{  $\Text{LDoF}_{\Text{region}}$ and $\Text{LDoF}_{\Text{sum}}$ for all  possible configurations of hybrid CSIT for 3-user MISO BC. 
}
\label{tab:template}
\vspace{-15pt}
\end{table*}

\begin{remark}
The bound in Theorem \ref{mainthm3} strictly improves the state-of-the-art bounds, and also  leads to complete characterization of $\Text{LDoF}_{\Text{region}}$  for $k=3$. 
For instance, for  $PDD$ (i.e. $\text{Rx}_1$ supplying instantaneous CSIT, while $\text{Rx}_2,\text{Rx}_3$ supply delayed CSIT) the prior results suggest that $\Text{LDoF}_{\Text{sum}} \leq \frac{17}{9}$  \cite{Mohanty, RassouliClerckx}, while by Theorem \ref{mainthm3},  $\Text{LDoF}_{\Text{sum}}$ is indeed equal to $\frac{9}{5}$.
Similarly, for the case of $PPD$, the prior results \cite{Mohanty, RassouliClerckx} imply that
 $\Text{LDoF}_{\Text{sum}} \leq \frac{7}{3}$, while by Theorem \ref{mainthm3},   $\Text{LDoF}_{\Text{sum}}=\frac{9}{4}$
\end{remark}
\begin{remark}
Theorem \ref{mainthm3}  implies that the state-of-the-art achievable schemes presented in \cite{3UserHybrid} for $PPD$ and $PDD$ are both optimal from the perspective of $\Text{LDoF}_{\Text{sum}}$.
\end{remark}

\begin{remark}
It is worth noting that in any CSIT configuration which involves receivers with state N, the inequalities that constitute the LDoF region have coefficient 1 for the degrees-of-freedom of receivers with state N.
In other words, receivers that supply no CSIT do not contribute to the $\Text{LDoF}_{\Text{sum}}$, and unless all receivers have state N, removing the no CSIT receivers from the network will not decrease the $\Text{LDoF}_{\Text{sum}}$.
\end{remark}

In the remainder of this section we prove Theorem \ref{mainthm3}.
To this aim, we first present the converse proof in Section \ref{3converse}, and then discuss the achievability in Section \ref{PDDachsection}.

\subsection{Proof of Converse for 3-User MISO Broadcast Channel with Hybrid CSIT}\label{3converse}
We first provide the three  main ingredients that are key  in proving the converse for  3-user MISO broadcast channel with hybrid CSIT.
We then show how those main ingredients are used to prove the converse for two representative CSIT configurations (i.e. $PDD$ and $PDN$). 
The proof of converse for other CSIT configurations  can be found in Appendix \ref{convmainthm3}.
The first two ingredients of the converse proof deal with lower bounding received signal dimension at a receiver which supplies delayed CSIT, while the third ingredient captures the impact of no CSIT.% on the received signal dimension.

The first key ingredient  is  Interference Decomposition Bound, which essentially provides a lower bound on the interference dimension at a receiver supplying delayed CSIT, based on the constituents of that interference, as well as the received signal dimension at other receivers.
It is stated below; and its proof
is provided in Appendix \ref{IDBproof}.

\begin{lemma}\label{MainIneq}
{\bf (Interference Decomposition Bound)} 
 Consider $k=3$, and a fixed  linear coding strategy $ f^{(n)}$, with corresponding precoding matrices $\bold{V}_{1}^{n}, \bold{V}_{2}^{n},\bold{V}_{3}^{n} $ as defined in (\ref{coding}). If $I_3= D$ (i.e., if $\text{Rx}_3$ supplies delayed CSIT),
\begin{align} 
 & \frac{ \Text{rank}[ \bold{G}_{1}^n [\bold{V}_{1}^{n} \quad \bold{V}_{2}^{n}] ]  - \Text{rank}[ \bold{G}_{1}^n  \bold{V}_{2}^{n} ]  + \Text{rank}[ \bold{G}_{3}^n \bold{V}_{2}^{n} ]   }{2}   \stackrel{a.s.}{\leq}     \Text{rank}[ \bold{G}_{3}^n [\bold{V}_{1}^{n} \quad \bold{V}_{2}^{n}]].\label{IDB}
\end{align}
\end{lemma}

\begin{remark}
The R.H.S. of Interference Decomposition Bound represents the dimension of interference caused at $\text{Rx}_3$, which supplies delayed CSIT, by the messages intended for $\text{Rx}_{1}, \text{Rx}_{2}$. 
On the other hand, the third term on the L.H.S.  (i.e. $ \Text{rank}[ \bold{G}_{3}^n \bold{V}_{2}^{n} ] $) is the dimension of the remaining interference at $\text{Rx}_3$ after removing the contribution of the message of $\text{Rx}_{1}$;
and the first two terms (i.e. $ \Text{rank}[ \bold{G}_{1}^n [\bold{V}_{1}^{n} \quad \bold{V}_{2}^{n}] ]  - \Text{rank}[ \bold{G}_{1}^n  \bold{V}_{2}^{n} ] $) can be shown by (\ref{decode}) and sub-modularity of rank (stated in Lemma \ref{ranksubmod}) to equal $ \Text{rank}[ \bold{G}_{1}^n \bold{V}_{1}^{n}  ] $, which is the dimension of message of $\text{Rx}_{1}$.
Hence, Interference Decomposition Bound provides an inequality which connects  the  dimension of interference at a receiver  to the average dimension of its constituents.
Note that  statement of Lemma \ref{MainIneq} does not assume any specific CSIT with respect to any receiver except $\text{Rx}_{3}$.
\end{remark}

The second main ingredient, called MIMO Rank Ratio Inequality for BC,   provides a lower bound on the dimension of received signal at a receiver supplying delayed CSIT. It  is  stated below; and  its proof
is provided in Appendix \ref{AppMIMORRI}.

\begin{lemma} \label{MIMORRI}
{\bf (MIMO Rank Ratio Inequality for BC)}
Consider $k=3$, and a linear coding strategy $f^{(n)} $, with corresponding $\bold{V}_{1}^{n},\bold{V}_{2}^{n} , \bold{V}_{3}^{n} $ as defined in (\ref{coding}).
If $I_3= D$ (i.e., if $\text{Rx}_3$ supplies delayed CSIT), then, for  each beamforming matrix $\bold{V}_{i}^{n}$, where  $i\in \{1,2,3\}$, and each $\ell\in \{1,2,3\}$, we have
\begin{align}
& \frac{ \Text{rank}[ [\bold{G}_{\ell}^n;  \bold{G}_{3}^n]\bold{V}_{i}^{n} ]  }{2}  \stackrel{a.s.}{\leq}\Text{rank}[ \bold{G}_{3}^n \bold{V}_{i}^{n} ],
\end{align}
where $[\bold{G}_{\ell}^n;  \bold{G}_{3}^n]$ denotes the column concatenation of  matrices $\bold{G}_{\ell}^n$ and  $ \bold{G}_{3}^n$.
\end{lemma}
\begin{remark}
Lemma \ref{MIMORRI} implies that for any transmit signal $\bold{V}_{i}^{n}$, the corresponding received signal dimension  at a receiver with delayed CSIT is at least half of the corresponding received signal dimension at any other receiver.
Note that  statement of Lemma \ref{MIMORRI} does not assume any specific CSIT with respect to any receiver except $\text{Rx}_{3}$.
\end{remark}

The third main ingredient of converse, Least Alignment Lemma,  demonstrates that when using linear schemes, once the transmitter has no CSIT with respect to a certain receiver,  the least amount of alignment will occur at that receiver, meaning that transmit signals will occupy the maximal signal dimensions at that receiver.
The lemma is stated below; and its proof
is provided in Appendix \ref{LALproof}.
\begin{lemma}\label{LAL}
{\bf (Least Alignment Lemma)}   
Consider $k=3$, and a linear coding strategy $f^{(n)} $, with corresponding $\bold{V}_{1}^{n},\bold{V}_{2}^{n} , \bold{V}_{3}^{n} $ as defined in (\ref{coding}).
For $\mathcal S\subseteq \{1,2,3\}$ let $\bold{V}^n \triangleq [\cup_{i\in \mathcal S} \bold{V}_i^n]$ denote the row concatenation of the precoding matrices $\bold{V}_i^n$, where $i\in \mathcal S$.
 If $I_3=N$ (i.e., if $\text{Rx}_3$ supplies no CSIT),
\begin{align*}
 \forall \ell\in \{1,2, 3\},\qquad    \Text{ rank} & \left[ \bold{G}_{\ell}^n \bold{V}^n  \right] \stackrel{a.s.}{\leq} \Text{ rank} \left[\bold{G}_{3}^n \bold{V}^n  \right].
\end{align*}
\end{lemma}
\begin{remark}
Note that the statement of Lemma \ref{LAL} does not assume any specific CSIT with respect to any receiver except $\text{Rx}_{3}$.
\end{remark}

\begin{remark}
Lemma \ref{LAL} can be seen as a variation of the corresponding result in the context of secrecy problems in \cite{OursISIT2014, OursGlobecom2014}. 
Moreover, as shown in \cite{GholamiJafar}, Least Alignment Lemma also holds for non-linear schemes; and  for this extension the reader is referred to \cite{GholamiJafar}.
\end{remark}

We now prove the converse for two representative CSIT configurations  $PDD$ and $PDN$, highlighting the applications of the above three lemmas.
Converse proofs for other CSIT configurations can be found in Appendix \ref{convmainthm3}.

\subsubsection{Proof of Converse for $PDD$} \label{PDDconv}

According to Table \ref{tab:template}, it is sufficient to show that  
$\frac{d_1}{2} + \frac{d_2}{4} + d_3 \leq 1$ and $\frac{d_1}{3} + \frac{d_2}{2} + d_3 \leq 1$; since the other two inequalities (i.e. $\frac{d_1}{2} + d_2 + \frac{d_3}{4} \leq 1$, and $\frac{d_1}{3} + d_2 + \frac{d_3}{2} \leq 1 $) can be proven similarly using symmetry.
Moreover, the bound $\frac{d_1}{3} + \frac{d_2}{2} + d_3 \leq 1$ follows directly from the existing state-of-the-art arguments used in \cite{ Mohanty, MAT}.
Henceforth, we focus on proving $\frac{d_1}{2} + \frac{d_2}{4} + d_3 \leq 1$.

Suppose $(d_1, d_2, d_3)$ degrees-of-freedom are linearly achievable. 
Hence, by Definition \ref{DoFdef}   there exists a sequence $\{  f^{(n)} \}_{n=1}^{\infty}$ such that for each $n$ and the corresponding choice of $(m_{1}(n),m_{2}(n), m_{3}(n))$, $(\bold{V}_{1}^n,\bold{V}_{2}^n, \bold{V}_{3}^n)$ satisfy the  conditions in (\ref{decode}) and (\ref{DoFcond}). 
Therefore, in order to prove $\frac{d_1}{2} + \frac{d_2}{4} + d_3 \leq 1$, it is sufficient to show that
\begin{equation} \label{PDDobj}
\frac{m_1(n)}{2} + \frac{m_2(n)}{4} + m_3(n) \stackrel{a.s.}{\leq} n.
\end{equation}

Note that since in the $PDD$ configuration receiver 3 supplies delayed CSIT, we can invoke Lemma \ref{MainIneq}, which states that:
\begin{eqnarray}
 2\times \Text{rank}[ \bold{G}_{3}^n [\bold{V}_1^n \hspace{6pt} \bold{V}_2^n]]
 & \stackrel{a.s.}{\geq} & \Text{rank}[ \bold{G}_{1}^n [\bold{V}_{1}^{n} \quad \bold{V}_{2}^{n}] ]  - \Text{rank}[ \bold{G}_{1}^n  \bold{V}_{2}^{n} ]  + \Text{rank}[ \bold{G}_{3}^n \bold{V}_{2}^{n} ]  \nonumber\\
& \substack{(\ref{decode}) \\ a.s.\\=}& \Text{rank}[ \bold{G}_{1}^n \bold{V}_1^n] + \Text{rank}[ \bold{G}_{3}^n \bold{V}_2^n].\label{conv1}
\end{eqnarray}
We now further bound each side of the above inequality.
We first upper bound the left-hand-side of the above inequality:
\begin{align}
 \Text{rank}[ \bold{G}_{3}^n [\bold{V}_1^n \quad \bold{V}_2^n]]\hspace{5pt}  \substack{\text{(\ref{decode})} \\ a.s. \\=}\hspace{5pt} & \Text{rank}[ \bold{G}_{3}^n [\bold{V}_1^n \quad \bold{V}_2^n \quad \bold{V}_3^n]] - m_3(n)   \stackrel{}{\leq}   n - m_3(n) . \label{conv2}
\end{align}
On the other hand, for the right-hand-side of (\ref{conv1}) we have
\begin{eqnarray}
  \Text{rank}[ \bold{G}_{1}^n \bold{V}_1^n] + \Text{rank}[ \bold{G}_{3}^n \bold{V}_2^n] \hspace{8pt} &\substack{\text{(\ref{decode})} \\ a.s. \\ =}&  m_1(n)+ \Text{rank}[ \bold{G}_{3}^n \bold{V}_2^n] \substack{\text{(Lemma \ref{MIMORRI})} \\ a.s. \\ \geq} \hspace{8pt} m_1(n)+ \frac{1}{2} \Text{rank}[[ \bold{G}_{2}^n; \bold{G}_{3}^n] \bold{V}_2^n ] \nonumber\\
&\substack{ \geq} &  m_1(n)+ \frac{1}{2} \Text{rank}[ \bold{G}_{2}^n \bold{V}_2^n] 
 \hspace{8pt} \substack{\text{(\ref{decode})} \\ a.s. \\ =} \hspace{8pt} m_1(n)+ \frac{1}{2} m_2(n).\label{conv3}
\end{eqnarray}
Hence, by considering (\ref{conv1})-(\ref{conv3}), we obtain
\begin{equation}\label{eq1}
 m_1(n) + \frac{1}{2} m_2(n) + 2m_3(n) \stackrel{a.s.}{\leq} 2n, 
\end{equation}
which proves (\ref{PDDobj}), and
therefore, completes the converse proof for $PDD$.

\begin{remark} \label{PDDtoPPD}
Note that in order to prove $\frac{d_1}{2} + \frac{d_2}{4} + d_3 \leq 1$ for $PDD$, we did not rely on  any specific CSIT assumption with respect to $\text{Rx}_2$. Therefore, the bound $\frac{d_1}{2} + \frac{d_2}{4} + d_3 \leq 1$ also holds for the case of $PPD$. 
Moreover, note that by symmetry one can conclude that $\frac{d_1}{4} + \frac{d_2}{2} + d_3 \leq 1$ also holds for  $PPD$.
Hence, since  according to Table \ref{tab:template}, $\frac{d_1}{2} + \frac{d_2}{4} + d_3 \leq 1$ and $\frac{d_1}{4} + \frac{d_2}{2} + d_3 \leq 1$ constitute the LDoF region for $PPD$, 
the above derivations suffice in proving the converse for the CSIT configuration $PPD$ as well.
\end{remark}

\subsubsection{Proof of Converse for $PDN$} \label{PDNconv}
According to Table \ref{tab:template}, it is sufficient to show that 
$\frac{d_1}{2} + d_2 + d_3 \leq 1$ and $d_1+ d_3 \leq 1$.
Suppose $(d_1, d_2, d_3)$ degrees-of-freedom are linearly achievable. 
Hence, by Definition \ref{DoFdef}   there exists a sequence $\{  f^{(n)} \}_{n=1}^{\infty}$ such that for each $n$ and the corresponding choice of $(m_{1}(n),m_{2}(n), m_{3}(n))$, $(\bold{V}_{1}^n,\bold{V}_{2}^n, \bold{V}_{3}^n)$ satisfy the conditions in (\ref{decode}) and (\ref{DoFcond}).
Therefore, in order to prove $\frac{d_1}{2} + d_2 + d_3 \leq 1$ and $d_1+ d_3 \leq 1$, it is sufficient to show that
\begin{equation} \label{PDNobj}
\frac{m_1(n)}{2} + m_2(n) + m_3(n) \stackrel{a.s.}{\leq} n,
\end{equation}
and
\begin{equation} \label{PDNobj2}
 m_1(n) + m_3(n) \stackrel{a.s.}{\leq} n.
\end{equation}

We have,
\begin{eqnarray}
 \frac{m_1(n)}{2} + m_2(n) + m_3(n)   & \substack{\text{(\ref{decode})} \\ a.s. \\ =}& \frac{\Text{rank}[ \bold{G}_{1}^n \bold{V}_1^n]}{2}  + m_2(n) + m_3(n) \nonumber\\
& \substack{\text{(\ref{decode})} \\ a.s. \\ =} &  \frac{\Text{rank}[ \bold{G}_{1}^n \bold{V}_1^n]}{2} + \Text{rank}[ \bold{G}_{2}^n [\bold{V}_1^n \quad \bold{V}_2^n  \quad \bold{V}_3^n]] -  \Text{rank}[ \bold{G}_{2}^n [\bold{V}_1^n \quad \bold{V}_3^n] ] + m_3(n) \nonumber\\
 & \substack{\text{(a)}  \\ \leq} &   \frac{\Text{rank}[ \bold{G}_{1}^n \bold{V}_1^n]}{2} + \Text{rank}[ \bold{G}_{2}^n [\bold{V}_1^n \quad \bold{V}_2^n]]  -  \Text{rank}[ \bold{G}_{2}^n \bold{V}_1^n ] + m_3(n) \nonumber\\
& \substack{\leq} & \frac{\Text{rank}[ [\bold{G}_{1}^n; \bold{G}_{2}^n  ]\bold{V}_1^n]}{2} + \Text{rank}[ \bold{G}_{2}^n [\bold{V}_1^n \quad \bold{V}_2^n]]  -  \Text{rank}[ \bold{G}_{2}^n \bold{V}_1^n ] + m_3(n) \nonumber\\
 & \substack{ (b) \\ a.s.  \\ \leq}  &   \Text{rank}[ \bold{G}_{2}^n [\bold{V}_1^n \quad \bold{V}_2^n]]    + m_3(n) \nonumber\\
& \substack{\text{(\ref{decode})} \\ a.s. \\ =} &   \Text{rank}[ \bold{G}_{2}^n [\bold{V}_1^n \quad \bold{V}_2^n]] + \Text{rank}[ \bold{G}_{3}^n [\bold{V}_1^n \quad \bold{V}_2^n  \quad \bold{V}_3^n]]   -  \Text{rank}[ \bold{G}_{3}^n [\bold{V}_1^n \quad \bold{V}_2^n] ]  \nonumber\\
& \substack{ (\text{Lemma \ref{LAL}}) \\ a.s. \\ \leq} & \Text{rank}[ \bold{G}_{3}^n [\bold{V}_1^n \quad \bold{V}_2^n  \quad \bold{V}_3^n]]  \leq n, \nonumber
\end{eqnarray}
where (a) follows from the sub-modularity of rank of matrices (see Lemma \ref{ranksubmod} stated below);
and (b) follows from Lemma \ref{MIMORRI} applied to  $\text{Rx}_2$ as the receiver which supplies delayed CSIT.
Therefore, the proof of (\ref{PDNobj}) is complete.
We now prove (\ref{PDNobj2}).
\begin{eqnarray}
 m_1(n) + m_3(n)   & \substack{ (\ref{decode})\\a.s. \\ =} &  
\Text{rank}[ \bold{G}_{1}^n \bold{V}_1^n]  + m_3(n) \nonumber\\
&  \substack{\text{(\ref{decode})} \\ a.s. \\ =} &   \Text{rank}[ \bold{G}_{1}^n \bold{V}_1^n] + \Text{rank}[ \bold{G}_{3}^n [\bold{V}_1^n \quad \bold{V}_2^n  \quad \bold{V}_3^n]]  -  \Text{rank}[ \bold{G}_{3}^n [\bold{V}_1^n \quad \bold{V}_2^n] ] \nonumber\\
 &  \substack{\text{(Lemma \ref{ranksubmod})}  \\ \leq} &  \Text{rank}[ \bold{G}_{1}^n \bold{V}_1^n]+ \Text{rank}[ \bold{G}_{3}^n [\bold{V}_1^n  \bold{V}_3^n]]  -  \Text{rank}[ \bold{G}_{3}^n \bold{V}_1^n ]\nonumber\\
&    \substack{ (\text{Lemma \ref{LAL})} \\ a.s.  \\ \leq} &  \Text{rank}[ \bold{G}_{3}^n [\bold{V}_1^n \quad \bold{V}_3^n]]    \leq n, \nonumber
\end{eqnarray}
which completes the proof of (\ref{PDNobj2}).
We now state the sub-modularity of rank of matrices (see ~\cite{Lovasz} for more details).

\begin{lemma} \label{ranksubmod}
{\bf (Sub-modularity of rank)} Consider a matrix $A^{m\times n} \in \mathbb C^{m\times n}$.
Let $A_{\mathcal S}$, ${\mathcal S} \subseteq \{  1,2 , \ldots , n\}$ denote the sub-matrix of $A$ created by those columns in $A$ which have their indices in ${\mathcal S}$.
Then, for any ${\mathcal S}_1,{\mathcal S}_2 \subseteq \{  1,2 , \ldots , n\}$,
$\Text{rank}[A_{{\mathcal S}_1}] + \Text{rank}[A_{{\mathcal S}_2}] \geq \Text{rank}[A_{{\mathcal S}_1\cap {\mathcal S}_2}] + \Text{rank}[A_{{\mathcal S}_1\cup {\mathcal S}_2}].$
\end{lemma}
Note that a similar statement  is true  for sub-modularity of rank with respect to the rows of a matrix, instead of the columns as stated in Lemma \ref{ranksubmod}.

\begin{figure}[t]
\centering
\includegraphics[scale=.5, trim= 10mm 80mm 10mm 75mm]{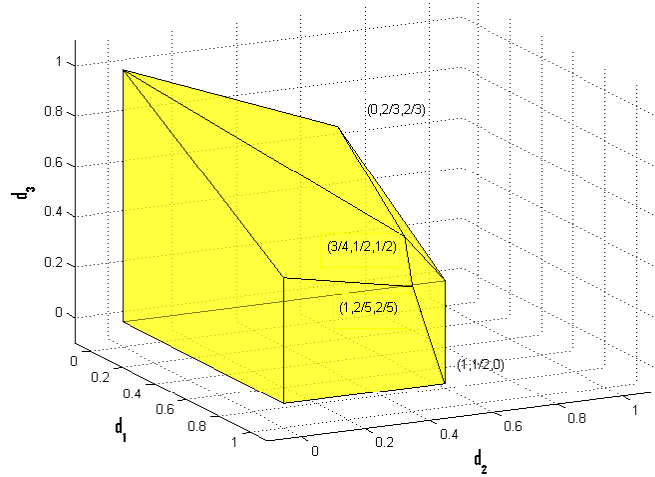}\\
\caption{LDoF Region for $PDD$.
}\label{PDDregion}
\end{figure}

\subsection{Proof of Achievability for Theorem \ref{mainthm3}} \label{PDDachsection}
The regions described in Theorem \ref{mainthm3} result in polytopes in $\mathbb R^3$; and therefore, the LDoF regions can be completely described via their extreme points. Many of such extreme points can be trivially achieved (e.g. the point $(1,1,0)$ for $PPD$); 
therefore, we only focus on the non-trivial extreme points and provide reference for each of them in Table \ref{tab:template2}.

\renewcommand{\arraystretch}{2}
\begin{table*}
\centering
\begin{tabular}{|c|c|}
\hline
$\text{CSIT States}$ & Non-trivial extreme points of the LDoF region and reference to the achievable scheme\\
\hline
PPD  & {$\begin{aligned}
& \left(1,0,\frac{1}{2}\right) ,\left(0,1,\frac{1}{2}\right) \text{ achieved in Section III-A of \cite{RaviMaddah}} \nonumber\\
&\left(1,1,\frac{1}{4}\right) \text{ achieved in Section IV-D of \cite{3UserHybrid}}
\end{aligned}$}  \\[1mm]
\hline
PDD & {$\!\begin{aligned}
&\left(1,0,\frac{1}{2}\right) ,\left(0,1,\frac{1}{2}\right) \text{ achieved in  Section III-A of \cite{RaviMaddah}} \nonumber\\
 &\left(1,\frac{2}{5},\frac{2}{5}\right) \text{ achieved  in Section IV-C of \cite{3UserHybrid}} \nonumber\\
&\left(\frac{3}{4},\frac{1}{2},\frac{1}{2}\right)  \text{ achieved  in Section \ref{PDDachsection} of this paper} \nonumber\\
&\left(0,\frac{2}{3},\frac{2}{3}\right) \text{ achieved in Section III-A of \cite{MAT}}
 \end{aligned}$} \\
\hline
PDN & {$\begin{aligned}
\left(1,\frac{1}{2},0\right) \text{ achieved in  Section III-A of \cite{RaviMaddah}} \end{aligned}$}  \\
\hline
DDD & {$\!\begin{aligned}
&\left(\frac{2}{3},\frac{2}{3},0\right), \left(\frac{2}{3},0,\frac{2}{3}\right), \left(0,\frac{2}{3},\frac{2}{3}\right) \text{ achieved in Section III-A of \cite{MAT}}  \nonumber \\
&\left(\frac{6}{11}, \frac{6}{11},\frac{6}{11}\right)  \text{ achieved in Section III-B of \cite{MAT}}  \nonumber
\end{aligned}$} \\
\hline
DDN & {$\!\begin{aligned}
\left(\frac{2}{3},\frac{2}{3},0\right) \text{ achieved in Section III-A of \cite{MAT}} \nonumber \end{aligned}$}  \\
\hline
\end{tabular}
\vspace{3mm}
\caption{Achievability results for extreme points of different configurations of hybrid CSIT for 3-user MISO BC}
\label{tab:template2}
\vspace{-15pt}
\end{table*}

The only non-trivial extreme point that has not yet been achieved in the literature according to Table \ref{tab:template2} belongs to $PDD$, and is $(\frac{3}{4},\frac{1}{2},\frac{1}{2})$.
The LDoF region suggested by Theorem \ref{mainthm3} for $PDD$ is shown in Fig. \ref{PDDregion}.
Therefore, we only  prove the achievability of $(\frac{3}{4},\frac{1}{2},\frac{1}{2})$ for $PDD$. 
The scheme is illustrated in Fig. \ref{PDDach}.
We will show how to deliver 3 symbols $(a_1,a_2,a_3)$ to $\text{Rx}_1$, 2 symbols $(b_1,b_2)$ to $\text{Rx}_2$, and 2 symbols $(c_1,c_2)$ to $\text{Rx}_3$ over 4 time slots in order to achieve $(d_1, d_2, d_3)= (\frac{3}{4},\frac{1}{2},\frac{1}{2})$.

\begin{figure}[t]
\centering
\includegraphics[scale=.34, trim= 10mm 10mm 10mm 10mm]{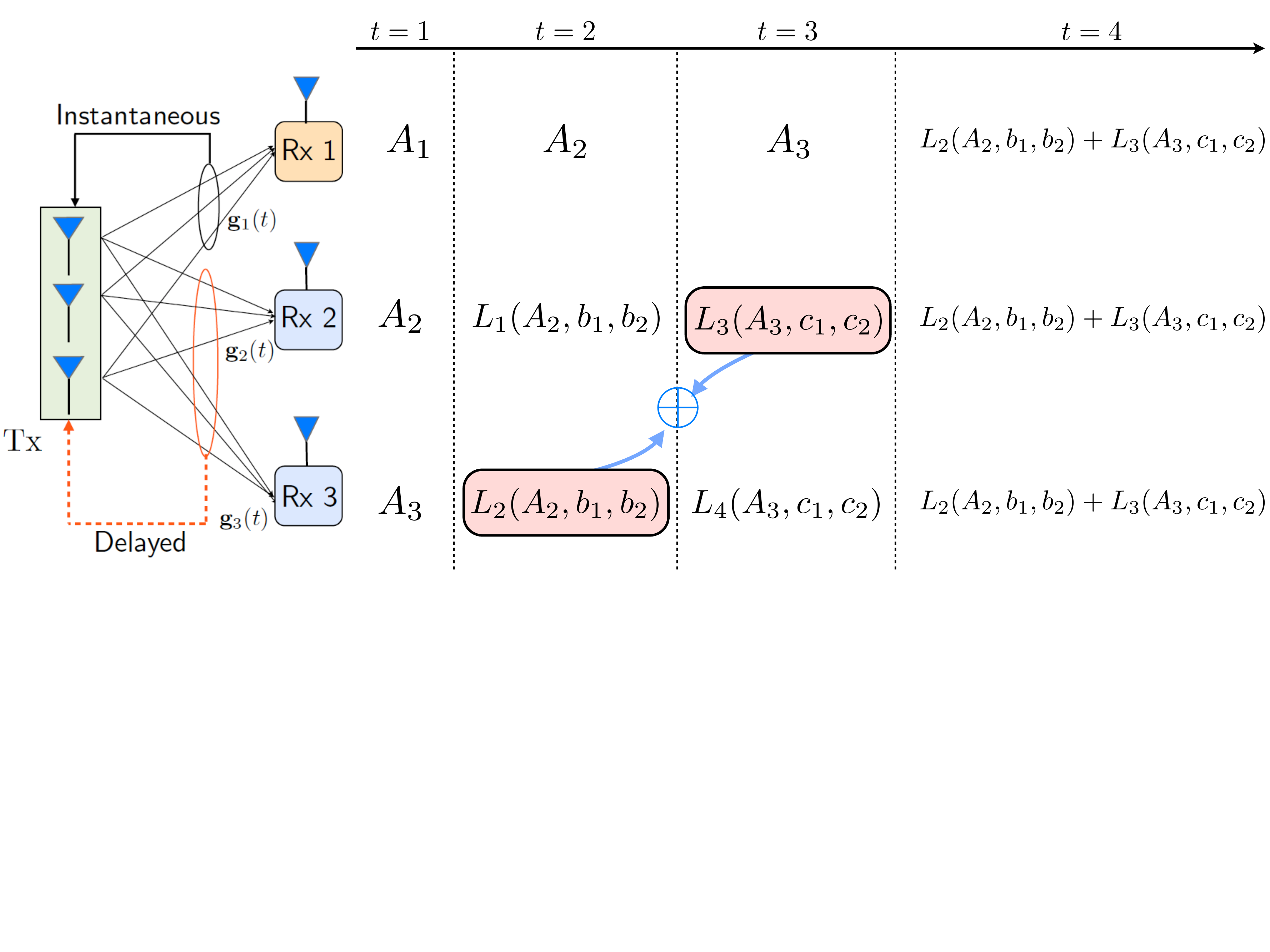}\\
\vspace{-65pt}
\caption{Achieving $(d_{1}, d_{2}, d_{3})=\left(\frac{3}{4},\frac{1}{2},\frac{1}{2}\right)$ for $ PDD$.
}\label{PDDach}
\vspace{-15pt}
\end{figure}

At $t=1$, we simply send the uncoded 3 symbols $(a_1,a_2,a_3)$, which are desired by  $\text{Rx}_1$.
Therefore, the transmit and received signals at the receivers are as follows (for the sake of DoF analysis, we ignore the additive noise):
\small
\begin{equation}
   \left[\begin{array}{c}{a_1}\vspace{-5pt} \\ { a_2} \vspace{-5pt}\\ {a_3} \end{array}\right], \quad 
\bm{y_j}(1) = \bm{\vec g_j}(1) \left[\begin{array}{c}{a_1}\vspace{-5pt} \\ { a_2}\vspace{-5pt} \\ {a_3}\end{array}\right], \quad j=1,2,3.
\end{equation}
\normalsize

Denote the linear combinations received by $\text{Rx}_1, \text{Rx}_2, \text{Rx}_3$ at $t=1$ by $ A_1, A_2,  A_3$.
Notice   that $\text{Rx}_1$ requires  $ A_2,A_3$ to be able to (almost surely) decode $(a_1,a_2,a_3)$.
Using delayed CSIT from $\text{Rx}_2,\text{Rx}_3$, transmitter can reconstruct $A_2,A_3$.

At $t=2$, the transmitter sends the symbols $A_2,b_1,b_2$ as
\small
\begin{equation}
\bm{\vec{x}}(2) = \left[\begin{array}{c}{1}\vspace{-5pt} \\ { 0}\vspace{-5pt} \\ {0}\end{array}\right] A_2 +\left[\begin{array}{c}{\bm{\vec g_1}(2)^\perp} \end{array}\right]^\top \left[\begin{array}{c}{b_1}\vspace{-5pt} \\  {b_2}\end{array}\right],
\end{equation}
where $[\bm{\vec g_1}(2)^\perp]$ is a $2\times 3$ matrix, where  $\bm{\vec g_1}(2)[\bm{\vec g_1}(2)^\perp]^\top =[0 \quad 0]$.
\normalsize
%\vspace{-2mm}
Therefore, the received signals at the $\text{Rx}_j$ is ($j=1,2,3$):
\small
\begin{equation}
\bm{y_j}(2) = \bm{\vec g_j}(2) \left (\left[\begin{array}{c}{1}\vspace{-5pt} \\ { 0}\vspace{-5pt} \\ {0} \end{array}\right] A_2 +\left[\begin{array}{c}{\bm{\vec g_1}(2)^\perp} \end{array}\right]^\top    \left[\begin{array}{c}{b_1}\vspace{-5pt} \\  {b_2}\end{array}\right]\right ).
\end{equation}
\normalsize
Note that by the end of time slot 2 $\text{Rx}_1$ is able to decode $A_2$.
We denote the linear combinations received by $ \text{Rx}_2, \text{Rx}_3$ at $t=2$ by $ L_1( A_2,  b_1,b_2),L_2( A_2,  b_1,b_2)$, respectively.

At $t=3$, the transmitted and received signals are:
\small
\begin{align*}
\bm{\vec{x}}(3) &= \left[\begin{array}{c}{1} \vspace{-5pt} \\ { 0}\vspace{-5pt} \\ {0}\end{array}\right] A_3 +\left[\begin{array}{c}{\bm{\vec g_1}(3)^\perp} \end{array}\right]^\top \left[\begin{array}{c}\vspace{-5pt}{c_1} \\  {c_2}\end{array}\right]
,\qquad \bm{y_j}(3) = \bm{\vec g_j}(3) \left (\left[\begin{array}{c}{1}\vspace{-5pt} \\ { 0}\vspace{-5pt} \\ {0}\end{array}\right] A_3 +\left[\begin{array}{c}{\bm{\vec g_1}(3)^\perp} \end{array}\right]^\top \left[\begin{array}{c}{c_1} \vspace{-5pt}\\  {c_2}\end{array}\right]\right ),
\end{align*}
\normalsize
which suggests that $\text{Rx}_1$ would be able to decode $A_3$.
We denote the linear combinations received by $ \text{Rx}_2, \text{Rx}_3$ at $t=3$ by $ L_3( A_3,  c_1,c_2),L_4( A_3,  c_1,c_2)$, respectively.

Note that if $\text{Rx}_2$ is given  $L_2( A_2,  b_1,b_2)$, it can use its past received signals (i.e., $A_2$ and $L_1( A_2,  b_1,b_2)$) together with $L_2( A_2,  b_1,b_2)$ to decode both $b_1,b_2$
Therefore, $\text{Rx}_2$ needs $L_2( A_2,  b_1,b_2)$. 
On the other hand, $\text{Rx}_2$ has access to $L_3( A_3,  c_1,c_2)$.
Similarly, $\text{Rx}_3$ needs $L_3( A_3,  c_1,c_2)$ to be able to decode both $c_1,c_2$, and it has access to $L_2( A_2,  b_1,b_2)$.
Therefore,  at $t=4$, the transmitter sends $L_2( A_2,  b_1,b_2)+L_3( A_3,  c_1,c_2)$, which is of interest to both $\text{Rx}_2,\text{Rx}_3$; this is because $\text{Rx}_3$ can then cancel $L_2( A_2,  b_1,b_2)$ from its received signal at $t=4$ to obtain $L_3( A_3,  c_1,c_2)$ which it needs. Similarly, $\text{Rx}_2$ can cancel $L_3( A_3,  c_1,c_2)$ from its received signal at $t=4$ to obtain $L_2( A_2,  b_1,b_2)$ which it needs.
Consequently, all receivers will be able to decode their desired symbols by the end of the fourth time slot; hence, the DoF tuple $(\frac{3}{4},\frac{1}{2},\frac{1}{2})$ is  achievable. 
See Fig. \ref{PDDach} for an illustration of the achievable scheme.

\section{$k$-User MISO BC with Hybrid CSIT}
In this section we focus on the general $k$-user MISO BC with hybrid CSIT.
In particular, we first present an outer bound on the LDoF region of the general $k$-user MISO BC  for any arbitrary hybrid CSIT configuration.
Then, we show that the bound provides an approximate characterization of $\Text{LDoF}_{\Text{sum}}$ for the case of $|\mathcal P| \geq |\mathcal D|$, and exact characterization of $\Text{LDoF}_{\Text{sum}}$ for  $|\mathcal D|  =  1$.
We then present the key tools needed for proving  the general outer bound; and finally, we prove the outer bound on the LDoF region.

\begin{theorem} \label{mainthmk}
Given a hybrid CSIT configuration, i.e., a partition of $k$ users into disjoint sets $\mathcal P, \mathcal D,$ and $\mathcal N$ as defined in Definition \ref{CSITsets},
 the $\Text{LDoF}_{\Text{region}}$ is contained in the following region:
\begin{align}
\Text{LDoF}_{\Text{region}} \subseteq \bigg \{ \quad (d_1,\ldots,d_k) \quad   | \quad  & 0\leq d_1,\ldots,d_k\leq 1, \nonumber\\
& \forall i\in   \mathcal D ,  \forall \pi_{\mathcal P \cup \mathcal D \setminus i}, \quad    \sum_{j =1}^{|\mathcal P|+ |\mathcal D|-1} \frac{d_{\pi_{\mathcal P \cup \mathcal D \setminus i} (j)}}{2^j} + d_{i} +\sum_{j\in \mathcal N} d_j \leq 1,  \label{IDBthm}\\
&\forall \pi_{\mathcal D}, \quad \sum_{j\in \mathcal P}\frac{d_j}{k} + \sum_{j=1}^{|\mathcal D|}  \frac{d_{\pi_{\mathcal D}(j)}}{j} +\sum_{j\in \mathcal N} d_j \leq 1,  \label{MATthm} \\
& \forall i\in  \mathcal P \cup \mathcal D , \quad     d_{i} +\sum_{j\in \mathcal N} d_j \leq 1 \quad
 \bigg \}. \label{LALbound}
\end{align}
\end{theorem}

Theorem \ref{mainthmk} enables us to approximately characterize $\text{LDoF}_{\text{sum}}$ to within $\frac{|\mathcal P|}{2^{|\mathcal P|}}$ for a broad range of CSIT configurations ($|\mathcal P| \geq |\mathcal D|$).
This  gap (i.e. $\frac{|\mathcal P|}{2^{|\mathcal P|}}$) is less than or equal to 0.5, and decays exponentially to zero as $|\mathcal P|$ increases.
Moreover, Theorem \ref{mainthmk} allows us to exactly
characterize $\text{LDoF}_{\text{sum}}$ for the  case of $|\mathcal D|=1$.
These results are stated more precisely in the following two Propositions.

\begin{proposition} \label{propapprox}
For general $k$-user MISO BC with $|\mathcal P| \geq |\mathcal D|$, 
\begin{align}
|\mathcal P| \leq   \Text{LDoF}_{\Text{sum}} \leq  |\mathcal P|+ \frac{|\mathcal P|}{2^{|\mathcal P|}} \leq |\mathcal P| + \frac{1}{2}.\nonumber
\end{align}
\end{proposition}

\begin{proposition} \label{D1}
For general $k$-user MISO BC with $|\mathcal D|  =  1$, 
\begin{align}
\Text{LDoF}_{\Text{sum}} = |\mathcal P| + \frac{1}{2^{|\mathcal P|}}.
\end{align}
\end{proposition}
Proofs of Propositions \ref{propapprox}, \ref{D1} are provided in Appendix \ref{constgap} and Appendix  \ref{kAchiev}, respectively.
%\begin{remark}
%Note that one can use the outer  bound on the $\Text{LDoF}_{\Text{region}}$ presented in Theorem \ref{mainthmk} to provide an upper bound on $\Text{LDoF}_{\Text{sum}}$.
%For scenarios where $|\mathcal P| < |\mathcal D|$ such upper bound  can lead to a gap with the existing lower bound \footnote{The lower bound $\max \{ |\mathcal P| , \frac{|\mathcal P|+|\mathcal D|}{1+\ldots + \frac{1}{|\mathcal P|+|\mathcal D|}} \} $ is in fact comprised of two lower bounds: the lower bound $ |\mathcal P|$ which can be achieved when only serving receivers which supply instantaneous CSIT; and the lower bound $ \frac{|\mathcal P|+|\mathcal D|}{1+\ldots + \frac{1}{|\mathcal P|+|\mathcal D|}}$ can be achieved even when assuming that receivers supplying instantaneous CSIT in fact supply delayed CSIT, and using the result in  \cite{MAT}. } of $\max \{ |\mathcal P| , \frac{|\mathcal P|+|\mathcal D|}{1+\ldots + \frac{1}{|\mathcal P|+|\mathcal D|}} \} $ which  can potentially be significant. Hence, an interesting direction is to research new tools and techniques that can further improve the outer bound in Theorem \ref{mainthmk} for the case of $|\mathcal P| < |\mathcal D|$.
%\end{remark}
We will now prove Theorem \ref{mainthmk}. In particular, 
we first present the key ingredients of the proof, which are the generalizations of Lemmas \ref{MainIneq}-\ref{LAL}.
We then prove (\ref{IDBthm})-(\ref{LALbound}).% each in a separate  sub-section.

\subsection{Key Ingredients for Proof of Theorem \ref{mainthmk}}
Similar to the proof for the case of 3-user MISO BC with hybrid CSIT, we need to extend the Lemmas  \ref{MainIneq}-\ref{LAL}.
We present the generalizations here, and then prove Theorem \ref{mainthmk}.
We first present the  generalized version of   Interference Decomposition Bound in Lemma \ref{MainIneq}. The  proof is provided in Appendix \ref{IDBproof}.

\begin{lemma} \label{MainIneqk}
{\bf (Interference Decomposition Bound)} 
 Consider a fixed  linear coding strategy $ f^{(n)}$, with corresponding precoding matrices $\bold{V}_{1}^{n}, \bold{V}_{2}^{n},\ldots, \bold{V}_{k}^{n} $ as defined in (\ref{coding}).
 For any $\mathcal S\subseteq \{1,2,\ldots, k\}$, any $\ell \in \mathcal S,$   
and any $j\notin \mathcal S$ for which  $I_j=D$,
\begin{align} 
 & \frac{ \Text{rank}[ \bold{G}_{\ell}^n [\cup_{i\in \mathcal S}\bold{V}_{i}^{n}] ]  - \Text{rank}[ \bold{G}_{\ell}^n [\cup_{\substack{i\in \mathcal S\\ i\ne \ell }} \bold{V}_{i}^{n}  ] ]  + \Text{rank}[ \bold{G}_{j}^n [\cup_{\substack{i\in \mathcal S\\ i\ne \ell}} \bold{V}_{i}^{n}  ]]   }{2}   \stackrel{a.s.}{\leq}     \Text{rank}[ \bold{G}_{j}^n [\cup_{i\in \mathcal S}\bold{V}_{i}^{n}]],\label{IDB}
\end{align}
where $[\cup_{i\in \mathcal S} \bold{V}_i^n]$  denotes the row concatenation of the corresponding precoding matrices $\bold{V}_i^n$, where $i\in \mathcal S$.
\end{lemma}

\begin{remark}
Lemma \ref{MainIneq} is a special case of Lemma \ref{MainIneqk} where $\mathcal S=\{1,2\}, j=3, \ell =1$.
\end{remark}
We now present the generalized version of Lemma \ref{MIMORRI}, which is the second main ingredient of the proof, and is proved  in Appendix \ref{AppMIMORRI}.

\begin{lemma} \label{MIMORRIk}
{\bf (MIMO Rank Ratio Inequality for BC)}
Consider a linear coding strategy $f^{(n)} $, with corresponding $\bold{V}_{1}^{n},\ldots, \bold{V}_{k}^{n} $ as defined in (\ref{coding}).
Let $\bold{Y}_j^n \triangleq  \bold{G}_j^n [\cup_{i\in \mathcal S}\bold{V}_i^n]$, where $\mathcal S \subseteq \{1,2,\ldots , k\}$.
Also, consider  distinct receivers $\text{Rx}_{i_1},\ldots, \text{Rx}_{i_{j+1}}$, where $j=1,2,\ldots, k-1$ and $i_1,\ldots, i_{j+1}\in \{1,\ldots,k\} $. 
If $\text{Rx}_{i_1},\ldots, \text{Rx}_{i_{j}}$ supply delayed CSIT, then,
\begin{align}
& \frac{ \Text{rank} [\bold{Y}_{i_1}^n; \ldots ;\bold{Y}_{i_{j+1}}^n]  }{j+1}  \stackrel{a.s.}{\leq}  \frac{ \Text{rank} [\bold{Y}_{i_1}^n; \ldots ;\bold{Y}_{i_j}^n]  }{j}.
\end{align}
\end{lemma}

\begin{remark}
Lemma \ref{MIMORRI} is a special case of Lemma \ref{MIMORRIk} where $j=1$, $i_1=3 , i_2 = \ell$, and $\mathcal S=\{i\}$.
\end{remark}

Finally, we present the general version of Lemma \ref{LAL}, which is the third main ingredient for the proof of Theorem \ref{mainthmk}.
%The  lemma has appeared in \cite{OursISIT2014, OursGlobecom2014} for a different setting, i.e. the wiretap channel, where the transmitter transmits a message together with artificial noise for secrecy purposes.
%However, the proof follows similar steps, and  is provided in Appendix \ref{LALproof} for completeness.
\begin{lemma}\label{LALk}
{\bf (Least Alignment Lemma)}    For any linear coding strategy $ f^{(n)}$, with corresponding $\bold{V}_{1}^{n}, \ldots,\bold{V}_{k}^{n} $ as defined in (\ref{coding}), and any $\mathcal S\subseteq \{1,2,\ldots, k\}$, if $I_j=N$ for some $j\in \{1,2,\ldots, k\}$,
\begin{align*}
 \forall \ell\in \{1,2,\ldots, k\},\qquad    \Text{ rank} & \left[ \bold{G}_{\ell}^n[\cup_{i\in \mathcal S} \bold{V}_i^n]  \right] \stackrel{a.s.}{\leq} \Text{ rank} \left[\bold{G}_{j}^n[\cup_{i\in \mathcal S} \bold{V}_i^n]  \right],
\end{align*}
where $[\cup_{i\in \mathcal S} \bold{V}_i^n]$  denotes the row concatenation of the precoding matrices $\bold{V}_i^n$, where $i\in \mathcal S$.
\end{lemma}

Using these three ingredients we now proceed to the proof of Theorem \ref{mainthmk}, and in particular proving the bounds (\ref{IDBthm})-(\ref{LALbound}).

\subsection{Proof of Bound (\ref{IDBthm}) in Theorem \ref{mainthmk}} \label{bound1}
Without loss of generality, suppose $\mathcal P = \{ 1, \ldots , |\mathcal P| \} $, and $\mathcal D = \{ |\mathcal P|+1, \ldots , |\mathcal P|+|\mathcal D| \}$, and $\mathcal N = \{ |\mathcal P|+|\mathcal D|+1, \ldots , k \}$.
In addition, let $i=|\mathcal P|+ |\mathcal D|$, and $\pi_{\mathcal P \cup \mathcal D \setminus i}$ be the identity permutation.
Consequently, we can rewrite (\ref{IDBthm}), and our goal is to show
\begin{equation} \label{unpermuted}
 \sum_{j =1}^{|\mathcal P|+ |\mathcal D|-1} \frac{d_j}{2^j} + d_{|\mathcal P|+ |\mathcal D|} +\sum_{j = |\mathcal P| + |\mathcal D| + 1}^k d_j \leq 1.
\end{equation}

If the $k$-tuple $(d_{1},d_{2},\ldots ,d_{k})$ degrees-of-freedom are linearly achievable, then by Definition \ref{DoFdef}   there exists a sequence $\{  f^{(n)} \}_{n=1}^{\infty}$ such that for each $n$ and the corresponding choice of $(m_{1}(n),m_{2}(n),\ldots, m_{k}(n))$, $(\bold{V}_{1}^n,\bold{V}_{2}^n,\ldots, \bold{V}_{k}^n)$ satisfy the  conditions in (\ref{decode}) and (\ref{DoFcond}).
Therefore, it is sufficient to show
\begin{align}
& \sum_{j =1}^{|\mathcal P|+ |\mathcal D|-1} \frac{m_j(n)}{2^j} +  m_{|\mathcal P|+ |\mathcal D|}(n) + \sum_{j = |\mathcal P| + |\mathcal D| + 1}^k m_j(n)  \stackrel{a.s.}{\leq}  n.\label{proofgoal}
\end{align}
We upper bound each of the three terms on the L.H.S. of (\ref{proofgoal}) separately.
By induction and  application of Lemma \ref{MainIneqk} and (\ref{decode}), one can prove the following claim, which provides an upper bound for the first term on the L.H.S. of (\ref{proofgoal}), and is proved in Appendix \ref{cl12proof}.

\begin{claim} \label{cl12}
\begin{equation} \label{cl12eq}
 \sum_{j=1}^{|\mathcal P|+ |\mathcal D|-1} \frac{m_j(n)}{2^{j}}  \stackrel{a.s.}{\leq } \Text{rank}[ \bold{G}_{|\mathcal P|+ |\mathcal D|}^n [\bold{V}_{1}^{n}\ldots \bold{V}_{|\mathcal P|+ |\mathcal D|-1}^{n}]].
\end{equation}
\end{claim}
We now upper bound $m_{|\mathcal P|+ |\mathcal D|}(n)$, which is the second term on the L.H.S. of (\ref{proofgoal}).
By (\ref{decode}) we obtain
%\small
\begin{eqnarray}
m_{|\mathcal P|+ |\mathcal D|}(n) & \stackrel{a.s.}{ =} & \text{rank}[\bold{G}_{|\mathcal P|+ |\mathcal D|}^n [ \cup_{j=1}^k\bold{V}_j^n ]]  -  \text{rank}[\bold{G}_{|\mathcal P|+ |\mathcal D|}^n [ \cup_{j\ne {|\mathcal P|+ |\mathcal D|}}\bold{V}_j^n ]]  \nonumber\\
&\stackrel{\text{(Lemma \ref{ranksubmod})}}{\leq} & \text{rank}[\bold{G}_{|\mathcal P|+ |\mathcal D|}^n [ \bold{V}_1^n  \ldots \bold{V}_{|\mathcal P|+ |\mathcal D|}^n]]   -   \text{rank}[\bold{G}_{|\mathcal P|+ |\mathcal D|}^n [ \bold{V}_1^n  \ldots \bold{V}_{|\mathcal P|+ |\mathcal D|-1}^n]]\nonumber\\
 &\substack{(a)\\ a.s.\\ \leq} & \text{rank}[\bold{G}_{|\mathcal P|+ |\mathcal D|+1}^n [ \bold{V}_1^n  \ldots \bold{V}_{|\mathcal P|+ |\mathcal D|}^n]] -   \text{rank}[\bold{G}_{|\mathcal P|+ |\mathcal D|}^n [ \bold{V}_1^n  \ldots \bold{V}_{|\mathcal P|+ |\mathcal D|-1}^n]], \label{x113}
\end{eqnarray}
where (a) follows by Least Alignment Lemma (Lemma \ref{LALk}) since receiver $|\mathcal P|+ |\mathcal D|+1$ supplies no CSIT.

We now upper bound $\sum_{j = |\mathcal P| + |\mathcal D| + 1}^k m_j(n)$, which is the third term on the L.H.S. of (\ref{proofgoal}).
By  (\ref{decode}), for all $ i\in \{|\mathcal P|+ |\mathcal D|+1, \ldots, k\},$
\begin{align} 
&   m_i(n) \stackrel{a.s.}{=}  \text{rank}[\bold{G}_{i}^n [ \bold{V}_1^n  \ldots \bold{V}_{k}^n]] - \text{rank}[\bold{G}_{i}^n [ \cup_{j\ne i}\bold{V}_j^n ]] \stackrel{\text{(Lemma \ref{ranksubmod})}}{\leq }   \text{rank}[\bold{G}_{i}^n [ \bold{V}_1^n  \ldots \bold{V}_{i}^n]] - \text{rank}[\bold{G}_{i}^n[ \bold{V}_1^n  \ldots \bold{V}_{i-1}^n]] . \nonumber
\end{align}
Hence, by summing over all the inequalities for  $ i\in \{|\mathcal P|+ |\mathcal D|+1, \ldots, k\}$, we obtain
\small
\begin{align}
\sum_{j = |\mathcal P| + |\mathcal D| + 1}^k  m_j(n) \stackrel{a.s.}{\leq} & \text{rank}[\bold{G}_{k}^n [ \bold{V}_1^n  \ldots \bold{V}_{k}^n]]  - \text{rank}[\bold{G}_{|\mathcal P|+ |\mathcal D|+1}^n[ \bold{V}_1^n  \ldots \bold{V}_{|\mathcal P|+ |\mathcal D| }^n]] \nonumber\\
& +  \sum_{i=|\mathcal P|+ |\mathcal D|+1}^{k-1} (\text{rank}[\bold{G}_{i}^n [ \bold{V}_1^n  \ldots \bold{V}_{i}^n]]- \text{rank}[\bold{G}_{i+1}^n[ \bold{V}_1^n  \ldots \bold{V}_{i}^n]]). \label{x115}
\end{align}
\normalsize
Note that since receivers with index in $  \{|\mathcal P|+ |\mathcal D|+1, \ldots, k\}$ supply no CSIT, and due to their channel symmetry, for each $ i\in \{|\mathcal P|+ |\mathcal D|+1, \ldots, k-1\}$ we have
\begin{equation} \label{x116}
\text{rank}[\bold{G}_{i}^n [ \bold{V}_1^n  \ldots \bold{V}_{i}^n]] \stackrel{a.s.}{=} \text{rank}[\bold{G}_{i+1}^n[ \bold{V}_1^n  \ldots \bold{V}_{i}^n]].
\end{equation}

Therefore, by (\ref{x115}), (\ref{x116}) we obtain
\begin{align}
\sum_{j = |\mathcal P| + |\mathcal D| + 1}^k   m_j(n) \stackrel{a.s.}{\leq}& \text{rank}[\bold{G}_{k}^n [ \bold{V}_1^n  \ldots \bold{V}_{k}^n]] - \text{rank}[\bold{G}_{|\mathcal P|+ |\mathcal D|+1}^n [ \bold{V}_1^n  \ldots \bold{V}_{|\mathcal P|+ |\mathcal D|}^n]]. \label{x127}
\end{align}

Hence, by summing the inequalities  in (\ref{cl12eq}), (\ref{x113}), and (\ref{x127}) we obtain 
\begin{align}
& \sum_{j =1}^{|\mathcal P|+ |\mathcal D|-1} \frac{m_j(n)}{2^j} +  m_{|\mathcal P|+ |\mathcal D|}(n) + \sum_{j = |\mathcal P| + |\mathcal D| + 1}^k m_j(n) \stackrel{a.s.}{\leq}\text{rank}[\bold{G}_{k}^n [ \bold{V}_1^n  \ldots \bold{V}_{k}^n]] \leq n  ,\nonumber
\end{align}
 which proves (\ref{proofgoal}), thus, completing the  proof of bound (\ref{IDBthm}) in Theorem \ref{mainthmk}.
%We now prove bound (\ref{MATthm}) in Theorem \ref{mainthmk}.

\subsection{Proof of Bound (\ref{MATthm}) in Theorem \ref{mainthmk}} \label{bound2}

%$\left[\begin{array}{c}{A \quad B} \\ {C \quad  D}\end{array}\right]).$
%We now prove (\ref{MATthm}).%, which is the last inequality in Theorem \ref{mainthm}.
 Without loss of generality, suppose $\mathcal P = \{ 1, \ldots , |\mathcal P| \} $, and $\mathcal D = \{ |\mathcal P|+1, \ldots , |\mathcal P|+|\mathcal D| \}$, and $\mathcal N = \{ |\mathcal P|+|\mathcal D|+1, \ldots , k \}$.
In addition, let  $\pi_{\mathcal D }$ be the reverse of the identity permutation.
Consequently, our goal becomes to show
\begin{equation} \label{unpermuted2}
\sum_{j=1}^{|\mathcal P|}\frac{d_j}{k} + \sum_{j=|\mathcal P|+1}^{|\mathcal P|+|\mathcal D|}  \frac{d_j}{ |\mathcal P|+|\mathcal D| +1 - j} +\sum_{j= |\mathcal P|+|\mathcal D|+1}^{k} d_j \leq 1 .
\end{equation}

Suppose $(d_1,\ldots , d_k)$ are linearly achievable as defined in Definition \ref{DoFdef}.
Then, by (\ref{DoFcond}), it is sufficient to show
\begin{equation} \label{proofgoal2}
\sum_{j=1}^{|\mathcal P|}\frac{m_j(n)}{k} + \sum_{j=|\mathcal P|+1}^{|\mathcal P|+|\mathcal D|}  \frac{m_j(n)}{ |\mathcal P|+|\mathcal D| +1 - j} +\sum_{j= |\mathcal P|+|\mathcal D|+1}^{k} m_j(n) \stackrel{a.s.}{\leq}  n.
\end{equation}
We upper bound each of the three terms on the L.H.S. of (\ref{proofgoal2}) separately.
We first upper bound the first term.
By (\ref{decode}), for all  $j=1,\ldots, |\mathcal P|, $
\begin{eqnarray}
    m_j(n) & \stackrel{a.s.}{=}  & \text{rank}[\bold{G}_j^n [\bold{V}_1^n \ldots  \bold{V}_k^n ]]   -  \text{rank}[\bold{G}_j^n [\cup_{i\ne j}\bold{V}_i^n ]]  \nonumber\\
&\substack{\text{(Lemma \ref{ranksubmod})} \\a.s.\\ \leq} & \text{rank}[\bold{G}_j^n [\bold{V}_1^n \ldots  \bold{V}_j^n ]]  -  \text{rank}[\bold{G}_j^n [\bold{V}_{1}^n \ldots  \bold{V}_{j-1}^n ] ]  \nonumber\\
&\substack{\text{(a)} \\ \leq} & \text{rank}[[\bold{G}_1^n; \ldots ; \bold{G}_k^n ] [\bold{V}_1^n \ldots  \bold{V}_j^n ]]   -  \text{rank}[[\bold{G}_1^n; \ldots ; \bold{G}_k^n ] [\bold{V}_{1}^n \ldots  \bold{V}_{j-1}^n ] ],  \nonumber
\end{eqnarray}
where (a) follows from the fact that for four matrices $A,B,C,D$,
 $\Text{ rank}[A \quad B] -  \Text{ rank}[ B] \leq  \Text{ rank}[A \quad B; C \quad D] -  \Text{ rank}[ B;D]$, and it can be proven using straightforward linear algebra.

By summing the above inequalities for $j=1,\ldots, |\mathcal P|$, and dividing both sides of the resulting inequality by $k$ we obtain
\begin{align}
\sum_{j=1}^{|\mathcal P|} \frac{m_j(n)}{k}  \stackrel{a.s.}{\leq}   \frac{ \text{rank}[[\bold{G}_1^n; \ldots ; \bold{G}_k^n ] [\bold{V}_1^n \ldots  \bold{V}_{|\mathcal P|}^n ]]  }{k}.\label{instbound}
\end{align}

We now upper bound the second term on the L.H.S of (\ref{proofgoal2}).
For the receivers supplying delayed CSIT, i.e. $\text{Rx}_j$, where $j=|\mathcal P|+ 1,\ldots, |\mathcal P|+|\mathcal D|,$
 by (\ref{decode}) we have:
\begin{eqnarray}
   m_j(n) &\stackrel{a.s.}{=}&  \text{rank}[\bold{G}_j^n [\bold{V}_1^n \ldots  \bold{V}_k^n ]]   -  \text{rank}[\bold{G}_j^n [\cup_{i\ne j}\bold{V}_i^n ]]  \nonumber\\
&\substack{\text{(Lemma \ref{ranksubmod})} \\a.s.\\ \leq} & \text{rank}[\bold{G}_j^n [\bold{V}_1^n \ldots  \bold{V}_j^n ]] -  \text{rank}[\bold{G}_j^n [\bold{V}_{1}^n \ldots  \bold{V}_{j-1}^n ] ]  \nonumber\\
&\substack{\text{(b)} \\ \leq} & \text{rank}[[\bold{G}_j^n; \ldots ; \bold{G}_{|\mathcal P|+|\mathcal D| }^n ] [\bold{V}_1^n \ldots  \bold{V}_j^n ]]  -  \text{rank}[[\bold{G}_j^n; \ldots ; \bold{G}_{|\mathcal P|+|\mathcal D|}^n ] [\bold{V}_{1}^n \ldots  \bold{V}_{j-1}^n ] ],  \nonumber 
\end{eqnarray}
where (b) follows from the fact that for four matrices $A,B,C,D$,
 $\Text{ rank}[A \quad B] -  \Text{ rank}[ B] \leq  \Text{ rank}[A \quad B; C \quad D] -  \Text{ rank}[ B;D]$.
Hence, if we divide both sides of the above inequality by  $|\mathcal P|+|\mathcal D|+1 - j$, and sum over all inequalities for $ j=|\mathcal P|+ 1,\ldots, |\mathcal P|+|\mathcal D|$, we obtain
\begin{align}
&\sum_{j=|\mathcal P|+1}^{|\mathcal P|+|\mathcal D|}  \frac{m_j(n)}{ |\mathcal P|+|\mathcal D|+1  - j} \stackrel{a.s.}{\leq } \text{rank}  [\bold{G}_{|\mathcal P|+|\mathcal D|}^n   [\bold{V}_{1}^n \ldots   \bold{V}_{|\mathcal P|+|\mathcal D|}^n]  -   \frac{\text{rank}[[\bold{G}_{|\mathcal P| + 1}^n; \ldots ; \bold{G}_{|\mathcal P|+|\mathcal D|}^n ] [\bold{V}_{1}^n \ldots  \bold{V}_{|\mathcal P| }^n ] ]}{|\mathcal D|} \nonumber\\
&\qquad + \sum_{j= |\mathcal P| +1 }^{|\mathcal P| +|\mathcal D|-1}   (\frac{\text{rank}[[\bold{G}_j^n; \ldots ; \bold{G}_{|\mathcal P|+|\mathcal D| }^n ] [\bold{V}_1^n \ldots  \bold{V}_j^n ]]}{|\mathcal P|+|\mathcal D|+1 - j}  -  \frac{\text{rank}[[\bold{G}_{j+1}^n; \ldots ; \bold{G}_{|\mathcal P|+|\mathcal D|}^n ] [\bold{V}_{1}^n \ldots  \bold{V}_{j}^n ] ]}{|\mathcal P|+|\mathcal D| - j} )  \nonumber\\
&\substack{\text{(Lemma \ref{MIMORRIk})} \\a.s.\\ \leq}  \text{rank}  [\bold{G}_{|\mathcal P|+|\mathcal D|}^n   [\bold{V}_{1}^n \ldots   \bold{V}_{|\mathcal P|+|\mathcal D|}^n] -   \frac{\text{rank}[[\bold{G}_{|\mathcal P| + 1}^n; \ldots ; \bold{G}_{|\mathcal P|+|\mathcal D|}^n ] [\bold{V}_{1}^n \ldots  \bold{V}_{|\mathcal P| }^n ] ]}{|\mathcal D|} \nonumber\\
&\substack{\text{(Lemma \ref{LALk})} \\a.s.\\ \leq}  \text{rank}  [\bold{G}_{|\mathcal P|+|\mathcal D|+1}^n   [\bold{V}_{1}^n \ldots   \bold{V}_{|\mathcal P|+|\mathcal D|}^n] -   \frac{\text{rank}[[\bold{G}_{|\mathcal P| + 1}^n; \ldots ; \bold{G}_{|\mathcal P|+|\mathcal D|}^n ] [\bold{V}_{1}^n \ldots  \bold{V}_{|\mathcal P| }^n ] ]}{|\mathcal D|}. \label{delayedbound}
\end{align}

We now upper bound the third term on the L.H.S of (\ref{proofgoal2}) exactly the same way as we upper bounded the third term on the L.H.S. of (\ref{proofgoal}). To avoid redundancy, we only restate the resulting bound which was stated in (\ref{x127}).
\begin{align}
\sum_{j=|\mathcal P|+ |\mathcal D|+1}^{k} m_j(n) \stackrel{a.s.}{\leq}& \text{rank}[\bold{G}_{k}^n [ \bold{V}_1^n  \ldots \bold{V}_{k}^n]] - \text{rank}[\bold{G}_{|\mathcal P|+ |\mathcal D|+1}^n [ \bold{V}_1^n  \ldots \bold{V}_{|\mathcal P|+ |\mathcal D|}^n]]. \label{nocsitbound}
\end{align}

We now merge the upper bounds on individual terms on the L.H.S. of (\ref{proofgoal2}).
By summing (\ref{instbound}), (\ref{delayedbound}), and (\ref{nocsitbound}), we obtain
\begin{eqnarray} 
&&\sum_{j=1}^{|\mathcal P|}\frac{m_j(n)}{k} + \sum_{j=|\mathcal P|+1}^{|\mathcal P|+|\mathcal D|}  \frac{m_j(n)}{ |\mathcal P|+|\mathcal D| +1 - j} +\sum_{j= |\mathcal P|+|\mathcal D|+1}^{k} m_j(n) \nonumber\\
&\stackrel{a.s.}{\leq}& \text{rank}  [\bold{G}_{k}^n   [\bold{V}_{1}^n \ldots   \bold{V}_{k}^n]] +   \frac{ \text{rank}[[\bold{G}_1^n; \ldots ; \bold{G}_k^n ] [\bold{V}_1^n \ldots  \bold{V}_{|\mathcal P|}^n ]]  }{k}-  \frac{\text{rank}[[\bold{G}_{|\mathcal P| + 1}^n; \ldots ; \bold{G}_{|\mathcal P|+|\mathcal D|}^n ] [\bold{V}_{1}^n \ldots  \bold{V}_{|\mathcal P| }^n ] ]}{|\mathcal D|}  \nonumber\\
&\substack{(c) \\a.s.\\ \leq}& \text{rank}  [\bold{G}_{k}^n   [\bold{V}_{1}^n \ldots   \bold{V}_{k}^n]] \leq n, \label{finaleq}
\end{eqnarray}
where (c) follows from Claim \ref{MIMOvariant}, which is stated below and proved in Appendix \ref{MIMOvariantApp}.

\begin{claim} \label{MIMOvariant}
\begin{equation}
   \frac{ \Text{rank}[[\bold{G}_1^n; \ldots ; \bold{G}_k^n ] [\bold{V}_1^n \ldots  \bold{V}_{|\mathcal P|}^n ]]  }{k}\stackrel{a.s.}{\leq} \frac{\Text{rank}[[\bold{G}_{|\mathcal P| + 1}^n; \ldots ; \bold{G}_{|\mathcal P|+|\mathcal D|}^n ] [\bold{V}_{1}^n \ldots  \bold{V}_{|\mathcal P| }^n ] ]}{|\mathcal D|}.
\end{equation}
\end{claim}

Hence, from (\ref{finaleq}), the proof of (\ref{proofgoal2}) is complete, which concludes the proof of (\ref{MATthm}) in Theorem \ref{mainthmk}.

\subsection{Proof of Bound (\ref{LALbound}) in Theorem \ref{mainthmk}} \label{bound3}
The proof of (\ref{LALbound}) is similar to proof of (\ref{IDBthm}); however, the proof is presented here for completeness.
 Without loss of generality, suppose  $ i = |\mathcal P|+|\mathcal D| $, and $\mathcal N = \{ |\mathcal P|+|\mathcal D|+1, \ldots , k \}$.
Consequently, our goal is to show
\begin{equation} \label{unpermuted3}
 d_{ |\mathcal P|+|\mathcal D| } +\sum_{j= |\mathcal P|+|\mathcal D|+1}^{k} d_j \leq 1 .
\end{equation}

Suppose $(d_1,\ldots , d_k)$ are linearly achievable as defined in Definition \ref{DoFdef}.
Then, by (\ref{DoFcond}), it is sufficient to show
\begin{equation} \label{proofgoal3}
 m_{ |\mathcal P|+|\mathcal D| }(n) +\sum_{j= |\mathcal P|+|\mathcal D|+1}^{k} m_j(n) \stackrel{a.s.}{\leq}  n.
\end{equation}
We upper bound each of the two terms on the L.H.S. of (\ref{proofgoal3}) separately. 
For the first term on the L.H.S of (\ref{proofgoal3}), by (\ref{decode}), we obtain
\begin{eqnarray}
    m_{|\mathcal P|+|\mathcal D|}(n)  & \stackrel{a.s.}{=}  & \text{rank}[\bold{G}_{|\mathcal P|+|\mathcal D|}^n [\bold{V}_1^n \ldots  \bold{V}_k^n ]]   -  \text{rank}[\bold{G}_{|\mathcal P|+|\mathcal D|}^n [\cup_{i\ne |\mathcal P|+|\mathcal D|}\bold{V}_i^n ]]  \nonumber\\
& \substack{\text{(Lemma \ref{ranksubmod})} \\a.s.\\ \leq} & \text{rank}[\bold{G}_{|\mathcal P|+|\mathcal D|}^n [\bold{V}_1^n \ldots  \bold{V}_{|\mathcal P|+|\mathcal D|}^n ]]  -  \text{rank}[\bold{G}_{|\mathcal P|+|\mathcal D|}^n [\bold{V}_{1}^n \ldots  \bold{V}_{|\mathcal P|+|\mathcal D|-1}^n ] ] \nonumber  \\
&\substack{\text{(Lemma \ref{LALk})} \\a.s.\\ \leq} & \text{rank}[\bold{G}_{|\mathcal P|+|\mathcal D|+1}^n [\bold{V}_1^n \ldots  \bold{V}_{|\mathcal P|+|\mathcal D|}^n ]]  -  \text{rank}[\bold{G}_{|\mathcal P|+|\mathcal D|}^n [\bold{V}_{1}^n \ldots  \bold{V}_{|\mathcal P|+|\mathcal D|-1}^n ] ] .
\label{1stterm}
\end{eqnarray}

We now upper bound the second term on the L.H.S of (\ref{proofgoal3}), exactly the same way as we upper bounded the third term on the L.H.S. of (\ref{proofgoal}). To avoid redundancy, we only restate the resulting bound which was stated in (\ref{x127}).
\begin{align}
\sum_{j=|\mathcal P|+ |\mathcal D|+1}^{k} m_j(n) \stackrel{a.s.}{\leq}& \text{rank}[\bold{G}_{k}^n [ \bold{V}_1^n  \ldots \bold{V}_{k}^n]] - \text{rank}[\bold{G}_{|\mathcal P|+ |\mathcal D|+1}^n [ \bold{V}_1^n  \ldots \bold{V}_{|\mathcal P|+ |\mathcal D|}^n]]. 
\end{align}

We now sum the upper bounds on individual terms on the L.H.S. of (\ref{proofgoal3}):
\begin{align} 
&  m_{|\mathcal P|+|\mathcal D|}(n)  +  \sum_{j= |\mathcal P|+|\mathcal D|+1}^{k} m_j(n) \stackrel{a.s.}{\leq} \text{rank}  [\bold{G}_{k}^n   [\bold{V}_{1}^n \ldots   \bold{V}_{k}^n]]  -  \text{rank}[\bold{G}_{|\mathcal P|+|\mathcal D|}^n [\bold{V}_{1}^n \ldots  \bold{V}_{|\mathcal P|+|\mathcal D|-1}^n ] ]  \leq n,
\end{align}
which completes the proof of (\ref{proofgoal3}), thus concluding the proof of (\ref{LALbound}) in Theorem \ref{mainthmk}.

\section{conclusion}
In this paper we studied the impact of hybrid CSIT on the linear DoF (LDoF) of broadcast channels with a multiple-antenna transmitter and $k$ single-antenna receivers (MISO BC), where the CSIT supplied by each receiver can be instantaneous ($P$), delayed ($D$), or none ($N$).
We first focused on the $3$-user MISO BC; and we  completely characterized the DoF region for all possible hybrid CSIT configurations, assuming  linear encoding strategies at the transmitters.
%The result indicates that the state-of-the-art achievable schemes in \cite{3UserHybrid} are indeed sum-DoF optimal, when restricted to linear encoding schemes.
In order to prove the result, we presented 3 key tools, and in particular, developed  a novel bound,  called \emph{Interference Decomposition Bound}, which provides a lower bound on the interference dimension at a receiver which supplies delayed CSIT based on the average dimension of constituents of that interference, thereby decomposing the interference into its individual components.

We then extended our main proof ingredients to the general $k$-user setting; and  we  presented a general outer bound on linear DoF region of the $k$-user MISO BC with arbitrary CSIT configuration.
We demonstrated how the bound  provides an approximate characterization of  linear sum-DoF to within an additive gap of $0.5$ for the broad range of scenarios in which the number of receivers supplying instantaneous CSIT is greater than the  number of receivers supplying delayed CSIT.
%For other scenarios, improving both the achievability and converse techniques is an interesting future direction.
In addition, for the case where only one receiver supplies delayed CSIT, we completely characterized the linear sum-DoF.

There are several future directions the one can pursue in regards to this work. 
An interesting direction is to improve both the inner and outer bounds for linear DoF of $k$-user MISO BC, where the number of receivers supplying instantaneous CSIT is less than the number of receivers supplying delayed CSIT.
Another interesting future direction is to extend the results to the non-linear setting (DoF). 
To this aim, one needs to extend the three main ingredients of the proof of outer bounds to the non-linear setting.
Least Alignment Lemma has recently been extended to the non-linear setting in \cite{GholamiJafar}.
Hence, an interesting direction would be to extend the Interference Decomposition Bound and MIMO Rank Ratio Inequality for BC to the non-linear setting.

\begin{appendices}

\section{Proof of Converse for Theorem \ref{mainthm3} } \label{convmainthm3}

For each CSIT configuration considered in Table \ref{tab:template}
we provide the converse proof.
Note that the converse proof for the cases $PDD$ and $PDN$ are already provided in Section \ref{3converse}.
Furthermore, since for the case of  $PPP$ the only bound is $0\leq d_1,d_2,d_3 \leq 1$ according to Table \ref{tab:template}, the proof is trivial.
We now prove the converse for Theorem \ref{mainthm3} for the rest of the CSIT configurations.

\subsection{$PPD$}

Note that as mentioned in Remark \ref{PDDtoPPD}, in order to prove $\frac{d_1}{2} + \frac{d_2}{4} + d_3 \leq 1$ for $PDD$, we did not rely on  any specific CSIT assumption with respect to $\text{Rx}_2$. Therefore, the bound $\frac{d_1}{2} + \frac{d_2}{4} + d_3 \leq 1$ also holds for the case of $PPD$. 
Moreover, note that by symmetry one can conclude that $\frac{d_1}{4} + \frac{d_2}{2} + d_3 \leq 1$ also holds for  $PPD$.
Hence, since  $\frac{d_1}{2} + \frac{d_2}{4} + d_3 \leq 1$ and $\frac{d_1}{4} + \frac{d_2}{2} + d_3 \leq 1$ constitute the LDoF region for $PPD$  according to Table \ref{tab:template}, the derivations in the converse proof of $PDD$ also prove the converse for $PPD$.

\subsection{$PPN$}
According to Table \ref{tab:template}, it is sufficient to show that 
$ d_1 + d_3 \leq 1$ and $d_2+ d_3 \leq 1$.
We only show $d_1 + d_3 \leq 1$; since $d_2+ d_3 \leq 1$ can be proven similarly due to symmetry.
Suppose $(d_1, d_2, d_3)$ is linearly achievable as defined in Definition \ref{DoFdef}.
Thus, according to (\ref{DoFcond}), it is sufficient to show that $m_1(n) + m_3(n) \stackrel{a.s.}{\leq} n$.
By the Decodability condition in (\ref{decode}) we have,
\begin{eqnarray}
m_1(n)  + m_3(n)  & \substack{\text{(\ref{decode})} \\ a.s. \\ =}&
\Text{rank}[ \bold{G}_{1}^n \bold{V}_1^n]  + m_3(n) \nonumber\\
& \substack{\text{(\ref{decode})} \\ a.s. \\ =} &  \Text{rank}[ \bold{G}_{1}^n \bold{V}_1^n]+ \Text{rank}[ \bold{G}_{3}^n [\bold{V}_1^n \quad \bold{V}_2^n  \quad \bold{V}_3^n]]   -  \Text{rank}[ \bold{G}_{3}^n [\bold{V}_1^n \quad \bold{V}_2^n] ] \nonumber\\
 & \substack{\text{(Lemma \ref{ranksubmod})}  \\ \leq} &  \Text{rank}[ \bold{G}_{1}^n \bold{V}_1^n]+ \Text{rank}[ \bold{G}_{3}^n [\bold{V}_1^n \quad \bold{V}_3^n]]  -  \Text{rank}[ \bold{G}_{3}^n \bold{V}_1^n ]  \nonumber\\
 & \substack{ (\text{Lemma \ref{LAL})} \\ a.s.  \\ \leq} &   \Text{rank}[ \bold{G}_{3}^n [\bold{V}_1^n \quad \bold{V}_3^n]]    \leq n, \nonumber
\end{eqnarray}
%Hence, by dividing both sides of the above inequality by $n$, and taking the limit as $n\to \infty$ and using (\ref{decodepr13User2}), we obtain
%$ d_1 +   d_3 \leq 1$, 
which completes the proof of converse for the case of $PPN$.

\subsection{$PNN,DNN,NNN$}
According to Table \ref{tab:template}, it is sufficient to show that 
$ d_1 + d_2+ d_3 \leq 1$.
In addition, note that it is sufficient to prove $ d_1 + d_2+ d_3 \leq 1$ for  the case of $PNN$; since any upper bound for $PNN$ is also a valid bound for $DNN$ and $NNN$.
Suppose $(d_1, d_2, d_3)$ is linearly achievable as defined in Definition \ref{DoFdef}.
Then, according to (\ref{DoFcond}), it is sufficient to show that $m_1(n) + m_2(n)+ m_3(n) \stackrel{a.s.}{\leq} n$.
By the Decodability condition in (\ref{decode}) we have,
\begin{eqnarray}
m_1(n) +m_2(n)+ m_3(n) & \substack{ a.s. \\ =}  &
\Text{rank}[ \bold{G}_{1}^n \bold{V}_1^n]  +m_2(n)+ m_3(n) \nonumber\\
& \substack{\text{(\ref{decode})} \\ a.s. \\ =} &  \Text{rank}[ \bold{G}_{1}^n \bold{V}_1^n]+ \Text{rank}[ \bold{G}_{2}^n [\bold{V}_1^n \quad \bold{V}_2^n  \quad \bold{V}_3^n]] -  \Text{rank}[ \bold{G}_{2}^n [\bold{V}_1^n \quad \bold{V}_3^n] ] + m_3(n)  \nonumber\\
 & \substack{\text{(Lemma \ref{ranksubmod})}  \\ \leq} &  \Text{rank}[ \bold{G}_{1}^n \bold{V}_1^n]+ \Text{rank}[ \bold{G}_{2}^n [\bold{V}_1^n \quad \bold{V}_2^n]]  -  \Text{rank}[ \bold{G}_{2}^n \bold{V}_1^n ]  + m_3(n) \nonumber\\
 & \substack{(a)  \\ a.s.  \\ \leq}  &  \Text{rank}[ \bold{G}_{2}^n [\bold{V}_1^n \quad \bold{V}_2^n]] +m_3(n) \nonumber\\
 &\substack{\text{(\ref{decode})} \\ a.s. \\ =} &    \Text{rank}[ \bold{G}_{2}^n [\bold{V}_1^n \quad \bold{V}_2^n]] + \Text{rank}[ \bold{G}_{3}^n [\bold{V}_1^n \quad \bold{V}_2^n  \quad \bold{V}_3^n]]  -  \Text{rank}[ \bold{G}_{3}^n [\bold{V}_1^n \quad \bold{V}_2^n] ]  \nonumber\\
 & \substack{ (\text{Lemma \ref{LAL})} \\ a.s.  \\ \leq}  &  \Text{rank}[ \bold{G}_{3}^n [\bold{V}_1^n \quad \bold{V}_2^n\quad \bold{V}_3^n]] \leq n , \nonumber
\end{eqnarray}
where (a) follows by applying Lemma \ref{LAL} and considering $\text{Rx}_2$ as the receiver which supplies no CSIT.
%Hence, by dividing both sides of the above inequality by $n$, and taking the limit as $n\to \infty$, we obtain
%$ d_1 +   d_3 \leq 1$, 
Hence, the proof of converse for the cases  $PNN,DNN,NNN$ is complete.

\subsection{$DDD$}
The bounds stated in Table \ref{tab:template} for $DDD$ have been proven in \cite{MAT} for general encoding schemes via network enhancement  and using the fact that in physically degraded broadcast channel feedback does not increase the capacity.
Therefore, the same bounds also hold for the class of linear schemes.
See \cite{MAT} for the bounds on the DoF of $k$-user MISO broadcast channel with delayed CSIT.

\subsection{$DDN$}
Note that according to Table \ref{tab:template} and due to symmetry of the first two users it is sufficient to show that 
$ \frac{d_1}{2} + d_2+ d_3 \leq 1$. The other inequality (i.e. $d_1 +  \frac{d_2}{2} + d_3 \leq 1$) can be proven similarly.
Suppose $(d_1, d_2, d_3)$ is linearly achievable, as defined in Definition \ref{DoFdef}. 
Thus, according to (\ref{DoFcond}), it is sufficient to show that $\frac{m_1(n)}{2} + m_2(n)+ m_3(n) \stackrel{a.s.}{\leq} n$.
We have
\begin{eqnarray}
\frac{m_1(n)}{2} +m_2(n)+ m_3(n)  &\substack{(\ref{decode})\\ a.s. \\ =}  &
\frac{\Text{rank}[ \bold{G}_{1}^n \bold{V}_1^n]}{2}  +m_2(n)+ m_3(n) \nonumber\\
& \substack{\text{(\ref{decode})} \\ a.s. \\ =}  &   \frac{\Text{rank}[ \bold{G}_{1}^n \bold{V}_1^n]}{2}+ \Text{rank}[ \bold{G}_{2}^n [\bold{V}_1^n \quad \bold{V}_2^n  \quad \bold{V}_3^n]]   -  \Text{rank}[ \bold{G}_{2}^n [\bold{V}_1^n \quad \bold{V}_3^n] ] + m_3(n)  \nonumber\\
 & \substack{\text{(Lemma \ref{ranksubmod})}  \\ \leq} &  \frac{\Text{rank}[ \bold{G}_{1}^n \bold{V}_1^n]}{2}+ \Text{rank}[ \bold{G}_{2}^n [\bold{V}_1^n \quad \bold{V}_2^n]]   -  \Text{rank}[ \bold{G}_{2}^n \bold{V}_1^n ]  + m_3(n) \nonumber\\
 & \substack{ \leq}  & \frac{\Text{rank}[[ \bold{G}_{1}^n ; \bold{G}_{2}^n ]\bold{V}_1^n]}{2}+ \Text{rank}[ \bold{G}_{2}^n [\bold{V}_1^n \quad \bold{V}_2^n]]   -  \Text{rank}[ \bold{G}_{2}^n \bold{V}_1^n ]  + m_3(n) \nonumber\\
 & \substack{ (a) \\ a.s.  \\ \leq} &   \Text{rank}[ \bold{G}_{2}^n [\bold{V}_1^n \quad \bold{V}_2^n]] +m_3(n) \nonumber\\
 &\substack{\text{(\ref{decode})} \\ a.s. \\ =} &   \Text{rank}[ \bold{G}_{2}^n [\bold{V}_1^n \quad \bold{V}_2^n]] + \Text{rank}[ \bold{G}_{3}^n [\bold{V}_1^n \quad \bold{V}_2^n  \quad \bold{V}_3^n]]   -  \Text{rank}[ \bold{G}_{3}^n [\bold{V}_1^n \quad \bold{V}_2^n] ]  \nonumber\\
 & \substack{ (\text{Lemma \ref{LAL})} \\ a.s.  \\ \leq}  &  \Text{rank}[ \bold{G}_{3}^n [\bold{V}_1^n \quad \bold{V}_2^n\quad \bold{V}_3^n]] \leq n ,\nonumber
\end{eqnarray}
where (a) follows by applying Lemma \ref{MIMORRI} to $\text{Rx}_2$ as the receiver which supplies delayed CSIT.
Hence, the proof of converse for the case of  $DDN$ is complete.

\section{Proof of Interference Decomposition Bound (Proof of Lemmas \ref{MainIneq},\ref{MainIneqk})}  \label{IDBproof}
Note  that Lemma \ref{MainIneq} is a special case of Lemma \ref{MainIneqk} where $k=3,$ $\mathcal S=\{1,2\}, j=3$, and  $\ell=1$.
Therefore, in order to prove Lemma \ref{MainIneq} and Lemma \ref{MainIneqk} it is sufficient to prove only Lemma \ref{MainIneqk}.
We first restate Lemma \ref{MainIneqk} here for convenience.

\begin{lemmarep}{Lemma~\ref{MainIneqk}}
{\bf (Interference Decomposition Bound)} 
 Consider a fixed  linear coding strategy $ f^{(n)}$, with corresponding precoding matrices $\bold{V}_{1}^{n}, \bold{V}_{2}^{n},\ldots, \bold{V}_{k}^{n} $ as defined in (\ref{coding}).
 For any $\mathcal S\subseteq \{1,2,\ldots, k\}$, any $\ell \in \mathcal S,$   
and any $j\notin \mathcal S$ for which  $I_j=D$,
\begin{align} 
 & \frac{ \Text{rank}[ \bold{G}_{\ell}^n [\cup_{i\in \mathcal S}\bold{V}_{i}^{n}] ]  - \Text{rank}[ \bold{G}_{\ell}^n [\cup_{\substack{i\in \mathcal S\\ i\ne \ell }} \bold{V}_{i}^{n}  ] ]  + \Text{rank}[ \bold{G}_{j}^n [\cup_{\substack{i\in \mathcal S\\ i\ne \ell}} \bold{V}_{i}^{n}  ]]   }{2}   \stackrel{a.s.}{\leq}     \Text{rank}[ \bold{G}_{j}^n [\cup_{i\in \mathcal S}\bold{V}_{i}^{n}]].
\end{align}
\end{lemmarep}

To prove Lemma \ref{MainIneqk}, we first introduce some definitions.
Consider a fixed linear encoding function $f^{(n)}$, with corresponding precoding matrices $\bold{V}_1^n,\ldots, \bold{V}_k^n$ as defined in (\ref{coding}).
\begin{definition} \label{Tdef}
For  $\mathcal S \subseteq \{1,\ldots, k\},   \ell \in \mathcal S,    j\in \{1,\ldots, k\} $, we define
\begin{align*}
\bm{\mathcal T_1} & \triangleq  \{ t  \in \{1,\ldots, n\}\quad  \text{s.t.}  \quad  \Text{rank}[\bold{G}_{\ell}^t [\cup_{i\in \mathcal S}\bold{V}_i^t   ]] = \Text{rank}[\bold{G}_{\ell}^{t-1} [\cup_{i\in \mathcal S}\bold{V}_i^{t-1}   ]] +1\} \nonumber\\
\bm{\mathcal T_2} & \triangleq  \{ t  \in \bm{\mathcal T_1}  \quad \text{s.t.} \quad [\vec{\bold{g}}_{\ell}(t) [\cup_{i\in \mathcal S}\bold{V}_i(t)   ]  ] \in \Text{rowspan}[\bold{G}_{j}^{t-1}[\cup_{i\in \mathcal S}\bold{V}_i^{t-1}   ]]  \} .
\end{align*}
\end{definition}

\begin{remark}
$\bm{\mathcal T_1}$ is the subset of  time slots in which the dimension of received signal at $\text{Rx}_{\ell}$ increases, while $\bm{\mathcal T_2}$ is the subset of  $\bm{\mathcal T_1}$ in which the received signal at $\text{Rx}_{\ell}$ is already recoverable by  using the past received signals at $\text{Rx}_j$. 
The definitions of $\bm{\mathcal T_1},\bm{\mathcal T_2}$ focus only on the contribution of $\bold{V}_i^n$, where $i\in \mathcal S$, on the dimension of received signals at different receivers; because the statement of Lemma \ref{MainIneqk} only involves $\bold{V}_i^n$, where $i\in \mathcal S$.
\end{remark}

We now state two lemmas that are the main building blocks of the proof of Lemma \ref{MainIneqk}.
\begin{lemma}\label{m1m3bound}
\begin{equation}
 \Text{rank}[ \bold{G}_{\ell}^n [\cup_{\substack{i\in \mathcal S}}\bold{V}_i^{n}   ]   ]  - |\bm{\mathcal T_2}|  \\  \stackrel{a.s.}{\leq}    \Text{rank}[ \bold{G}_{j}^n [\cup_{\substack{i\in \mathcal S}}\bold{V}_i^{n}   ]   ].
\end{equation}
\end{lemma}

\begin{lemma}\label{T4bound}
\begin{align}
|\bm{\mathcal T_2}| - \Text{rank}[ \bold{G}_{\ell}^n  [\cup_{\substack{i\in \mathcal S\\ i\ne \ell}}\bold{V}_i^{n}   ]]  \stackrel{}{\leq}    \Text{rank}[ \bold{G}_{j}^n [  \cup_{\substack{i\in \mathcal S}}\bold{V}_i^{n}   ]]   -   \Text{rank}[ \bold{G}_{j}^n [\cup_{\substack{i\in \mathcal S\\ i\ne \ell}}\bold{V}_i^{n}    ]] .
\end{align}
\end{lemma}

Note that proof of Lemma \ref{MainIneqk} is immediate from summing the inequalities in Lemma \ref{m1m3bound} and Lemma \ref{T4bound}.
Hence, we will prove  Lemma \ref{m1m3bound} and Lemma \ref{T4bound}.

\subsection{Proof of Lemma \ref{m1m3bound}}
Before proving  Lemma \ref{m1m3bound}, we first provide its proof sketch for the special case of  $k=3, j=3, \ell=1, \mathcal S = \{1,2\}$, the same special case as considered in Lemma 1,  to emphasize the underlying ideas. 
For such special case, Lemma \ref{m1m3bound}  reduces to the following  inequality:
\begin{equation}
 \Text{rank}[ \bold{G}_{1}^n [\bold{V}_1^{n}   \quad \bold{V}_2^{n}   ]   ]  - |\bm{\mathcal T_2}|  \\  \stackrel{a.s.}{\leq}    \Text{rank}[ \bold{G}_{3}^n [\bold{V}_1^{n}   \quad \bold{V}_2^{n}   ]   ], \label{Lemma8special}
\end{equation}
which  can be re-written in the following equivalent form:
\begin{equation}
 n  -   \Text{rank}[ \bold{G}_{3}^n [\bold{V}_1^{n}   \quad \bold{V}_2^{n}   ]   ]  \stackrel{a.s.}{\leq}   n  - \Text{rank}[ \bold{G}_{1}^n [\bold{V}_1^{n}   \quad \bold{V}_2^{n}   ]   ]  + |\bm{\mathcal T_2}| . \label{Lemma8speq}
\end{equation}

Note that the L.H.S. of (\ref{Lemma8speq}) is basically the number of time slots $t\in \{ 1,\ldots , n\}$ in which $\Text{rank}[ \bold{G}_{3}^t [\bold{V}_1^{t}   \quad \bold{V}_2^{t}   ]   ]$ does not increase (compared to $\Text{rank}[ \bold{G}_{3}^{t-1} [\bold{V}_1^{t-1}   \quad \bold{V}_2^{t-1}   ]   ]$).
Let us denote the set of such time slots by $\bm{\mathcal T}$.
First, note that in each  $t\in \bm{\mathcal T}$, either  $\Text{rank}[ \bold{G}_{1}^t [\bold{V}_1^{t}   \quad \bold{V}_2^{t}   ]   ]$ increases by 1 (compared to $\Text{rank}[ \bold{G}_{1}^{t-1} [\bold{V}_1^{t-1}   \quad \bold{V}_2^{t-1}   ]   ]$), or it remains constant.
Accordingly, we partition $\bm{\mathcal T}$ into two sets, and upper bound the cardinality of each set. 
The number of those time slots $t\in \bm{\mathcal T}$ in which $\Text{rank}[ \bold{G}_{1}^t [\bold{V}_2^{t}   \quad \bold{V}_2^{t}   ]   ]$ remains constant is at most $  n  - \Text{rank}[ \bold{G}_{1}^n [\bold{V}_1^{n}   \quad \bold{V}_2^{n}   ]   ]$, which constitutes the first two terms on the R.H.S of (\ref{Lemma8speq}).

We now upper bound the number of time slots  $t\in \bm{\mathcal T}$, in which $\Text{rank}[ \bold{G}_{1}^t [\bold{V}_2^{t}   \quad \bold{V}_2^{t}   ]   ]$ increases by 1.
In each such time slot, $\text{Rx}_3$  receives an equation which is already recoverable by using its past received equations (since $\Text{rank}[ \bold{G}_{3}^t [\bold{V}_1^{t}   \quad \bold{V}_2^{t}   ]   ]$ does not increase).
But note that due to the assumption of delayed CSIT for $\text{Rx}_3$, the transmitter does not know the channels to $\text{Rx}_3$ when  transmitting its signals at time slot $t$; and the received signal at $\text{Rx}_3$ would be a random linear combination of transmit signals. 
In order for this random linear combination to be known at $\text{Rx}_3$, $\text{Rx}_3$ must have already been able to recover each of the individual signals transmitted at time $t$, based on its past received signals.
Note that if $\text{Rx}_3$ knows  each individual transmit signal at time $t$, it also knows any linear combination of them.
Hence, it can already recover what $\text{Rx}_1$ receives at time $t$.
Therefore,  the number of time slots $t\in \bm{\mathcal T}$ in which $\Text{rank}[ \bold{G}_{1}^t [\bold{V}_2^{t}   \quad \bold{V}_2^{t}   ]   ]$ increases by 1  is upper bounded by the number of time slots in which the received signal at  $\text{Rx}_1$  is already recoverable by $\text{Rx}_3$, and $\Text{rank}[ \bold{G}_{1}^t [\bold{V}_2^{t}   \quad \bold{V}_2^{t}   ]   ]$ increases by 1, which in turn, by the definition of $\bm{\mathcal T_2}$ is equal to $|\bm{\mathcal T_2}| $, the last term on the R.H.S. of (\ref{Lemma8speq}).
Thus, the proof sketch is complete.

The following is the general mathematical proof for Lemma \ref{m1m3bound}, which relies on the above approach.
Let us denote the indicator function by $I(.)$. We then have
%\small
\begin{align}
& n - \Text{rank}[ \bold{G}_{j}^n [\cup_{\substack{i\in \mathcal S}}\bold{V}_i^{n} ]] =  \sum_{t=1}^{n}  I(\Text{rank}[\bold{G}_{j}^t [\cup_{\substack{i\in \mathcal S}}\bold{V}_i^{t}  ]] = \Text{rank}[\bold{G}_{j}^{t-1}[\cup_{\substack{i\in \mathcal S}}\bold{V}_i^{t-1}   ]]  )  \nonumber\\
& = \sum_{t\in \bm{\mathcal T_1} }^{}  I(\Text{rank}[\bold{G}_{j}^t  [\cup_{\substack{i\in \mathcal S}}\bold{V}_i^{t}   ] ] = \Text{rank}[\bold{G}_{j}^{t-1}[\cup_{\substack{i\in \mathcal S}}\bold{V}_i^{t-1}  ] ]) +\sum_{t\in \bm{ \mathcal T_1^c}}^{}  I(\Text{rank}[\bold{G}_{j}^t [\cup_{\substack{i\in \mathcal S}}\bold{V}_i^{t} ]  ] = \Text{rank}[\bold{G}_{j}^{t-1}[\cup_{\substack{i\in \mathcal S}}\bold{V}_i^{t-1}   ] ])  \nonumber\\
&\substack{ \text{ (a) } \\ a.s.\\=}  \sum_{t\in \bm{\mathcal T_1}}^{}  I(  \Text{rowspan}[\cup_{\substack{i\in \mathcal S}}\bold{V}_i(t)   ] \subseteq \Text{rowspan}[\bold{G}_{j}^{t-1}[\cup_{\substack{i\in \mathcal S}}\bold{V}_i^{t-1}  ]  ])  +\sum_{t\in \bm{\mathcal T_1^c}}^{}  I(\Text{rank}[\bold{G}_{j}^t [\cup_{\substack{i\in \mathcal S}}\bold{V}_i^{t}  ]  ] = \Text{rank}[\bold{G}_{j}^{t-1}[\cup_{\substack{i\in \mathcal S}}\bold{V}_i^{t-1}   ] ])  \nonumber\\
& \leq   \sum_{t\in \bm{\mathcal T_1}}^{}  I(\vec{\bold{g}}_{\ell}(t) [\cup_{\substack{i\in \mathcal S}}\bold{V}_i(t)   ]  \in \Text{rowspan}[\bold{G}_{j}^{t-1}[\cup_{\substack{i\in \mathcal S}}\bold{V}_i^{t-1}  ] ])  +\sum_{t\in \bm{\mathcal T_1^c}}^{}  I(\Text{rank}[\bold{G}_{j}^t [\cup_{\substack{i\in \mathcal S}}\bold{V}_i^{t} ] ]= \Text{rank}[\bold{G}_{j}^{t-1} [\cup_{\substack{i\in \mathcal S}}\bold{V}_i^{t-1}  ] ]) 
 \nonumber\\
&=   |\bm{\mathcal T_2}| +\sum_{t\in \bm{\mathcal T_1^c}}^{}  I(\Text{rank}[\bold{G}_{j}^t [\cup_{\substack{i\in \mathcal S}}\bold{V}_i^{t}  ] ]= \Text{rank}[\bold{G}_{j}^{t-1} [\cup_{\substack{i\in \mathcal S}}\bold{V}_i^{t-1}   ] ]) \nonumber\\
& \leq   |\bm{\mathcal T_2}|  +\sum_{t\in \bm{\mathcal T_1^c}}^{}  1  =  |\bm{\mathcal T_2}| +n- | \bm{\mathcal T_1}| \nonumber\\
& \stackrel{ (b)}{=}   |\bm{\mathcal T_2}|  +n- \Text{rank}[\bold{G}_{\ell}^n [\cup_{\substack{i\in \mathcal S}}\bold{V}_i^{n}  ]  ], \nonumber
\end{align}
where 
(a) is due to Lemma \ref{DCSIT}, which is stated and proved in Appendix \ref{DCSITlemproof} \footnote{Lemma \ref{DCSIT} is a variation of Lemma 6 in \cite{Ours} which was stated for the setting with \emph{distributed} transmit antennas. Proof of Lemma \ref{DCSIT} follows similar steps as in the proof of Lemma 6 in \cite{Ours}; it is provided in Appendix \ref{DCSITlemproof} for completeness.};
and (b) follows immediately from the definition of $\bm{\mathcal T_1}$.
By rearranging the above inequality, the proof of Lemma \ref{m1m3bound} will be complete.

\subsection{Proof of Lemma \ref{T4bound}} \label{T4boundproof}
We first  state a claim which is useful in lower bounding the R.H.S. of the inequality in Lemma \ref{T4bound}, and it can be proved using simple linear algebra;  hence the proof is omitted for brevity.
\begin{claim} \label{dimrank}
For two matrices $A,B$ of the same row size, 
$ \Text{rank}[A \quad B] - \Text{ rank}[B] =  \Text{dim} ( \Text{span} ( [\vec s \quad \vec 0] \text{ s.t. }   [\vec s \quad \vec 0]\in \Text{rowspan}[A \quad B])) . $
\end{claim}

We are now ready to prove Lemma \ref{T4bound}. 
Let $[\vec{\bold{g}}_{\ell}(t) [\cup_{\substack{i\in \mathcal S\\ i\ne \ell}}\bold{V}_i(t)  ]]_{t\in \bm{\mathcal T_2}}$ denote the matrix constructed by rows  $\vec{\bold{g}}_{\ell}(t) [\cup_{\substack{i\in \mathcal S\\ i\ne \ell}}\bold{V}_i(t)  ]$, where $t\in \bm{\mathcal T_2}$.
We have
\begin{eqnarray}
 \Text{rank}[\bold{G}_{j}^n [\cup_{\substack{i\in \mathcal S}}\bold{V}_i^n]  ]]- \Text{rank}[\bold{G}_{j}^n  [\cup_{\substack{i\in \mathcal S\\ i\ne \ell}}\bold{V}_i^n]  ] ] & \stackrel{\text{(Claim \ref{dimrank})}}{=}&
\Text{dim} ( \Text{span} ( [\vec s \quad \vec 0] \text{ s.t. }   [\vec s \quad \vec 0]\in \Text{rowspan}[\bold{G}_{j}^n  [\cup_{\substack{i\in \mathcal S}}\bold{V}_i^n]  ]))  \nonumber\\
& \stackrel{(a)}{\geq}& 
\Text{dim} ( \Text{span} ( [\vec s \quad \vec 0] \text{ s.t. }  [\vec s \quad \vec 0]\in \Text{rowspan}[[\vec{\bold{g}}_{\ell}(t) [\cup_{\substack{i\in \mathcal S}}\bold{V}_i(t)  ]]_{t\in \bm{\mathcal T_2}} ] ))
 \nonumber\\
& \stackrel{\text{(Claim \ref{dimrank})}}{=}&  \Text{rank} [[\vec{\bold{g}}_{\ell}(t) [\cup_{\substack{i\in \mathcal S}}\bold{V}_i(t)  ]]_{t\in \bm{\mathcal T_2}}  ]     -  \Text{rank} [[\vec{\bold{g}}_{\ell}(t) [\cup_{\substack{i\in \mathcal S\\ i\ne \ell}}\bold{V}_i(t)  ]]_{t\in \bm{\mathcal T_2}}  ] \nonumber\\
& \stackrel{(b)}{=} &  |\bm{\mathcal T_2}|  - \Text{rank} [[\vec{\bold{g}}_{\ell}(t) [\cup_{\substack{i\in \mathcal S\\ i\ne \ell}}\bold{V}_i(t)  ]]_{t\in \bm{\mathcal T_2}}  ]  \nonumber\\
& \geq & |\bm{\mathcal T_2}|  - \text{rank}[\bold{G}_{\ell}^n  [\cup_{\substack{i\in \mathcal S\\ i\ne \ell}}\bold{V}_i^n]],  \nonumber
\end{eqnarray}
where (a) follows from the fact that for each $t\in \bm{\mathcal T_2}$,  $\vec{\bold{g}}_{\ell}(t) [\cup_{\substack{i\in \mathcal S}}\bold{V}_i(t)] \in \Text{rowspan}[\bold{G}_{j}^{t-1}[\cup_{\substack{i\in \mathcal S}}\bold{V}_i^{t-1}] ]$; hence, for each $t\in \bm{\mathcal T_2}$,  $\vec{\bold{g}}_{\ell}(t)  [\cup_{\substack{i\in \mathcal S}}\bold{V}_i(t)]  \in \Text{rowspan}[\bold{G}_{j}^{n} [\cup_{\substack{i\in \mathcal S}}\bold{V}_i^{n}]  ]]$;
 and therefore, 
$\Text{rowspan}[[\vec{\bold{g}}_{\ell}(t) [\cup_{\substack{i\in \mathcal S}}\bold{V}_i(t)  ]]_{t\in \bm{\mathcal T_2}} ] \subseteq \Text{rowspan}[\bold{G}_{j}^{n} [\cup_{\substack{i\in \mathcal S}}\bold{V}_i^{n}]  ] ]$.
Furthermore, (b) holds since $\Text{rank}[[\vec{\bold{g}}_{\ell}(t) [\cup_{\substack{i\in \mathcal S}}\bold{V}_i(t)  ] ]_{ t\in \bm{\mathcal T_2}}] = |\bm{ \mathcal T_2} |$, which is due to the following:
 since $\bm{\mathcal T_2} \subseteq \bm{\mathcal T_1}$, if $t\in \bm{\mathcal T_2}$, then $t\in \bm{\mathcal T_1}$.
Therefore, using the definition of $\bm{\mathcal T_1}$, we get 
\begin{equation}
\forall t\in \bm{\mathcal T_2}, \qquad  \vec{\bold{g}}_{\ell}(t) [\cup_{\substack{i\in \mathcal S}}\bold{V}_i(t)] \notin \Text{rowspan} \left ( \bold{G}_{\ell}^{t-1}[\cup_{\substack{i\in \mathcal S}}\bold{V}_i^{t-1}] \right  ).
\end{equation}
Consequently, the vectors $\vec{\bold{g}}_{\ell}(t) [\cup_{\substack{i\in \mathcal S}}\bold{V}_i(t)]$, where $t\in \bm{\mathcal T_2}$, are linearly independent; and therefore, we have $\Text{rank}[[\vec{\bold{g}}_{\ell}(t) [\cup_{\substack{i\in \mathcal S}}\bold{V}_i(t)  ] ]_{ t\in \bm{\mathcal T_2}}] = |\bm{ \mathcal T_2} |$.
Hence, the proof of Lemma \ref{T4bound}  is complete.

\section{Statement and Proof of Lemma \ref{DCSIT}} \label{DCSITlemproof}

\begin{lemma}\label{DCSIT}
Consider a fixed
linear coding strategy $f^{(n)}$ with corresponding precoding matrices $ \bold{V}_{1}^{n},\ldots , \bold{V}_{k}^{n}$ as defined in (\ref{coding}). 
Consider an arbitrary index $j$, where $j\in \{1,\ldots,k\}$, and assume 
$I_j\in \{D,N\}$; i.e.,
the transmitter has either delayed or no CSIT with respect to $\text{Rx}_j$.
In addition, consider an arbitrary set of receiver indices $\mathcal S$, where $\mathcal S \subseteq \{1,\ldots,k\}$.
For any $t \in \{1,2, \ldots ,n \}$, let $\mathcal  A_t,\mathcal B_t$,  denote the following sets of channel realizations:
\begin{itemize}
\item $\mathcal A_t\triangleq \{ \mathcal G^n| \quad \Text{rank}[G_{j}^{t}[\cup_{i\in \mathcal S} V_i^{t}]]=\Text{rank}[G_{j}^{t-1}[\cup_{i\in \mathcal S} V_i^{t-1}]]  \} .  $
\item $\mathcal B_t\triangleq \{ \mathcal G^n|  \quad \Text{rowspan} [\cup_{i\in \mathcal S} V_i(t)] \subseteq   \Text{rowspan} [G_{j}^{t-1}[\cup_{i\in \mathcal S} V_i^{t-1}]] \}$.
\end{itemize}
Then,
\begin{equation}
\Pr (\bm{\mathcal G}^n \in \cup_{t=1}^{n}(\mathcal A_t\cap \mathcal B_t^c))=0.\nonumber
\end{equation}
\end{lemma}

\begin{proof}
Note that due to Union Bound, it is sufficient to show that for any $t \in \{1,2, \ldots ,n \}$,
%\begin{equation}
$\Pr (\bm{\mathcal G}^n \in \mathcal A_t\cap \mathcal B_t^c)=0.$
%\end{equation}
Consider an arbitrary $t \in \{1,2, \ldots ,n \}$. 
Due to Total Probability Law, it is sufficient to show that for any channel realization of the first $t-1$ timeslots, denoted by $\mathcal G^{t-1}$, we have
\begin{equation}
\label{eq:mainEq1}  \Pr(\bm{\mathcal G}^n \in \mathcal A_t\cap  \mathcal B^c_t | \bm{\mathcal G}^{t-1}=\mathcal G^{t-1})=0.
\end{equation}

Consider an arbitrary channel realization of the first $t-1$ time slots $\mathcal G^{t-1}$ and precoding matrices $V_1^{t},\ldots, V_k^{t}$ (which are now deterministic because they are only function of the channel realizations for the first $t-1$ time slots).
Also, suppose that given $\mathcal G^{t-1}$, $ \mathcal  B^c_t$ occurs; since otherwise, the proof of (\ref{eq:mainEq1}) would be complete.
We denote the row $h$ of the matrix $[\cup_{i\in \mathcal S} V_i(t)]$ by $[\cup_{i\in \mathcal S} V_{i,h}(t)]$.
Note that by assuming $ \mathcal  B^c_t$ occurs, and denoting $\mathcal L = \Text{rowspan} [G_{j}^{t-1}[\cup_{i\in \mathcal S} V_i^{t-1}]]$, the following is true (according to the definition of $ \mathcal  B_t$):
\begin{align}
&\exists h\in\{1,\ldots,m \} \quad s.t. \quad [\cup_{i\in \mathcal S} V_{i,h}(t)] \notin   \mathcal L \nonumber\\
  \Rightarrow \quad & \exists h\in\{1,\ldots ,m  \} \quad s.t. \quad    \text{Proj}_{\mathcal L^\perp} [\cup_{i\in \mathcal S} V_{i,h}(t)] \ne 0.
\end{align}
Therefore, the $m \times (\sum_{i\in \mathcal S}^{}m_i(n))$ matrix 
$[\text{Proj}_{\mathcal L^\perp} [\cup_{i\in \mathcal S} V_{i,1}(t)]; \ldots ; \text{Proj}_{\mathcal L^\perp} [\cup_{i\in \mathcal S} V_{i,m}(t)] ]$ is non-zero, which means that its null space has dimension strictly lower than $m$, the number of its rows.
Hence, we have,
\begin{align*}
\Pr  (  \bm{\mathcal G}^n \in \mathcal A_t\cap  \mathcal B^c_t       &    | \bm{\mathcal G}^{t-1}=\mathcal G^{t-1}) \stackrel{(a)}{=} 
\Pr(\bm{\mathcal G}^n \in \mathcal A_t| \bm{\mathcal G}^{t-1}=\mathcal G^{t-1})   \\
&  \stackrel{(b)}{=}   \Pr(     \text{Proj}_{\mathcal L^\perp} [\vec{\bold{g}_{j}}(t)[\cup_{i\in \mathcal S} V_i(t)]] = 0     | \bm{\mathcal G}^{t-1}=\mathcal G^{t-1})   \\
& \stackrel{(c)}{=}   \Pr(  \vec{\bold{g}_{j}}(t)  [\text{Proj}_{\mathcal L^\perp} [\cup_{i\in \mathcal S} V_{i,1}(t)];\ldots;  \text{Proj}_{\mathcal L^\perp} [\cup_{i\in \mathcal S} V_{i,m}(t)] ] = 0     | \bm{\mathcal G}^{t-1}=\mathcal G^{t-1})   \\
& = \Pr( \vec{\bold{g}_{j}}(t)^\top    \in    \Text{nullspace}([\text{Proj}_{\mathcal L^\perp} [\cup_{i\in \mathcal S} V_{i,1}(t)];\ldots ;  \text{Proj}_{\mathcal L^\perp} [\cup_{i\in \mathcal S} V_{i,m}(t)] ]^\top)   | \bm{\mathcal G}^{t-1}=\mathcal G^{t-1}) \\
& \stackrel{(d)}{=}  0,
\end{align*}
where (a) holds since we assumed that for realization $\mathcal G^{t-1}$, $\mathcal B_t^c$ occurs;
(b) holds according to the definition of $\mathcal A_t$;
(c) holds due to linearity of orthogonal projection;
and (d)  holds since, as mentioned before, 
the  matrix 
$[\text{Proj}_{\mathcal L^\perp} [\cup_{i\in \mathcal S} V_{i,1}(t)]; \ldots; \text{Proj}_{\mathcal L^\perp} [\cup_{i\in \mathcal S} V_{i,m}(t)] ]^\top$ is non-zero, meaning that its null space, which is a subspace in $\mathbb R^m$, has dimension strictly lower than $m$.
Therefore, the probability that the random vector $ \vec{\bold{g}_{j}}(t)$ lies in a subspace in $\mathbb R^m$ of strictly lower dimension (than $m$) is zero.

\end{proof}

%\text{Proj}_{\mathcal L^c} [\cup_{i\in \mathcal S} V_{i,1}(t)]; \ldots ; \text{Proj}_{\mathcal L^c} [\cup_{i\in \mathcal S} V_{i,m}(t)] ]

\section{Proof of MIMO Rank Ratio Inequality for BC  (Proof of Lemmas \ref{MIMORRI},\ref{MIMORRIk})} \label{AppMIMORRI}
Note that Lemma \ref{MIMORRI} is a special case of Lemma \ref{MIMORRIk} where $k=3$, $j=1$, $\mathcal S=\{i\}$, and $i_1=3 , i_2 = \ell$.
Therefore, in order to prove Lemma \ref{MIMORRI} and Lemma \ref{MIMORRIk} it is sufficient to prove only Lemma \ref{MIMORRIk}.
We first re-state Lemma \ref{MIMORRIk} here for convenience.

\begin{lemmarep}{Lemma~\ref{MIMORRIk}}
{\bf (MIMO Rank Ratio Inequality for BC)}
Consider a linear coding strategy $f^{(n)} $, with corresponding $\bold{V}_{1}^{n},\ldots, \bold{V}_{k}^{n} $ as defined in (\ref{coding}).
Let $\bold{Y}_j^n \triangleq  \bold{G}_j^n [\cup_{i\in \mathcal S}\bold{V}_i^n]$, where $\mathcal S \subseteq \{1,2,\ldots , k\}$.
Also, consider  distinct receivers $\text{Rx}_{i_1},\ldots, \text{Rx}_{i_{j+1}}$, where $j=1,2,\ldots, k-1$ and $i_1,\ldots, i_{j+1}\in \{1,\ldots,k\} $. 
If $\text{Rx}_{i_1},\ldots, \text{Rx}_{i_{j}}$ supply delayed CSIT, then,
\begin{align}
& \frac{ \Text{rank} [\bold{Y}_{i_1}^n; \ldots ;\bold{Y}_{i_{j+1}}^n]  }{j+1}  \stackrel{a.s.}{\leq}  \frac{ \Text{rank} [\bold{Y}_{i_1}^n; \ldots ;\bold{Y}_{i_j}^n]  }{j}.
\end{align}
\end{lemmarep}

\begin{proof}
Without loss of generality, we suppose that $i_1=1, i_2=2, \ldots, i_{j+1}=j+1$.
Thus, we need to show that 
\begin{equation} \label{MIMORRIobj}
\frac{ \text{rank} [\bold{Y}_{1}^n; \ldots ;\bold{Y}_{{j+1}}^n]  }{j+1}  \stackrel{a.s.}{\leq}  \frac{ \text{rank} [\bold{Y}_{1}^n; \ldots ;\bold{Y}_{j}^n]  }{j}.
\end{equation}
Let us denote 
\begin{equation}
\text{rank}[A |  B] \triangleq \text{rank}[A ;  B] -  \text{rank}[B] .\label{conditiondef}
\end{equation}
Hence, by sub-modularity property of rank (Lemma \ref{ranksubmod}),  for matrices $A,B,C$ with the same number of columns,
\begin{equation}
\text{rank}[A |  B] \geq  \text{rank}[A |  B;C]; \label{conditioning}
\end{equation}
\begin{equation}
\text{rank}[A |  C] + \text{rank}[B |  C]   \geq  \text{rank}[A;B |  C]. \label{condSubAdd}
\end{equation}
Moreover, we denote $\bold{Y}_j(t) \triangleq  \vec{\bold{g}}_j(t) [\cup_{h\in \mathcal S}\bold{V}_h(t)]$ and $\bold{Y}^t \triangleq [\bold{Y}_1^t ;\ldots ;\bold{Y}_{j}^t]$.
For each $i=1,\ldots , j,$ we have
\begin{eqnarray}
&&\text{rank} [\bold{Y}_i(t)| [\bold{Y}^{t-1};  \bold{Y}_1(t); \ldots ; \bold{Y}_{i-1}(t) ]]  \nonumber\\
&\stackrel{(\ref{conditiondef})}{=}&  I\left (  \vec{\bold{g}}_i(t) [\cup_{h\in \mathcal S}\bold{V}_h(t)]  \notin \text{rowspan} [\bold{Y}^{t-1};  \bold{Y}_1(t); \ldots ; \bold{Y}_{i-1}(t) ]\right) \label{a1}\\ 
& =& 1 - I\left (  \vec{\bold{g}}_i(t) [\cup_{h\in \mathcal S}\bold{V}_h(t)]  \in \text{rowspan} [\bold{Y}^{t-1};  \bold{Y}_1(t); \ldots ; \bold{Y}_{i-1}(t) ]\right) \label{a2}\\
&\substack{(a)\\ a.s.\\ =}& 1 - I\left (   \text{rowspan}[\cup_{h\in \mathcal S}\bold{V}_h(t)]  \subseteq \text{rowspan} [\bold{Y}^{t-1};  \bold{Y}_1(t); \ldots ; \bold{Y}_{i-1}(t) ]\right) \label{a3}  \\
&\stackrel{(b)}{\geq}& 1 - I\left (  \vec{\bold{g}}_{j+1}(t)  [\cup_{h\in \mathcal S}\bold{V}_h(t)]  \in \text{rowspan} [\bold{Y}^{t-1};  \bold{Y}_1(t); \ldots ; \bold{Y}_{i-1}(t) ]\right) \label{a4} \\
&\stackrel{(\ref{conditiondef})}{ =}& \text{rank} [\bold{Y}_{j+1}(t)| [\bold{Y}^{t-1};  \bold{Y}_1(t); \ldots ; \bold{Y}_{i-1}(t) ]] \label{a5}  \\
&\stackrel{(\ref{conditioning})}{\geq} & \text{rank} [\bold{Y}_{j+1}(t)| [\bold{Y}^{t-1};   \bold{Y}_1(t); \ldots ;  \bold{Y}_j(t) ]]  \\
&\stackrel{(\ref{conditiondef})}{=} & \text{rank} [[\bold{Y}_1(t); \ldots ;  \bold{Y}_{j+1}(t)]| \bold{Y}^{t-1}] - \text{rank} [[\bold{Y}_1(t); \ldots ;  \bold{Y}_{j}(t)]| \bold{Y}^{t-1}]  \\
&\stackrel{(\ref{conditioning})}{\geq} & \text{rank} [[\bold{Y}_1(t); \ldots ;  \bold{Y}_{j+1}(t)]  | [\bold{Y}^{t-1}; \bold{Y}_{j+1}^{t-1}]] - \text{rank} [[\bold{Y}_1(t); \ldots ;  \bold{Y}_{j}(t)]  | \bold{Y}^{t-1}], \label{ineqs4MIMORRI}
\end{eqnarray}
where to see why (a) holds, we first present  the following variant of  Lemma \ref{DCSIT}: if $\mathcal A_t$  denotes the event:\\ $\vec{\bold{g}}_i(t) [\cup_{h\in \mathcal S}\bold{V}_h(t)]  \in \text{rowspan} [\bold{Y}^{t-1};  \bold{Y}_1(t); \ldots ; \bold{Y}_{i-1}(t) ]$, and $\mathcal B_t$ denotes the event:  $\text{rowspan}[\cup_{h\in \mathcal S}\bold{V}_h(t)]  \subseteq \\ \text{rowspan} [\bold{Y}^{t-1};  \bold{Y}_1(t); \ldots ; \bold{Y}_{i-1}(t) ]$, then
\begin{equation} \label{DCSITvariant}
\Pr \left(   \mathcal A_t   \cap   \mathcal B_t ^c    \right ) = 0,
\end{equation}
 which can be proven using the same steps as in proof of Lemma \ref{DCSIT}; therefore, its proof is omitted for brevity. 
As a result, $I(\mathcal A_t) = I(\mathcal A_t \cap \mathcal B_t)+I(\mathcal A_t \cap \mathcal B_t^c) \substack{(\ref{DCSITvariant})\\ a.s.\\ = }   I(\mathcal A_t \cap \mathcal B_t) = I(\mathcal B_t)$,
where the last equality holds since  occurrence of $\mathcal B_t$ implies occurrence of   $\mathcal A_t$.
Therefore, $1- I(\mathcal A_t) \stackrel{a.s.}{=}  1- I(\mathcal B_t) .$
In addition, note that the left-hand-side of (a) is $1- I(\mathcal A_t)$, and  the right-hand-side of (a) is $1- I(\mathcal B_t)$.
Hence,  (a) holds.
Moreover, (b) holds due to the fact that if 
$$  \text{rowspan}[\cup_{h\in \mathcal S}\bold{V}_h(t)]  \subseteq \text{rowspan} [\bold{Y}^{t-1};  \bold{Y}_1(t); \ldots ; \bold{Y}_{i-1}(t) ],$$ 
then, 
$$\vec{\bold{g}}_{j+1}(t)  [\cup_{h\in \mathcal S}\bold{V}_h(t)]  \in \text{rowspan} [\bold{Y}^{t-1};  \bold{Y}_1(t); \ldots ; \bold{Y}_{i-1}(t) ].$$

By summing both sides of (\ref{ineqs4MIMORRI}) over $i=1,\ldots , j,$ and using (\ref{conditiondef}) we obtain,
\begin{equation}
\text{rank} [[\bold{Y}_1(t); \ldots ; \bold{Y}_j(t)]  | \bold{Y}^{t-1} ]   \stackrel{a.s.}{\geq }  j \left ( \text{rank} [[\bold{Y}_1(t); \ldots ;  \bold{Y}_{j+1}(t)] | [\bold{Y}^{t-1}; \bold{Y}_{j+1}^{t-1}]] - \text{rank} [[\bold{Y}_1(t); \ldots ;  \bold{Y}_{j}(t)]| \bold{Y}^{t-1}]\right); 
\end{equation}
and by rearranging the above inequality and dividing both sides by $j(j+1)$ we obtain
\begin{equation}
\frac{  \text{rank} [[\bold{Y}_1(t); \ldots ; \bold{Y}_j(t)]  | \bold{Y}^{t-1} ] }{j}  \stackrel{a.s.}{\geq }  \frac{ \text{rank} [[\bold{Y}_1(t); \ldots ;  \bold{Y}_{j+1}(t)]|[ \bold{Y}^{t-1}; \bold{Y}_{j+1}^{t-1}]]}{j+1}. 
\end{equation}
Finally, by summing both sides of the above inequality over all $t=1,\ldots, n$, and using the definition in (\ref{conditiondef}), the  proof of (\ref{MIMORRIobj}) would be complete, which concludes the proof of  Lemma \ref{MIMORRIk}.
\end{proof}

\section{Proof of Least Alignment Lemma  (Proof of Lemmas \ref{LAL},\ref{LALk})} \label{LALproof}
Note that Lemma \ref{LAL} is a special case of Lemma \ref{LALk} where $k=3$ and $j=3$.
Therefore, in order to prove Lemma \ref{LAL} and Lemma \ref{LALk} it is sufficient to prove only Lemma \ref{LALk}.
We first re-state Lemma \ref{LALk} here for convenience.

\begin{lemmarep}{Lemma~\ref{LALk}}
{\bf (Least Alignment Lemma)}    For any linear coding strategy $ f^{(n)}$, with corresponding $\bold{V}_{1}^{n}, \ldots,\bold{V}_{k}^{n} $ as defined in (\ref{coding}), and any $\mathcal S\subseteq \{1,2,\ldots, k\}$, if $ I_j=N$ for some $j\in \{1,2,\ldots, k\}$,
\begin{align*}
 \forall \ell\in \{1,2,\ldots, k\},\qquad    \Text{ rank} & \left[ \bold{G}_{\ell}^n[\cup_{i\in \mathcal S} \bold{V}_i^n]  \right] \stackrel{a.s.}{\leq} \Text{ rank} \left[\bold{G}_{j}^n[\cup_{i\in \mathcal S} \bold{V}_i^n]  \right].
\end{align*}
\end{lemmarep}

\begin{proof}
Define $m(n)\triangleq \sum_{i\in \mathcal S}^{}  m_i(n)$.
We first state a lemma that will be later useful  in  proving  Lemma \ref{LALk}.

\begin{lemma}\label{roots}
{\bf(\cite{Measure})} For $n\in\mathbb N$, a multi-variate polynomial function on $\mathbb C^n$ to $\mathbb C$, is either identically $0$, or non-zero almost everywhere.
%\footnote{A version for analytic functions with values in an arbitrary Banach space can be found in \cite{Measure}.}
\end{lemma}

We now prove Lemma \ref{LALk}.
Denote by $[1:n]$ the set $\{1,\ldots, n\}$.
For any matrix $B_{n\times m(n)}$ and  $I_1\subseteq [1:n]$, and $I_2\subseteq [1:m(n)]$, we
denote by $B_{I_1,I_2}$ the sub-matrix of $B$ whose rows and columns are specified by $I_1$ and  $I_2$, respectively.
Define the set of channel  realizations $\mathcal A$ as:
%\small
\begin{equation}
 \mathcal A \triangleq \left\{\mathcal G^n | \text{rank}[G_{\ell}^n[\cup_{i\in \mathcal S} {V}_i^n]] > \text{rank}[G_{j}^n[\cup_{i\in \mathcal S} {V}_i^n]]\right\}. \label{Adef}
\end{equation}
%\normalsize

In order to prove $\text{rank}[\bold{G}_{\ell}^n[\cup_{i\in \mathcal S} \bold{V}_i^n]] \stackrel{a.s.}{\leq} \text{rank}[\bold{G}_{j}^n[\cup_{i\in \mathcal S} \bold{V}_i^n]]$, we only need to show $\Pr(\mathcal A)=0$.
Since a matrix $B_{n\times m(n)}$ has rank $r$ if and only if the maximum size of a square sub-matrix of $B$ with non-zero determinant is $r$, 
%\small
\begin{align}
 \mathcal A \subseteq  \{\mathcal G^n | \quad & \exists I_1\subseteq [1:n], I_2\subseteq [1:m(n)],  |I_1|=|I_2|, \nonumber\\
& s.t. \quad  \text{det}([G_{\ell}^n[\cup_{i\in \mathcal S} {V}_i^n]]_{I_1,I_2}) \ne   \text{det}([G_{j}^n[\cup_{i\in \mathcal S} {V}_i^n]]_{I_1,I_2}) = 0\}, \nonumber  
\end{align}
%\normalsize
which can be rewritten as
%\small
\begin{align}
 A \subseteq \cup_{\substack{I_1\subseteq [1: n] \\  I_2\subseteq [1: m(n)]   \\ |I_1|=|I_2|} } &  \left\{\mathcal G^n| 
  \text{det}([G_{\ell}^n[\cup_{i\in \mathcal S} {V}_i^n] ]_{I_1,I_2}) \ne 0, \quad \text{det}  ([G_{j}^n[\cup_{i\in \mathcal S} {V}_i^n] ]_{I_1,I_2}) = 0  \right\}. \label{detne}
\end{align}
%\normalsize

Let $X^n$ denote a diagonal matrix of size $n\times n$  where the elements on the diagonal  are variables in $ \mathbb C$.
Then, for any $I_1\subseteq [1: n], I_2\subseteq [1:m(n)] , \Text{ where } |I_1|=|I_2|$,  $\text{det}([X^n[\cup_{i\in \mathcal S} {V}_i^n]]_{I_1,I_2})$ is a multi-variate polynomial function in
the elements of $X^n$.
 Note that  if for some realization $X^n=G_{\ell}^n$,
$\text{det}([G_{\ell}^n[\cup_{i\in \mathcal S} {V}_i^n]]_{I_1,I_2}) \ne 0$,
then the polynomial function defined by $\text{det}([X^n[\cup_{i\in \mathcal S} {V}_i^n] ]_{I_1,I_2})$ is not identical to zero ( i.e., $\text{det}([X^n  [\cup_{i\in \mathcal S} {V}_i^n] ]_{I_1,I_2}) \stackrel{\Text{identical}}{\ne} 0$).
So, by (\ref{detne}), we have
%\small
\begin{align}
\mathcal A \subseteq  \cup_{\substack{I_1\subseteq [1: n] \\  I_2\subseteq [1:m(n)]   \\ |I_1|=|I_2|} } & \{\mathcal G^n   |   \text{det}([X^n[\cup_{i\in \mathcal S} {V}_i^n] ]_{I_1,I_2})\stackrel{\Text{identical}}{\ne} 0,\quad  \text{det}([G_{j}^n[\cup_{i\in \mathcal S} {V}_i^n] ]_{I_1,I_2}) = 0\} \nonumber\\
=  \cup_{\substack{I_1\subseteq [1: n] \\  I_2\subseteq [1:m(n)]   \\ |I_1|=|I_2|} } & \{\mathcal G^n   |   \text{det}([X^n [\cup_{i\in \mathcal S} {V}_i^n] ]_{I_1,I_2})\stackrel{\Text{identical}}{\ne} 0, \quad  G_{j}^n \text{ is root of } \text{det}([X^n[\cup_{i\in \mathcal S} {V}_i^n]  ]_{I_1,I_2})  \}. \label{detpol}
\end{align}
%\normalsize
Note that  by Lemma \ref{roots}, for every $I_1\in [1:n] ,I_2\in [1:m(n)], |I_1|=|I_2|$, we have
%\small
\begin{align}
\Pr(  \{ & \mathcal  G^n   |   \text{det}([X^n [\cup_{i\in \mathcal S} {V}_i^n]  ]_{I_1,I_2})\stackrel{\Text{identical}}{\ne} 0,
\quad  G_{j}^n   \text{  is root of } \text{det}([X^n [\cup_{i\in \mathcal S} {V}_i^n]]_{I_1,I_2})    \} ) =0.
\end{align}
%\normalsize
So, since finite union of measure-zero sets has measure zero, 
%\small
\begin{align}
\Pr(  \cup_{\substack{I_1\subseteq [1: n] \\  I_2\subseteq [1:m(n)]   \\ |I_1|=|I_2|} }&  \{\mathcal G^n   |   \text{det}([X^n [\cup_{i\in \mathcal S} {V}_i^n] ]_{I_1,I_2})\stackrel{\Text{identical}}{\ne} 0,\quad   G_{j}^n \text{: root of }  \text{det}([X^n [\cup_{i\in \mathcal S} {V}_i^n] ]_{I_1,I_2})  \}  )=0,
\end{align}
%\normalsize
which by (\ref{detpol})  implies that
$\Pr(\mathcal A)=0.$
Therefore, according to the definition of $\mathcal A$ in (\ref{Adef}),
\begin{align}
\Text{ rank}  \left[ \bold{G}_{\ell}^n[\cup_{i\in \mathcal S} \bold{V}_i^n]  \right] \stackrel{a.s.}{\leq} \Text{ rank} \left[\bold{G}_{j}^n[\cup_{i\in \mathcal S} \bold{V}_i^n]  \right],
\end{align}
which completes the proof of Least Alignment Lemma.
% $\blacksquare$
\end{proof}

\begin{remark} \label{LALExt}
Using the same line of argument as in the proof of Lemma \ref{LALk}, 
one can prove Lemma \ref{LALk} for a more general network setting where there are arbitrary number of transmitters, and the transmitters have arbitrary number of antennas.
In addition, the statement of Lemma \ref{LALk} holds even if $\text{Rx}_j,\text{Rx}_{\ell}$ have multiple but equal number of antennas.
\end{remark}

\section{Proof of Proposition \ref{propapprox} (Constant Gap Characterization for $|\mathcal P| \geq |\mathcal D|) $} \label{constgap}
In this Appendix we show that for $|\mathcal P| \geq |\mathcal D| $,  Theorem \ref{mainthmk} leads to  an approximate characterization of $\Text{LDoF}_{\Text{sum}}$ to within an additive gap of  $\frac{1}{2}$, as presented in Proposition \ref{propapprox}.
First, note that for the special case of $|\mathcal P| = |\mathcal D| = 0$,  $\text{LDoF}_{\text{region}} $ is completely characterized by
$\{(d_1,\ldots ,d_k) \quad |\quad  \sum_{i=1}^{k}d_i\leq 1\}$.
Thus, henceforth we assume that $|\mathcal P| > 0$.

Moreover, note that a naive lower bound for $\Text{LDoF}_{\Text{sum}}$ is $|\mathcal P|$; since we can focus only on the $|\mathcal P|$ receivers that provide instantaneous CSIT, and for those $|\mathcal P|$ receivers we can perform zero-forcing to cancel interference and achieve  $|\mathcal P|$ as a lower bound on $\text{LDoF}_{\text{sum}}$. 
Using this lower bound we show that for the case where $|\mathcal P| \geq |\mathcal D| $, the statement of Proposition \ref{propapprox} holds.
In particular, we first consider the case where  $ |\mathcal D| =0. $ 
For this case, by (\ref{LALbound}) in Theorem \ref{mainthmk} we have
\begin{align}
 \forall i\in  \mathcal P \cup \mathcal D , \quad     d_{i} +\sum_{j\in \mathcal N} d_j \leq 1,
\end{align}
which, together with $\forall i, d_i\leq 1$, yields
\begin{align}
\text{LDoF}_{\text{sum}} \leq |\mathcal P|.
\end{align}
Hence, the naive lower bound of $|\mathcal P|$ on $\Text{LDoF}_{\Text{sum}}$ is tight for the case where  $ |\mathcal D| =0. $ 
Moreover, $\text{LDoF}_{\text{sum}}$ for the special case where  $ |\mathcal D| =1 $ is characterized in Proposition \ref{D1}.
Therefore, we only need to prove Proposition \ref{propapprox} for the case of $|\mathcal P| \geq |\mathcal D| > 1. $
Recall that by (\ref{IDBthm}) in Theorem \ref{mainthmk},
\begin{equation}
\forall i\in   \mathcal D ,  \forall \pi_{\mathcal P \cup \mathcal D \setminus i}, \quad    \sum_{j =1}^{|\mathcal P|+ |\mathcal D|-1} \frac{d_{\pi_{\mathcal P \cup \mathcal D \setminus i} (j)}}{2^j} + d_{i} +\sum_{j\in \mathcal N} d_j \leq 1.
\end{equation}

Without loss of generality, suppose $\mathcal P = \{ 1, \ldots , |\mathcal P| \} $, and $\mathcal D = \{ |\mathcal P|+1, \ldots , |\mathcal P|+|\mathcal D| \}$, and $\mathcal N = \{ |\mathcal P|+|\mathcal D|+1, \ldots , k \}$.
In addition, let $i=|\mathcal P|+ |\mathcal D|$, and $\pi_{\mathcal P \cup \mathcal D \setminus i}$ be the identity permutation.
Consequently,  by  (\ref{IDBthm}) in Theorem \ref{mainthmk} we obtain:
\begin{equation} \label{ourtechnique}
\sum_{i=1}^{|\mathcal P|+|\mathcal D|-1} \frac{d_i}{2^{i}} + d_{|\mathcal P|+|\mathcal D|}+\sum_{j\in \mathcal N} d_j  \leq 1, 
\end{equation}
or equivalently,
\begin{equation} \label{ourtechnique2}
(\sum_{i=1}^{|\mathcal P|} \frac{d_i}{2^{i}} )+(\sum_{i=|\mathcal P|+1}^{|\mathcal P|+|\mathcal D|-1 } \frac{d_i}{2^{i}}+  d_{|\mathcal P|+|\mathcal D|})+\sum_{j\in \mathcal N} d_j  \leq 1. 
\end{equation}
Note that in the above inequality there are $|\mathcal P|$ different coefficients (i.e. $\frac{1}{2}, \ldots , \frac{1}{2^{|\mathcal P|}}$) for receivers in $\mathcal P$, and $|\mathcal D|$ different coefficients (i.e. $\frac{1}{2^{|\mathcal P|+1}}, \ldots , \frac{1}{2^{|\mathcal P|+|\mathcal D|-1}}, 1$) for receivers in $\mathcal D$.
Due to symmetry, we can consider all the possible $|\mathcal P|!\times|\mathcal D|!$ joint permutations of the receivers in $\mathcal P$ and $\mathcal D$, leading to permutations of the corresponding coefficients in (\ref{ourtechnique2}).
By  summing over all those resulting inequalities,  and diving by $|\mathcal P|!\times|\mathcal D|!$, we obtain 
\begin{equation} \label{permuted}
 ( 1 - \frac{1}{2^{|\mathcal P|}} ) (\sum_{i=1}^{|\mathcal P|} \frac{d_i}{|\mathcal P|}) + ( 1+\frac{1}{2^{|\mathcal P|}}- \frac{1}{2^{|\mathcal P|+|\mathcal D|-1}} ) (\sum_{i=|\mathcal P|+1}^{|\mathcal P|+|\mathcal D|} \frac{d_i}{|\mathcal D|}) +\sum_{j\in \mathcal N} d_j \leq 1. 
\end{equation}
Note that $\Text{LDoF}_{\Text{sum}} \leq \max \sum_{i=1}^k d_i$ subject to (\ref{permuted}) and $d_i\leq 1$ for all $i$, which is basically a simple linear program.
By solving the linear program, one can easily see that 
 \begin{align}
\Text{LDoF}_{\Text{sum}} \leq  \max \left( |\mathcal P| + \frac{|\mathcal D|}{2^{|\mathcal P|}+1 - \frac{1}{2^{|\mathcal D|-1}}} , |\mathcal P| + \frac{1}{2^{|\mathcal P|}}  ,1 ,  \frac{|\mathcal D|}{1 + \frac{1}{2^{|\mathcal P|}} - \frac{1}{2^{|\mathcal P|+|\mathcal D|-1}}}   \right).
\end{align}

Note that since we assumed $|\mathcal P| \geq |\mathcal D| > 1 $,  the above inequality simplifies as follows:
 \begin{align}
\Text{LDoF}_{\Text{sum}} \leq  |\mathcal P| + \frac{|\mathcal D|}{2^{|\mathcal P|}+1 - \frac{1}{2^{|\mathcal D|-1}}},
\end{align}
which together with $\Text{LDoF}_{\Text{sum}} \geq |\mathcal P|$ leads to
\begin{equation}
|\mathcal P| \leq \Text{LDoF}_{\Text{sum}} \leq  |\mathcal P|+ \frac{|\mathcal D|}{2^{|\mathcal P|}+1 - \frac{1}{2^{|\mathcal D|-1}}}.
\end{equation}
Therefore, the gap between upper and lower bounds on $\Text{LDoF}_{\Text{sum}}$ is upper bounded as
\begin{align*}
\text{Gap} = \frac{|\mathcal D|}{2^{|\mathcal P|}+1 - \frac{1}{2^{|\mathcal D|-1}}} \stackrel{}{\leq} \frac{|\mathcal D|  }{2^{|\mathcal P|}} \stackrel{}{\leq} \frac{|\mathcal P|  }{2^{ |\mathcal P| }}   \leq \frac{1}{2}.
\end{align*}
Hence, the proof of Proposition \ref{propapprox} is complete.

\section{Proof of Proposition \ref{D1} ($\Text{LDoF}_{\Text{sum}} = |\mathcal P| + \frac{1}{2^{|\mathcal P|}}$ for $|\mathcal D|=1$)} \label{kAchiev}
We focus on the $k$-user MISO BC with only one receiver supplying delayed CSIT. 
We first prove the converse. 
Assume without loss of generality that  $\mathcal P=\{1,\ldots , |\mathcal P|\}$, $\mathcal D=\{|\mathcal P|+1\}$ and $\mathcal N=\{ |\mathcal P|+2, \ldots, k\}$.
Further, let $i=|\mathcal P|+1$, and $\pi_{\mathcal P \cup \mathcal D \setminus i}$ denote the identity permutation.
Then, by Theorem \ref{mainthmk} the solution to the following linear program provides an upper bound on $\Text{LDoF}_{\Text{sum}}$:
\begin{align}
\Text{LDoF}_{\Text{sum}} \leq \quad  \max\quad &\sum_{i=1}^{k} d_i \nonumber\\
 s.t. \quad &\sum_{i=1}^{|\mathcal P|} \frac{d_i}{2^{i}} + d_{|\mathcal P|+1}+\sum_{j\in \mathcal N} d_j  \leq 1,  \label{x117}\\
 \quad  &0\leq d_i \leq 1, \qquad i=1,\ldots, k,
\end{align}
where the first constraint in the linear program is due to  (\ref{IDBthm}) in Theorem \ref{mainthmk}.
Thus, by solving the above linear program one can readily see that
\begin{equation}
\Text{LDoF}_{\Text{sum}} \leq |\mathcal P| + \frac{1}{2^{|\mathcal P|}}.
\end{equation}
Hence, the converse proof is complete. 
We now present the achievable scheme, which is a multi-phase scheme that uses hybrid CSIT available to the transmitter to perform interference alignment.
The new achievable scheme  generalizes the schemes for $PD$ in \cite{RaviMaddah} and $PPD$ in \cite{3UserHybrid} (see Figure \ref{FigureScheme} for the special case of $PPPD$).

To achieve $\Text{LDoF}_{\Text{sum}}$ of $|\mathcal P| + \frac{1}{2^{|\mathcal P|}}$, we will ignore the receivers in $\mathcal N$; and we show that we can linearly achieve $(d_1, \ldots , d_{|\mathcal P|+1})= \left(1, \ldots, 1, \frac{1}{2^{|\mathcal P|}}\right)$.
Therefore, if, with slight abuse of notation, we denote $K\triangleq |\mathcal P|+1$,  we need to show that the following DoF tuple is linearly achievable:
\begin{align}
(d_1, \ldots, d_{K-1}, d_{K})&= \left(1, \ldots, 1, \frac{1}{2^{K-1}}\right).
\end{align}
To this end, we present a new multi-phase communication scheme which 
\begin{itemize}
\item operates over $2^{K-1}$ time slots;
\item delivers $2^{K-1}$ symbols to each of the receivers $1, \ldots, K-1$;
\item delivers $1$ symbol to receiver $K$.
\end{itemize}
The overall scheme is split into $K$ phases, indexed as $i=0, 1, 2, \ldots, (K-1)$:
\begin{itemize}
\item the duration of $i$-th phase is ${K-1 \choose i}$ time slots;
\item each of the first $(K-1)$ receivers obtain new (interference-free) linear equations in every time slot;
\item receiver $K$ obtains ${K-1 \choose i}$ equations during phase $i$ (one corresponding to each time  slot). 
\end{itemize}
At the end of the $i$-th phase, receiver $K$ does the following: it uses its received ${K-1 \choose i-1}$ equations from phase $(i-1)$ and ${K-1 \choose i}$ equations in phase $i$ to obtain ${K-1 \choose i}$ new equations with the following specific property: 
each equation is a linear combination of the desired symbol by $\text{Rx}_K$ and $(K-1-i)$ undesired symbols, where each undesired symbol is in fact desired by another receiver. 

Throughout the proof of the achievable scheme we only utilize the first $K$ transmit antennas; therefore, without loss of generality we can assume as well that there are only $K$ transmit antennas.
We first start with  Phase 0, and then explain the transmission strategy for an  arbitrary phase $i$  in full detail.

\subsection{Phase $0$}
Phase $0$ is of duration ${K-1 \choose 0}=1$, i.e., this phase only has $1$ time slot. In this phase, the transmitter sends 
$2$  information symbols  for each of $\text{Rx}_{1}, \text{Rx}_2, \ldots, \text{Rx}_{K-1}$, denoted  by $(\bold{s}^{1}_{1}, \bold{s}^{2}_{1})$, $(\bold{s}^{1}_{2}, \bold{s}^{2}_{2})$, $\ldots$ $,(\bold{s}^{1}_{K-1}, \bold{s}^{2}_{K-1})$, along with one symbol, denoted by $\bold{s}_K$, for the $K$-th receiver. 
%\begin{itemize}
%\item 
Let $\vec{\bold{g}}_{\mathcal S}(1)^\perp$, where $\mathcal S\subseteq \{1,\ldots, K-1\}$, denote a full row rank matrix of size $(K-|\mathcal S|) \times K$, where  each row of $\vec{\bold{g}}_{\mathcal S}(1)^\perp$ is perpendicular to any $\vec{\bold{g}}_{i}(1)$ where $i\in \mathcal S$.
We need to deliver one equation about $(\bold{s}^{1}_{i}, \bold{s}^{2}_{i})$ interference-free to $\text{Rx}_i$, for $i=1,\ldots, K-1$. 
To this aim, the transmit signal at time $1$ will be:
\begin{align}
\vec{\bold{x}}_1(1) = \sum_{i=1}^{K-1} [\vec{\bold{g}}_{\{1,\ldots , K-1\} \setminus \{i\} }(1)^\perp] ^\top \left[\begin{array}{c}\bold{s}^{1}_{i}\\
\bold{s}^{2}_{i}\end{array}\right]    +   [\vec{\bold{g}}_{\{1,\ldots , K-1\}  }(1)^\perp] ^\top  \bold{s}_{K}.
\end{align}

As a result, each of the first $K-1$ receivers obtain one equation in $2$ desired symbols:
\begin{align}
\bold{y}_{i}(1)&=  \vec{\bold{g}}_{i}(1)    [\vec{\bold{g}}_{\{1,\ldots , K-1\} \setminus \{i\} }(1)^\perp] ^\top  \left[\begin{array}{c}\bold{s}^{1}_{i}\\
\bold{s}^{2}_{i}\end{array}\right]   , \qquad i=1,\ldots , K-1;
\end{align}
%\item 
and receiver $K$ obtains $\bold{s}_K$ along with interference from the other symbols:
\begin{align}
\bold{y}_{K}(1)  = \sum_{i=1}^{K-1} \vec{\bold{g}}_{K}(1)      [\vec{\bold{g}}_{\{1,\ldots , K-1\} \setminus \{i\} }(1)^\perp] ^\top \left[\begin{array}{c}\bold{s}^{1}_{i}\\
\bold{s}^{2}_{i}\end{array}\right]    +\vec{\bold{g}}_{K}(1)   [\vec{\bold{g}}_{\{1,\ldots , K-1\}  }(1)^\perp] ^\top  \bold{s}_{K},
\end{align}
which  can be re-written as:
\begin{align}
\bold{y}_{K}(1)&=  L_{1}(\bold{s}^{1}_{1}, \bold{s}^{2}_{1}) + L_{2}(\bold{s}^{1}_{2}, \bold{s}^{2}_{2}) + \ldots + L_{K-1}(\bold{s}^{1}_{K-1}, \bold{s}^{2}_{K-1}) + \left( \vec{\bold{g}}_{K}(1)   [\vec{\bold{g}}_{\{1,\ldots , K-1\}  }(1)^\perp] ^\top \right)    \bold{s}_{K} ,
\end{align}
where  $ L_{i}(\bold{s}^{1}_{i}, \bold{s}^{2}_{i}) =        \vec{\bold{g}}_{K}(1)    [\vec{\bold{g}}_{\{1,\ldots , K-1\} \setminus \{i\} }(1)^\perp] ^\top \left[\begin{array}{c}\bold{s}^{1}_{i}\\
\bold{s}^{2}_{i}\end{array}\right]     $.
We observe that the $K$-th receiver has obtained $1$ equation, and this equation has $(K-1)$ interfering order-2 symbols, where each order-2 symbol is 
desirable by one of the other $(K-1)$ receivers. 
In particular, each order-2 symbol $L_{i}(\bold{s}^{1}_{i}, \bold{s}^{2}_{i})$
is desired by $\text{Rx}_i$.

The purpose of subsequent phases of the scheme is the following: in each phase $i$, we deliver the interference symbols of phase $i-1$ to the intended receivers while simultaneously sending new information symbols. 
 This should be done in an iterative manner to create a new set of equations at the $K$-th receiver with net interference from a smaller set of receivers, where the interference is useful for that set of receivers. 
With this broad goal in mind, we next describe the transmission strategy for the general phase $i$.

\subsection{Phase $i$}
Duration of Phase $i$ is ${K-1 \choose i}$ time slots. Let us index the slots as $j=1, 2, \ldots, {K-1 \choose i}$.

\subsubsection{Transmission in slot $j$, $j=1, 2, \ldots, {K-1 \choose i}$}
In each time slot $j$, the transmitter selects $i$ receivers out of first $(K-1)$ receivers. This splits the set of $(K-1)$ receivers into two disjoint sets, and for simplicity we denote these as:
\begin{itemize}
\item $\mathcal{R}$ (Repetition set): this is a set of $i$ receivers. Let us denote the indices of the receivers in this set by $(p_1, p_2, \ldots, p_{i})$.
\item $\mathcal{F}$ (Fresh set): this is the remaining set of $(K-1-i)$ receivers, and we denote this set of receivers as $(p_{i+1}, \ldots, p_{K-1})$
\end{itemize}
The basic idea behind the scheme can now be explained clearly:
\begin{itemize}
\item Note that in phase $(i-1)$, the $K$-th receiver has obtained ${K-1 \choose i-1}$ equations, where each equation is a linear combination of $(K-i)$ undesired symbols and the intended symbol (of the $K$-th receiver). 
\item Via delayed CSIT, the transmitter can reconstruct all of these equations within noise distortion. 
\item Out of these ${K-1 \choose i-1}$ equations, the transmitter focuses on those equations which consist of all symbols from the receivers $p_{i+1}, \ldots, p_{K-1}$ (i.e., the receivers belonging to the fresh set $\mathcal{F}$). In total, there are exactly ${i \choose 1} = i$ such equations. The reason is that each equation in phase $(i-1)$ has interference from exactly $(K-i)$ receivers. We zoom in on such equations with interference from $(K-1-i)$ receivers $p_{i+1}, \ldots, p_{K-1}$, and thus the remaining flexibility is to choose $1$ more interference symbol. The total remaining receivers to select from are $(K-1) -(K-1-i) = i$ and hence the number of ways is ${i \choose 1} = i$.

\item From each of these $i$ equations, the transmitter reconstructs the only  symbol in the equation which is desired by one of the receivers in the repetition set $(p_1, p_2, \ldots, p_i)$. Let us denote the reconstructed symbols by $\bold{s}_{p_1}(j), \ldots, \bold{s}_{p_{i}}(j)$ . 
Also, we denote those $i$ equations as following:
\begin{align}
& \bold{s}_{p_{1}}(j) +  LC_{1}: \text{ where $ LC_{1}$ is a linear combination of symbols for receivers in set } \mathcal{F} \cup \{K\}  \label{E1}\\
& \bold{s}_{p_{2}}(j) +  LC_{2}: \text{ where $ LC_{2}$ is a linear combination of symbols for receivers in set } \mathcal{F} \cup \{K\}\label{E2}\\  
&\vdots\\
& \bold{s}_{p_{i}}(j) +  LC_{i}: \text{ where $ LC_{i}$ is a linear combination of symbols for receivers in set } \mathcal{F} \cup \{K\}. \label{Ei}
\end{align}

\item For each of the $(K-1-i)$ receivers in the fresh set, the transmitter sends $2$ precoded fresh (i.e. new) symbols. Let us denote these as $(\bold{s}^{1}_{p_{i+1}}(j), \bold{s}^{2}_{p_{i+1}}(j))$, $(\bold{s}^{1}_{p_{i+2}}(j), \bold{s}^{2}_{p_{i+2}}(j))$ up to $(\bold{s}^{1}_{p_{K-1}}(j), \bold{s}^{2}_{p_{K-1}}(j))$.  
\end{itemize} 
Hence in the $j$-th slot of phase $i$, the transmitter sends:
\begin{align}
\vec{\bold{x}}_{i}(j)&= \underbrace{\sum_{r=1}^{i}   [\vec{\bold{g}}_{\{1,\ldots , K-1\} \setminus \{p_{r}(j)\} }(j)^\perp] ^\top  \left[\begin{array}{c}\bold{s}_{p_{r}}(j)\\
0\end{array}\right]}_{i \text{ repetition symbols}} + \underbrace{\sum_{r=i+1}^{K-1}   [\vec{\bold{g}}_{\{1,\ldots , K-1\} \setminus \{p_{r}(j)\} }(j)^\perp] ^\top  \left[\begin{array}{c}\bold{s}^{1}_{p_{r}}(j)\\
\bold{s}^{2}_{p_{r}}(j)\end{array}\right]}_{2(K-1-i) \text{ fresh symbols}}.
\end{align}
Clearly, the set of repetition receivers $\{p_{1}, p_2, \ldots, p_{i}\}$ receive one symbol without interference. Similarly, the set of receivers $\{p_{i+1}, \ldots, p_{K-1}\}$ also receive one clean (interference-free) useful symbol in this slot (which is a linear combination of the two fresh symbols).

\subsubsection{Operation at  $\text{Rx}_K$ in time slot $j$ of phase $i$}
Let us now focus on $\text{Rx}_K$ at the $j$-th time slot of phase $i$. $\text{Rx}_K$ obtains 
\begin{align}
\bold{y}_{K}(i, j)&= \sum_{r=1}^{i}\alpha_{r}(i, j)\bold{s}_{p_{r}}(j) + \sum_{r=i+1}^{K-1} LC(\bold{s}^{1}_{p_{r}}(j), \bold{s}^{2}_{p_{r}}(j)),\label{EA}
\end{align}
where 
$\alpha_{r}(i, j)$ denotes the coefficient of the symbol $\bold{s}_{p_{r}}(j)$ when received at $\text{Rx}_K$; and
$ LC(\bold{s}^{1}_{p_{r}}(j), \bold{s}^{1}_{p_{r}}(j))$ denotes the linear combination of $\bold{s}^{1}_{p_{r}}(j), \bold{s}^{1}_{p_{r}}(j)$ received at $\text{Rx}_K$.
Note that from phase $(i-1)$, the receiver also has $i$ equations $\bold{s}_{p_{1}}(j) +  LC_{1},\ldots, \bold{s}_{p_{i}}(j) +  LC_{i}$ as mentioned in (\ref{E1})-(\ref{Ei}).
Using these $i$ equations together with  (\ref{EA}), receiver $K$ eliminates the $i$ symbols $\bold{s}_{p_{1}}(j), \bold{s}_{p_{2}}(j), \ldots, \bold{s}_{p_{i}}(j)$; and it is left with an equation of the following form:
\begin{align}
LC_{p_{i+1}}(j) + LC_{p_{i+2}}(j) + \ldots + LC_{p_{K-1}}(j) + \bold{s}_{K}.
\end{align}
%where $LC_{p_{r}}(j)$ denotes a linear combination of symbols intended for $\text{Rx}_r$.
This equation consists of $(K-1-i)$ interfering symbols, where each interfering symbol $LC_{p_{r}}(j)$ is desired by  $\text{Rx}_{p_{r}}$, for $r=i+1, i+2, \ldots, (K-1)$.

Recall that  the slot index  $j$ varies from $1$ to ${K-1 \choose i}$, each slot corresponding to the partitioning of the set of $K-1$ receivers into two disjoint sets of size $i$ and $(K-1-i)$. 
Each slot gives the $K$-th receiver one equation with interference from exactly $(K-1-i)$ receivers. Hence, in total, at the end of phase $i$, the receiver has ${K-1 \choose i}$ equations, and each equation has interference from symbols desired by exactly $(K-1-i)$ receivers. Thus, we can now readily apply this process iteratively.

\subsubsection{Phase $K-1$ (the last phase; corresponding to $i=K-1$)}
Before the last phase $K-1$, (i.e., just after phase $K-2$), the $K$-th receiver has ${K-1 \choose i-1} = {K-1 \choose K-2} = K-1$ equations, and each equation has interference from exactly $(K-1)- (i-1) = (K-1)- (K-2) = 1$ receiver. 
Hence, the $K$-th receiver has $K-1$  equations of the following form before the last phase: 
\begin{align}
 LC'_{1} + \bold{s}_K ,LC'_{2} + \bold{s}_K, \ldots, LC'_{K-1} + \bold{s}_K,\label{RK}
\end{align}
where $LC'_{1}$ is desired by receiver $1$, $LC'_{2}$ is desired by receiver $2$, etc.

In the last phase, whose duration is only $1$ slot (since ${K-1 \choose K-1} =1)$, the transmitter sends $LC'_{1}, \ldots, LC'_{K-1}$ without any interference to receivers $1, \ldots, K-1$ by utilizing instantaneous  CSIT. Receiver $K$ obtains a linear combination of $LC'_{1}, \ldots LC'_{K-1}$.
Hence, the $K$th receiver has $K$ equations  in $K$ variables $LC'_1, LC'_2, \ldots, LC'_K$ and $s_K$. Therefore, it can decode $s_K$;  
and the proof is complete.

\begin{figure}[t]
  \centering
\includegraphics[width=14.0cm]{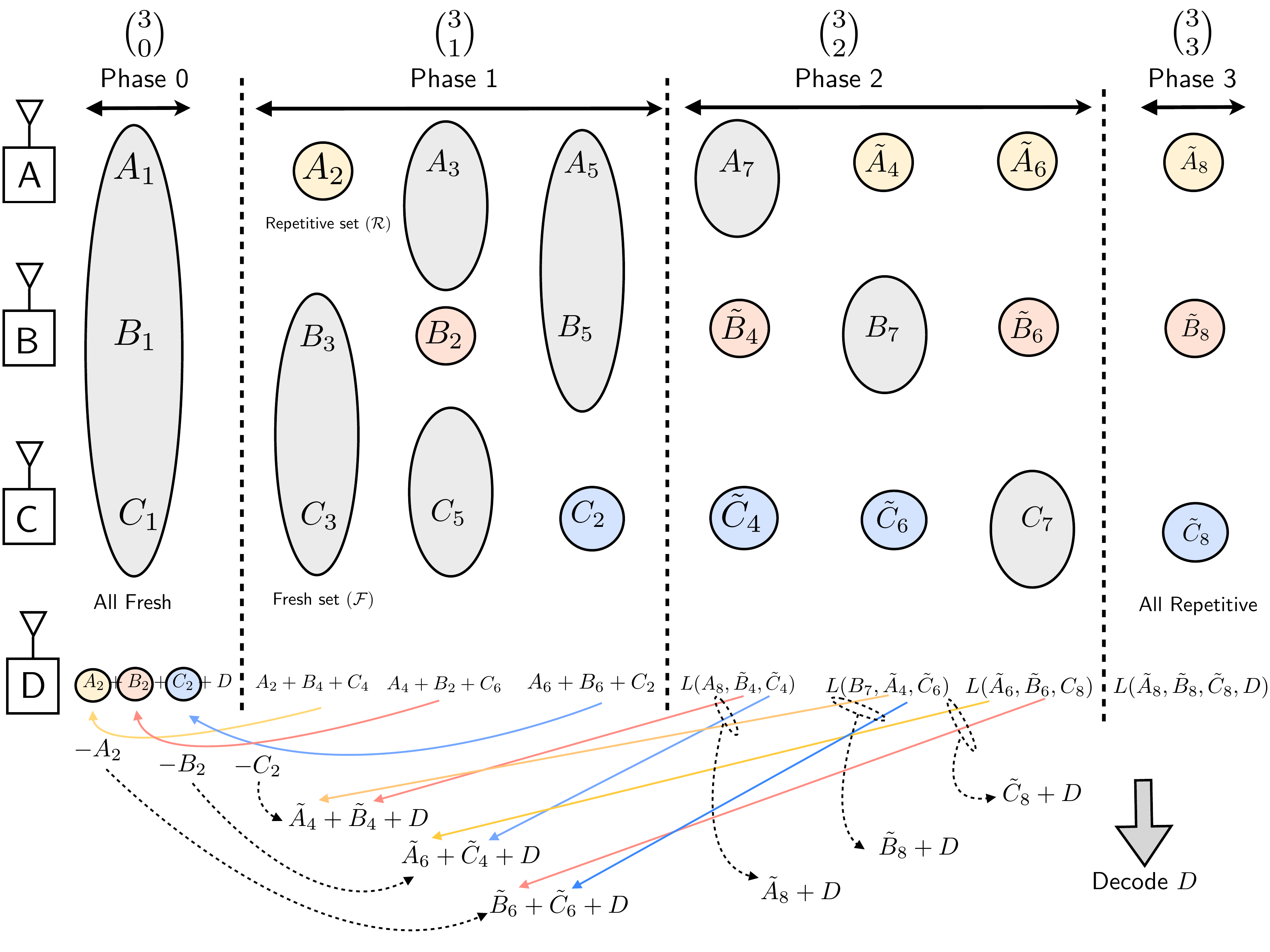}
\vspace{-0.2in}
\caption{Scheme for $4$-user MISO BC: $PPPD$ Setting.}\label{FigureScheme}
\end{figure}

\subsection{Illustrative Example -- $4$ User MISO BC}
Here, we present the achievable scheme for $K=4$  to clearly illustrate the idea behind the iterative scheme. For the case of 4-user MISO BC with $PPPD$, the goal is to  achieve:
\begin{align}
(d_1, d_2, d_3, d_4)&= \left(1, 1, 1, \frac{1}{2^3}\right).
\end{align}
Here, the scheme has $K= 4$ phases, with the following phase durations:
\begin{itemize}
\item Phase $0$: ${3 \choose 0} =1$ time slots; Tx sends two new information symbols for each of  the first three receivers, and one symbol for the fourth receiver. Each of the first three receivers will receive a linear combination of its two desired symbols without any interference. 
\item Phase $1$: ${3 \choose 1} =2$ time slots; in each time slot, Tx sends 
the signal received in the past by $\text{Rx}_K$ with respect to the symbols of one of the first three receivers; and it also sends
two new  information symbols for each of the other 2 receivers supplying instantaneous CSIT.
\item Phase $2$: ${3 \choose 2} =3$ time slots; in each slot, Tx sends fresh information for $1$ receiver with instantaneous CSIT and supplies past signals received by $\text{Rx}_K$ with respect to the remaining $2$ receivers supplying instantaneous CSIT.
\item Phase $3$: ${3 \choose 3} =1$ time slot; Tx sends past received signals by $\text{Rx}_K$ which are  desired by the three receivers supplying instantaneous CSIT. 
\end{itemize}
See Figure \ref{FigureScheme} which illustrates the achievable scheme for $4$-user MISO BC, where the first 3 receivers supply instantaneous CSIT, while the fourth receiver supplies delayed CSIT.

\section{Proof of Claim \ref{cl12}} \label{cl12proof}
We first re-state Claim \ref{cl12} here for convenience.

\begin{claimrep}{Claim~\ref{cl12}} 
\begin{equation} 
 \sum_{j=1}^{|\mathcal P|+ |\mathcal D|-1} \frac{m_j(n)}{2^{j}}  \stackrel{a.s.}{\leq } \Text{rank}[ \bold{G}_{|\mathcal P|+ |\mathcal D|}^n [\bold{V}_{1}^{n}\ldots \bold{V}_{|\mathcal P|+ |\mathcal D|-1}^{n}]].
\end{equation}
\end{claimrep}

To prove Claim \ref{cl12} we first prove the following inequality by induction, and then show how it leads to proving Claim \ref{cl12}.
\begin{align} 
&  \sum_{j=1}^{ i-1 } \frac{m_j(n)}{2^{j}} + \frac{\Text{rank}[ \bold{G}_{|\mathcal P|+ |\mathcal D|}^n [\bold{V}_{i}^{n}\ldots \bold{V}_{|\mathcal P|+ |\mathcal D|-1}^{n}]] }{2^{i-1}} \stackrel{a.s.}{\leq } \Text{rank}[ \bold{G}_{|\mathcal P|+ |\mathcal D|}^n [\bold{V}_{1}^{n}\ldots \bold{V}_{|\mathcal P|+ |\mathcal D|-1}^{n}]], \quad i=2,\ldots, |\mathcal P|+ |\mathcal D|-1.\label{cl12proofproxy}
\end{align}

We prove (\ref{cl12proofproxy}) by induction on $i$. 
For the base case of $i=2$, the inequality in (\ref{cl12proofproxy}) simplifies to
\begin{align} 
&   \frac{m_1(n)}{2} + \frac{\Text{rank}[ \bold{G}_{|\mathcal P|+ |\mathcal D|}^n [\bold{V}_{2}^{n}\ldots \bold{V}_{|\mathcal P|+ |\mathcal D|-1}^{n}]] }{2} \stackrel{a.s.}{\leq } \Text{rank}[ \bold{G}_{|\mathcal P|+ |\mathcal D|}^n [\bold{V}_{1}^{n}\ldots \bold{V}_{|\mathcal P|+ |\mathcal D|-1}^{n}]].
\end{align}
Hence, the base case of $i=2$
 holds due to Lemma \ref{MainIneqk} and (\ref{decode}).
Suppose that the induction hypothesis is true for $i=s$. We show that it will also hold for $i=s+1$.
By our assumption we have  
\begin{eqnarray} 
 && \Text{rank}[ \bold{G}_{|\mathcal P|+ |\mathcal D|}^n [\bold{V}_{1}^{n}\ldots \bold{V}_{|\mathcal P|+ |\mathcal D|-1}^{n}]]  \stackrel{a.s.}{\geq }  \sum_{j=1}^{ s-1 } \frac{m_j(n)}{2^{j}} + \frac{\Text{rank}[ \bold{G}_{|\mathcal P|+ |\mathcal D|}^n [\bold{V}_{s}^{n}\ldots \bold{V}_{|\mathcal P|+ |\mathcal D|-1}^{n}]] }{2^{s-1}}\nonumber \\
& \substack{(\text{Lemma \ref{MainIneqk}})\\ a.s. \\ \geq } &  \sum_{j=1}^{ s-1 } \frac{m_j(n)}{2^{j}} + \frac{ \frac{  \Text{rank}[ \bold{G}_{s}^n [\bold{V}_{s}^{n}\ldots \bold{V}_{|\mathcal P|+ |\mathcal D|-1}^{n}]] - \Text{rank}[ \bold{G}_{s}^n [\bold{V}_{s+1}^{n}\ldots \bold{V}_{|\mathcal P|+ |\mathcal D|-1}^{n}]]      +   \Text{rank}[ \bold{G}_{|\mathcal P|+ |\mathcal D|}^n [\bold{V}_{s+1}^{n}\ldots \bold{V}_{|\mathcal P|+ |\mathcal D|-1}^{n}]] }{2}}{2^{s-1}}\nonumber \\
& \substack{\text{(Lemma \ref{ranksubmod})} \\ a.s. \\ \geq } &  \sum_{j=1}^{ s-1 } \frac{m_j(n)}{2^{j}}  + \frac{  \Text{rank}[ \bold{G}_{s}^n [\bold{V}_{1}^{n}\ldots \bold{V}_{k}^{n}]] - \Text{rank}[ \bold{G}_{s}^n [\cup_{\substack{i\in \{1,\ldots, k\} \\ i\ne s}}\bold{V}_{i}^{n}]]      +   \Text{rank}[ \bold{G}_{|\mathcal P|+ |\mathcal D|}^n [\bold{V}_{s+1}^{n}\ldots \bold{V}_{|\mathcal P|+ |\mathcal D|-1}^{n}]] }{2^s}\nonumber \\
& \substack{(\ref{decode})\\ a.s. \\ =} &  \sum_{j=1}^{ s-1 } \frac{m_j(n)}{2^{j}}  + \frac{\Text{rank}[ \bold{G}_{s}^n \bold{V}_{s}^{n}]    +   \Text{rank}[ \bold{G}_{|\mathcal P|+ |\mathcal D|}^n [\bold{V}_{s+1}^{n}\ldots \bold{V}_{|\mathcal P|+ |\mathcal D|-1}^{n}]] }{2^s}\nonumber \\
& \substack{(\ref{decode})\\ a.s. \\ =}  & \sum_{j=1}^{ s} \frac{m_j(n)}{2^{j}} + \frac{      \Text{rank}[ \bold{G}_{|\mathcal P|+ |\mathcal D|}^n [\bold{V}_{s+1}^{n}\ldots \bold{V}_{|\mathcal P|+ |\mathcal D|-1}^{n}]] }{2^{s}}.\nonumber
\end{eqnarray}
Hence, the induction hypothesis holds for $i=s+1$ as well; and as a result, the proof of (\ref{cl12proofproxy}) is complete.
We now show how (\ref{cl12proofproxy}) leads to proof of Claim \ref{cl12}.
Let $ i = |\mathcal P|+ |\mathcal D|  -  1$. Then, by (\ref{cl12proofproxy}),
\begin{eqnarray}
 \Text{rank}[ \bold{G}_{|\mathcal P|+ |\mathcal D|}^n [\bold{V}_{1}^{n}\ldots \bold{V}_{|\mathcal P|+ |\mathcal D|-1}^{n}]]   &\stackrel{a.s.}{\geq }& \sum_{j=1}^{ |\mathcal P|+ |\mathcal D|  - 2 } \frac{m_j(n)}{2^{j}} + \frac{\Text{rank}[ \bold{G}_{|\mathcal P|+ |\mathcal D|}^n  \bold{V}_{|\mathcal P|+ |\mathcal D|-1}^{n}] }{2^{|\mathcal P|+ |\mathcal D|-2}}\nonumber \\
& \substack{\text{(Lemmaa \ref{MIMORRIk})} \\  a.s.\\ \geq }& \sum_{j=1}^{ |\mathcal P|+ |\mathcal D|  - 2 } \frac{m_j(n)}{2^{j}} + \frac{\Text{rank}[[ \bold{G}_{|\mathcal P|+ |\mathcal D|-1}^n ; \bold{G}_{|\mathcal P|+ |\mathcal D|}^n  ]  \bold{V}_{|\mathcal P|+ |\mathcal D|-1}^{n}] }{2^{|\mathcal P|+ |\mathcal D|-1}}\nonumber \\
& \substack{ \geq }& \sum_{j=1}^{ |\mathcal P|+ |\mathcal D|  - 2 } \frac{m_j(n)}{2^{j}} + \frac{\Text{rank}[ \bold{G}_{|\mathcal P|+ |\mathcal D|-1}^n   \bold{V}_{|\mathcal P|+ |\mathcal D|-1}^{n}] }{2^{|\mathcal P|+ |\mathcal D|-1}}\nonumber \\
& \substack{(\ref{decode}) \\  a.s.\\ = }& \sum_{j=1}^{ |\mathcal P|+ |\mathcal D|  - 1 } \frac{m_j(n)}{2^{j}} ,\nonumber 
\end{eqnarray}
which completes the proof of Claim \ref{cl12}.

\section{Proof of Claim \ref{MIMOvariant}} \label{MIMOvariantApp}
We first re-state the Claim for convenience.

\begin{claimrep}{Claim~\ref{MIMOvariant}}
\begin{equation}
   \frac{ \Text{rank}[[\bold{G}_1^n; \ldots ; \bold{G}_k^n ] [\bold{V}_1^n \ldots  \bold{V}_{|\mathcal P|}^n ]]  }{k}\stackrel{a.s.}{\leq} \frac{\Text{rank}[[\bold{G}_{|\mathcal P| + 1}^n; \ldots ; \bold{G}_{|\mathcal P|+|\mathcal D|}^n ] [\bold{V}_{1}^n \ldots  \bold{V}_{|\mathcal P| }^n ] ]}{|\mathcal D|}.
\end{equation}
\end{claimrep}

\begin{proof}

%The proof is similar to the proof of Lemma \ref{MIMORRIk}, provided in Appendix \ref{AppMIMORRI}; but we present it for the sake of completeness. 
We consider the notations (\ref{conditiondef}), (\ref{conditioning}); and we  use
$\bold{Y}_j^n \triangleq  \bold{G}_j^n [\bold{V}_1^n \ldots  \bold{V}_{|\mathcal P|}^n ]$, and $\bold{Y}_j(t) \triangleq  \vec {\bold{g}}_j(t) [\bold{V}_1(t) \ldots  \bold{V}_{|\mathcal P|}(t) ]$.
Furthermore, we denote by $\bold{Y}_{\mathcal S}^n$ the column concatenation of matrices $\bold{G}_j^n [\bold{V}_1^n \ldots  \bold{V}_{|\mathcal P|}^n ]$, where $j\in \mathcal S$.
Therefore, we need to show that 
\begin{equation} \label{MIMOvarObj}
   \frac{ \text{rank}[\bold{Y}_1^n; \ldots ; \bold{Y}_k^n ]  }{k}\stackrel{a.s.}{\leq} \frac{\text{rank}[\bold{Y}_{\mathcal D}^n ]}{|\mathcal D|}.
\end{equation}
For all $ t=1,\ldots, n,$  we have
\begin{align}
( |\mathcal P|+ |\mathcal N| ) \times & \text{rank}[\bold{Y}_{\mathcal D}(t)   |   \bold{Y}_{\mathcal D}^{t-1}  ]  \stackrel{(\ref{conditiondef})}{=}\sum_{i=|\mathcal P|+1}^{|\mathcal P| +|\mathcal D|} ( |\mathcal P|+ |\mathcal N| ) \times \text{rank}[\bold{Y}_{i}(t)   |  [ \bold{Y}_{\mathcal D}^{t-1} ; \bold{Y}_{|\mathcal P|+1}(t) ;\ldots ; \bold{Y}_{i-1}(t) ]] \nonumber\\
& \substack{(a)\\ a.s.\\ \geq}   \sum_{i=|\mathcal P|+1}^{|\mathcal P| +|\mathcal D|}  \sum_{j \in \mathcal P \cup \mathcal N}    \text{rank}[\bold{Y}_{j}(t)  | [ \bold{Y}_{\mathcal D}^{t-1} ; \bold{Y}_{|\mathcal P|+1}(t) ;\ldots ; \bold{Y}_{i-1}(t) ]] \nonumber\\
& \substack{ (\ref{conditioning}) \\ \geq}   \sum_{i=|\mathcal P|+1}^{|\mathcal P| +|\mathcal D|}  \sum_{j \in \mathcal P \cup \mathcal N}    \text{rank}[\bold{Y}_{j}(t)   | [ \bold{Y}_{\mathcal D}^{n}; \bold{Y}_{\mathcal P \cup \mathcal N}^{t-1} ]]  \quad \substack{ =} \quad   |\mathcal D|  \sum_{j \in \mathcal P \cup \mathcal N}    \text{rank}[\bold{Y}_{j}(t)   |  [ \bold{Y}_{\mathcal D}^{n}; \bold{Y}_{\mathcal P \cup \mathcal N}^{t-1}] ] \nonumber\\
& \substack{(\ref{condSubAdd})\\  \geq }   |\mathcal D|   \times  \text{rank}[\bold{Y}_{\mathcal P \cup \mathcal N}(t)   | [  \bold{Y}_{\mathcal D}^{n}; \bold{Y}_{\mathcal P \cup \mathcal N}^{t-1}] ], \label{MIMOvarder}
\end{align}
where
(a) follows from the same arguments as in (\ref{a1})-(\ref{a5}) which were used to show that 
$$\text{rank} [\bold{Y}_i(t)| [\bold{Y}^{t-1};  \bold{Y}_1(t); \ldots ; \bold{Y}_{i-1}(t) ]]  \stackrel{a.s.}{\geq} 
\text{rank} [\bold{Y}_{j+1}(t)| [\bold{Y}^{t-1};  \bold{Y}_1(t); \ldots ; \bold{Y}_{i-1}(t)] ],$$
for the case where $i\in \{1,\ldots , j\} \subseteq \mathcal D$,   and $\bold{Y}^{t-1}\triangleq [\bold{Y}_1^{t-1}; \ldots ;\bold{Y}_{j}^{t-1}] $.

By summing both sides of the inequality (\ref{MIMOvarder}) over all $t=1,\ldots, n$, we obtain 
\begin{align}
( |\mathcal P|+ |\mathcal N| ) \times \text{rank}[  \bold{Y}_{\mathcal D}^{n}  ]  & \stackrel{a.s.}{\geq } |\mathcal D|   \times  \text{rank}[\bold{Y}_{\mathcal P \cup \mathcal N}^n   |   \bold{Y}_{\mathcal D}^{n} ] \nonumber\\
&\stackrel{(\ref{conditiondef})}{=} |\mathcal D|   \times  \text{rank}[\bold{Y}_{1}^n;\ldots ;  \bold{Y}_{k}^{n} ]
-   |\mathcal D|   \times  \text{rank}[ \bold{Y}_{\mathcal D}^{n} ] . 
\end{align}
Finally, by rearranging the above inequality we obtain
(\ref{MIMOvarObj}), which proves Claim \ref{MIMOvariant}. 
\end{proof}

\end{appendices}

\bibliographystyle{ieeetr}
\bibliography{HybridCSIT2}

\begin{thebibliography}{10}

\bibitem{OursICC2015}
S.~Lashgari, R.~Tandon, and S.~Avestimehr, ``{Three-user MISO broadcast
  channel: how much can CSIT heterogeneity help?},'' {\em to be presented at
  IEEE ICC 2015}, 2015.

\bibitem{OursISIT2015}
S.~Lashgari, R.~Tandon, and S.~Avestimehr, ``{A general outer bound for MISO
  broadcast channel with heterogeneous CSIT},'' {\em to be presented at IEEE
  ISIT 2015}, 2015.

\bibitem{MAT}
M.~A. Maddah-Ali and D.~N. Tse, ``{Completely stale transmitter channel state
  information is still very useful},'' {\em IEEE Transactions on Information
  Theory}, vol.~58, no.~7, pp.~4418--4431, 2012.

\bibitem{abdoli}
M.~Abdoli, A.~Ghasemi, and A.~Khandani, ``{On the degrees of freedom of
  $K$-user SISO interference and $X$ channels with delayed CSIT},'' {\em
  arXiv:1109.4314}, 2011.

\bibitem{varanasi}
C.~S. Vaze and M.~K. Varanasi, ``{The degrees of freedom region and intererence
  alignment for the MIMO inteference channel with delayed CSIT},'' {\em IEEE
  Transactions on Information Theory}, vol.~58, no.~7, pp.~4396--4417, 2012.

\bibitem{VMABinaryIC}
A.~Vahid, M.~Maddah-Ali, and A.~S. Avestimehr, ``{Capacity results for binary
  fading interference channels with delayed CSIT},'' {\em IEEE Transactions on
  Information Theory}, vol.~60, no.~10, pp.~6093--6130, 2014.

\bibitem{OursAllerton2014}
S.~Lashgari and A.~S. Avestimehr, ``{Transmitter cooperation in interference
  channel with delayed CSIT},'' {\em Proc. of Allerton Conference on
  Communicaiton, Control, and Computing}, 2014.

\bibitem{xchannel}
A.~Ghasemi, A.~Motahari, and A.~Khandani, ``{On the degrees of freedom of $X$
  channel with delayed CSIT},'' {\em Proc. of IEEE ISIT}, 2011.

\bibitem{Ours}
S.~Lashgari, A.~S. Avestimehr, and C.~Suh, ``{Linear degrees of freedom of the
  X-channel with delayed CSIT},'' {\em IEEE Transactions on Information
  Theory}, vol.~60, no.~4, pp.~2180--2189, 2014.

\bibitem{AAMultiHopJ}
J.~Abdoli and A.~S. Avestimehr, ``{Layered interference networks with delayed
  CSI: DoF scaling with distributed transmitters},'' {\em IEEE Transactions on
  Information Theory}, vol.~60, no.~3, pp.~1822--1839, 2014.

\bibitem{Yang}
S.~Yang, M.~Kobayashi, D.~Gesbert, and X.~Yi, ``{Degrees of freedom of time
  correlated MISO broadcast channel with delayed CSIT},'' {\em IEEE
  Transactions on Information Theory}, vol.~59, no.~1, pp.~315--328, 2012.

\bibitem{retrospective}
H.~Maleki, S.~A. Jafar, and S.~Shamai, ``{Retrospective interference alignment
  over interference networks},'' {\em IEEE Journal of Selected Topics in Signal
  Processing}, vol.~6, no.~3, pp.~228--240, 2012.

\bibitem{RaviMaddah}
R.~Tandon, M.~Maddah-Ali, A.~Tulino, H.~V. Poor, and S.~Shamai, ``{On fading
  broadcast channels with partial channel state information at the
  transmitter},'' {\em Proc. of IEEE ISWCS}, 2012.

\bibitem{MIMOICHybrid}
K.~Mohanty, C.~S. Vaze, and M.~K. Varanasi, ``{The degrees of freedom region
  for the MIMO interference channel with hybrid CSIT},'' {\em arXiv preprint},
  2012.

\bibitem{synergistics}
R.~Tandon, S.~A. Jafar, S.~Shamai, and H.~V. Poor, ``{On the synergistic
  benefits of alternating CSIT for the MISO broadcast channel},'' {\em IEEE
  Transactions on Information Theory}, vol.~59, no.~7, pp.~4106--4128, 2013.

\bibitem{MTU}
P.~Mukherjee, R.~Tandon, and S.~Ulukus, ``{Secure degrees of freedom region of
  the two-user MISO broadcast channel with alternating CSIT},'' {\em arXiv
  preprint arXiv:1502.02647}, 2015.

\bibitem{Mohanty}
K.~Mohanty and M.~K. Varanasi, ``{On the DoF region of the K-user MISO
  broadcast channel with hybrid CSIT},'' {\em arXiv preprint arXiv:1312.1309},
  2013.

\bibitem{3UserHybrid}
S.~Amuru, R.~Tandon, and S.~Shamai, ``{On the degrees-of-freedom of the 3-user
  MISO broadcast channel with hybrid CSIT},'' {\em arXiv preprint
  arXiv:1402.4729}, 2014.

\bibitem{GholamiJafar}
A.~G. Davoodi and S.~A. Jafar, ``{Aligned Image Sets under Channel Uncertainty:
  Settling a Conjecture by Lapidoth, Shamai and Wigger on the Collapse of
  Degrees of Freedom under Finite Precision CSIT},'' {\em arXiv preprint
  arXiv:1403.1541}, 2014.

\bibitem{OursAllerton}
S.~Lashgari, A.~S. Avestimehr, and C.~Suh, ``{A rank ratio inequality and the
  linear degrees of freedom of X-channel with delayed CSIT},'' {\em Proc. of
  Allerton Conference on Communicaiton, Control, and Computing}, 2013.

\bibitem{RassouliClerckx}
B.~Rassouli, C.~Hao, and B.~Clerckx, ``{DoF analysis of the K-user MISO
  broadcast channel with hybrid CSIT},'' {\em arXiv preprint}, 2015.

\bibitem{OursISIT2014}
S.~Lashgari and A.~S. Avestimehr, ``{Blind wiretap channel with delayed
  CSIT},'' {\em Proc. of IEEE ISIT}, 2014.

\bibitem{OursGlobecom2014}
S.~Lashgari and A.~S. Avestimehr, ``{Blind MIMO wiretap channel with delayed
  CSIT},'' {\em IEEE Globecom, Second workshop on trusted communications with
  Physical Layer Security}, 2014.

\bibitem{Lovasz}
L.~Lovasz, ``{Submodular functions and convexity},'' {\em Springer}, 1983.

\bibitem{Measure}
H.~Federer, ``{Geometric measure theory},'' {\em Reprint of the 1969 Edition.
  Springer}, 1996.

\end{thebibliography}

\end{document}